\documentclass[sigconf]{acmart}

\makeatletter
\def\@ACM@checkaffil{
	\if@ACM@instpresent\else
	\ClassWarningNoLine{\@classname}{No institution present for an affiliation}%
	\fi
	\if@ACM@citypresent\else
	\ClassWarningNoLine{\@classname}{No city present for an affiliation}%
	\fi
	\if@ACM@countrypresent\else
	\ClassWarningNoLine{\@classname}{No country present for an affiliation}%
	\fi
}
\makeatother

\renewcommand\footnotetextcopyrightpermission[1]{} 

\usepackage{stmaryrd}
\usepackage{balance}
\usepackage{tikz}
\usetikzlibrary{shapes}
\usetikzlibrary{calc,arrows,positioning}
\usepackage{algorithm,tabularx}
\usepackage{algpseudocode}
\usepackage{amsmath}
\usepackage{booktabs}
\usepackage{amsfonts}
\usepackage[english]{babel}
\usepackage{amsthm}
\usepackage{tablefootnote}

\usepackage{amssymb}
\usepackage{multirow}
\usepackage{hyperref}
\usepackage{graphicx}
\usepackage{subfigure}

\usepackage{pifont}
\usepackage{xcolor}

\usepackage{mathabx}
\usepackage{makecell}
\usepackage{multicol}
\makeatletter
\newcommand{\multiline}[1]{%
	\begin{tabularx}{\dimexpr\linewidth-\ALG@thistlm}[t]{@{}X@{}}
		#1
	\end{tabularx}
}
\makeatother

\makeatletter
\newsavebox{\@brx}
\newcommand{\llangle}[1][]{\savebox{\@brx}{\(\m@th{#1\langle}\)}%
	\mathopen{\copy\@brx\kern-0.5\wd\@brx\usebox{\@brx}}}
\newcommand{\rrangle}[1][]{\savebox{\@brx}{\(\m@th{#1\rangle}\)}%
	\mathclose{\copy\@brx\kern-0.5\wd\@brx\usebox{\@brx}}}
\makeatother

\algrenewcommand\textproc{}

\newcommand{\main}{\textsf{Clover}}
\newcommand{\baseline}{\textsf{ORAM (SH-Sec.)}}
\UseRawInputEncoding

\newcommand{\revise}{\textcolor{black}}

\theoremstyle{plain}
\newtheorem{thm}{Theorem}
\newtheorem{lem}[thm]{Lemma}

\theoremstyle{definition}
\newtheorem{defn}{Definition}

\newtheorem{assump}{Assumption}

\copyrightyear{2025}
\acmYear{2025}
\setcopyright{cc}
\setcctype{by-nc}
\acmConference[CCS '25]{Proceedings of the 2025 ACM SIGSAC Conference on
	Computer and Communications Security}{October 13--17, 2025}{Taipei}
\acmBooktitle{Proceedings of the 2025 ACM SIGSAC Conference on Computer
	and Communications Security (CCS '25), October 13--17, 2025, Taipei}\acmDOI{10.1145/3719027.3765044}
\acmISBN{979-8-4007-1525-9/2025/10}

\begin{document}
	
	\title{Harnessing Sparsification in Federated Learning: A Secure, Efficient, and Differentially Private Realization}
	
	\titlenote{This is a full version of the paper originally published in ACM CCS 2025 \cite{clover-submission}.}
	
	\author{Shuangqing Xu}
	\affiliation{
		\institution{Harbin Institute of Technology, Shenzhen}
		\city{Shenzhen}
		\country{China}}
	\affiliation{
		\institution{The Hong Kong Polytechnic University}
		\city{Hong Kong}
		\country{China}}
	
	\author{Yifeng Zheng}
	\authornote{Corresponding authors.}
	\affiliation{
		\institution{The Hong Kong Polytechnic University}
		\city{Hong Kong}
		\country{China}}
	
	\author{Zhongyun Hua}
	\authornotemark[2]
	\affiliation{
		\institution{Harbin Institute of Technology, Shenzhen}
		\city{Shenzhen}
		\country{China}}
	
	\renewcommand{\shortauthors}{Shuangqing Xu, Yifeng Zheng, \& Zhongyun Hua}
	\begin{abstract}

		Federated learning (FL) enables multiple clients to jointly train a model by sharing only gradient updates for aggregation instead of raw data.
		Due to the transmission of very high-dimensional gradient updates from many clients, FL is known to suffer from a communication bottleneck. 
		Meanwhile, the gradients shared by clients as well as the trained model may also be exploited for inferring private local datasets, making privacy still a critical concern in FL.
		We present \textsf{Clover}, a novel system framework for communication-efficient, secure, and differentially private FL. 
		To tackle the communication bottleneck in FL, \textsf{Clover} follows a standard and commonly used approach---top-$k$ gradient sparsification, where each client sparsifies its gradient update such that only $k$ largest gradients (measured by magnitude) are preserved for aggregation.
		\textsf{Clover} provides a tailored mechanism built out of a trending distributed trust setting involving three servers, which allows to efficiently aggregate multiple sparse vectors (top-$k$ sparsified gradient updates) into a dense vector while hiding the values and indices of non-zero elements in each sparse vector.  
		This mechanism outperforms a baseline built on the general distributed ORAM technique by several orders of magnitude in server-side communication and runtime, with also smaller client communication cost.
		We further integrate this mechanism with a lightweight distributed noise generation mechanism to offer differential privacy (DP) guarantees on the trained model.
		To harden \textsf{Clover} with security against a malicious server, we devise a series of lightweight mechanisms for integrity checks on the server-side computation.
		Extensive experiments show that \textsf{Clover} can achieve utility comparable to vanilla FL with central DP and no use of top-$k$ sparsification.
		Meanwhile, achieving malicious security introduces negligible overhead in client-server communication, and only modest overhead in server-side communication and runtime, compared to the semi-honest security counterpart.

	\end{abstract}
	
	
	\settopmatter{printfolios=true}
	
	\maketitle
	
	\section{Introduction} 
	Federated learning (FL) \cite{fedavg} is an increasingly popular privacy-aware machine learning paradigm where clients, coordinated by a central server, jointly train a model by sharing only gradient updates instead of raw data. 
	In each round of FL training, participating clients compute gradient updates using their local data and share them with the server for aggregation to produce an updated global model.
	Due to the sharing of very high-dimensional gradient updates from many clients, it is widely known that communication could easily become a performance bottleneck in FL \cite{Kairouz21OpenChallenge}.
	
	To alleviate the issue of high client-server communication costs in FL, a widely recognized effective method is gradient sparsification \cite{WangniWLZ18,olive,HuGG24, LuLLGY23,Qsparse}. 
	It enables clients to share only a small subset of gradients for substantial communication reduction, while retaining model utility comparable to the case of sharing dense gradient updates.
	Among others, top-$k$ sparsification \cite{WangniWLZ18,olive,HuGG24} is a standard and well-known approach, wherein each client selects $k$ gradients with the largest magnitudes in its gradient update vector, and then only shares out their values along with indices to the server.
	The top-$k$ sparsified gradients from different clients are aggregated into a dense aggregated gradient update at the server side for updating the global model. 
	Since $k$ is far less than the dimension $d$ of a client's gradient update vector (e.g., $k/d=0.3\%$ \cite{olive}), applying top-$k$ sparsification can substantially reduce client-server communication costs in FL.

	Integrating top-$k$ gradient sparsification into FL systems, however, needs a careful treatment due to the existence of critical privacy concerns. 
	Firstly, it is widely known that even just sharing gradients can still lead to privacy leakage of clients' local datasets \cite{MelisSCS19,ZhuLH19}.
	This poses a demand for securing the gradients shared by clients in FL.
	In the context of FL with top-$k$ sparsification, protection is required not only for the selected gradient values but also for their indices. 
	Indeed, it has been recently shown that even just exploiting the indices of gradients in top-$k$ sparsification could allow inference of private information of local datasets \cite{olive}. 
	While there exists a popular cryptographic tool called secure aggregation (denoted by SecAgg) \cite{BonawitzIKMMPRS17} for FL, it was proposed for securely aggregating \emph{dense} gradient update vectors.
	Extending SecAgg to work with top-$k$ gradient sparsification would require that the clients share the same indices of top-$k$ gradients, which is impractical as observed in \cite{Ergun21}.
	It is thus imperative to explore how to realize correct aggregation of top-$k$ sparsified gradient updates while hiding both selected gradient values and indices of individual clients.
	
	Protecting individual gradient updates during the aggregation process, however, is still insufficient for strong privacy protection, due to potential privacy leakage from the trained model \cite{MelisSCS19,shokri2017membership}. 
	In order to prevent such kind of leakage, a widely recognized practice \cite{olive,YangHYGC23,StevensSVRCN22} is to make the trained model differentially private. 
	Specifically, instead of directly exposing the trained model, calibrated differential privacy (DP) noise is added so that only a noisy version of the trained model is revealed. 
	Such addition of calibrated noise can provide DP guarantees, thwarting inference attacks on the trained model.
	In the literature, there exist a number of prior studies on providing DP guarantees in FL \cite{McMahanRT018,AndrewTMR21,Truex0CGW20,AgarwalKL21,kairouz21,StevensSVRCN22,olive}; see Section \ref{sec:related-work} for detailed discussion.
	However, to the best of our knowledge, only the very recent work in \cite{olive} has specifically explored how to effectively render FL systems adopting standard top-$k$ sparsification with DP guarantees on the trained model, while ensuring protection of individual gradient updates and good model utility at the same time.
	The approach in \cite{olive}, however, builds on the use of trusted hardware (Intel SGX) and thus has to place trust in hardware vendors.
	Also, it is widely known that trusted hardware is actually vulnerable to various side-channel attacks \cite{sgx-attacks}.
	
	In light of the above, in this paper, we propose {\main}, a new system framework for top-$k$ sparsification-enabled secure and differentially private FL. 
	{\main} allows the use of standard top-$k$ sparsification in FL for high client-server communication efficiency, while offering full-fledged privacy protection (for individual sparsified gradient updates as well as the trained models) and robustness against a malicious server.
	At a high level, we build {\main} through a delicate synergy of gradient sparsification, data encoding, lightweight cryptography, and differential privacy.
	{\main} leverages a distributed trust setting, where three servers from different trust domains jointly empower the FL service for clients.
	Such a distributed trust setting has been increasingly adopted in recent years for building secure systems related to FL as well as other application domains \cite{ELSA, GehlharM0SWY23, tang2024flexible, camel,DautermanRPS22, NDSS22, RatheeZCP24, PLASMA, bruggemann2024don,he2024rhombus}.
	
	Leveraging distributed trust and lightweight secret sharing techniques, our starting point is the design of a communication-efficient secure sparse vector aggregation mechanism \textsf{SparVecAgg} (see Section \ref{sec:sparvecagg}).
	Consider a set of clients, each holding a $d$-dimensional sparse vector in which there are only $k$ non-zero elements ($k \ll d$).
	Such a sparse vector can be treated as a top-$k$ sparsified gradient update in the context of FL with top-$k$ sparsification. 
	\textsf{SparVecAgg} aims to aggregate these sparse vectors into a dense vector, while achieving client communication independent of $d$ and hiding the indices and values of non-zero elements in each sparse vector.
	
	\revise{
		A strawman solution is to have the clients directly secret-share the non-zero values along with their indices among the servers, and rely on the general distributed ORAM (Oblivious Random Access Machine) \cite{ORAM1, ORAM2}, a standard cryptographic primitive supporting oblivious read and write operations, for oblivious server-side aggregation. }
	While this strawman can enable client communication independent of $d$, it is prohibitively expensive in server-side performance, due to the heavy oblivious write operation being performed over a very high-dimensional vector for every index-value pair.
	To solve this issue, we propose a customized mechanism that substantially outperforms the strawman in server-side performance.
	Our key insight is to conduct a customized permutation-based encoding of sparse vectors and delicately apply lightweight secret-shared shuffle at the server side to efficiently reconstruct original sparse vectors in secret-shared form for aggregation. 
	We also develop a tailored permutation compression mechanism for ensuring high client communication efficiency, which allows \textsf{SparVecAgg} to even achieve a reduction of about 33\% in client communication compared to the strawman solution.

	We next consider how to ensure DP guarantees for the trained model derived from aggregating the (sparse) gradient updates.
	The problem here is how to securely inject DP noise into the aggregated gradient update before revealing it.
	While there exist secure noise sampling protocols \cite{FuW24,WeiYFCW23,ARES24} which allow the servers to jointly sample DP noise in an oblivious manner, they would incur prohibitively high performance overheads \cite{PETS25}. 
	Instead we turn to an alternative strategy of distributed noise generation \cite{PETS25} to enable secure and efficient noise addition.
	The main idea is to have each server locally sample a calibrated noise, which is then secret-shared among the servers and ultimately integrated into the aggregated gradient update with high efficiency.

	With the design for semi-honest security as a basis (Section \ref{sec:semi-honest-construction}), we further design a series of lightweight integrity checks for the server-side computation (Section \ref{sec:malicious}), defending against a compromised server that might maliciously deviate from the designated protocol and compromise computation integrity.
	Specifically, for integrity check on the reconstruction of secret-shared sparse vectors based on permutation decompression and secret-shared shuffle, we devise a tailored blind MAC verification mechanism. 
	Compared to the use of standard information-theoretic MACs which require two secure dot product operations for verification \cite{AsharovHIKNPTT22}, our mechanism requires only one.
	\revise{
		To verify the correctness of each server's locally sampled DP noise, we propose a lightweight verification mechanism using the Kolmogorov-Smirnov (KS) test, which constrains the malicious manipulation in noise sampling to be insignificant. 
	}

	We implement the protocols of {\main} and conduct an extensive performance evaluation regarding the trained model utility and system efficiency (Section \ref{sec:experiments}).
	Results show that {\main} achieves utility comparable to vanilla \textsf{DP-FedAvg} \cite{McMahanRT018} (under central DP), with promising performance. 
	For instance, for securely aggregating 100 sparse vectors of dimension $10^5$ with 1\% non-zero entries in the semi-honest security setting, our \textsf{SparVecAgg} achieves up to 1602$\times$ reduction in inter-server communication cost and up to 12,041$\times$ less server-side computation cost compared to the baseline built on distributed ORAM (implemented using the popular generic secure multi-party computation (MPC) framework MP-SPDZ \cite{Keller20}). 
	Compared to the semi-honest security counterpart, the extension for achieving malicious security introduces negligible overhead in client-server communication, and leads to only about 2.67$\times$ higher inter-server communication cost and at most 3.67$\times$ higher server-side runtime.
	We highlight our main contributions below: 
	\vspace{-2pt}
	\begin{itemize}
		
		\item We present {\main}, a new system framework for secure, efficient, and differentially private FL with top-$k$ sparsification.
		
		\item We propose a novel lightweight secure sparse vector aggregation mechanism through a synergy of customized data encoding and lightweight secret-shared shuffle technique, and further integrate it with a lightweight distributed noise generation mechanism for securing trained models.
		
		\item We devise a series of lightweight mechanisms for integrity checks on the holistic server-side secure computation, achieving robustness against a malicious server. 
		
		\item We formally analyze the privacy, communication, convergence, and security of {\main}. 
		We make a prototype implementation and conduct extensive experiments. 
		The results demonstrate that {\main} achieves utility comparable to FL with central DP, with promising performance. 
	\end{itemize}

	\section{Related Work}
	\label{sec:related-work}

	%
	In order to achieve differential privacy guarantees for FL, some prior works \cite{McMahanRT018,Geyer17,AndrewTMR21} adopt the central DP model. 
	These approaches assume a \textit{trusted} central server that collects gradient updates from clients, aggregates them, and injects DP noise into the aggregated result.
	%
	%
	For instance, the \textsf{DP-FedAvg} algorithm \cite{McMahanRT018} extends the classical \textsf{FedAvg} algorithm \cite{fedavg} by having the central server add Gaussian noise to the aggregated gradient update. 
	However, the assumption of a trusted central server is challenging to satisfy in practice due to growing concerns about data privacy \cite{roy2020crypt}. 
	
	%
	Another line of work \cite{Truex0CGW20, ChamikaraLCNGBK22, miao2022compressed, sun2021ldp, liu2020fedsel, JiangZG22} explores using local differential privacy (LDP) mechanisms in FL to eliminate the need for a trusted central server. 
	With LDP, clients locally perturb gradient updates with calibrated noises before sending them to the untrusted server. 
	Among these works, only a few \cite{JiangZG22,liu2020fedsel} incorporate gradient sparsification into LDP-based FL. 
	\revise{
		For instance, Liu et al. \cite{liu2020fedsel} propose a two-stage LDP-FL framework \textsf{FedSel} consisting of a dimension selection (DS) stage and a value perturbation (VP) stage. 
		In the DS stage, each client privately selects one important dimension from its local top-$k$ set. 
		In the VP stage, the selected value is perturbed using an LDP mechanism to construct a sparse privatized update.}
	However, the LDP-protected index information may still exhibit correlations with the client's private training data, potentially leading to privacy leakage \cite{olive}. 
	Furthermore, LDP typically requires heavy noise addition, causing significant degradation in model utility \cite{RuanXFWWH23}. 
	In other works \cite{CVPR22, CVPR23}, each client also locally adds noise to its sparsified gradient update, yet the noise is very small and insufficient for privacy protection, not even meeting the requirement of LDP. 
	
	%

	To achieve utility comparable to FL with central DP without relying on a trusted server, recent works \cite{kairouz21, AgarwalKL21, StevensSVRCN22} combine secure aggregation (SecAgg) \cite{BonawitzIKMMPRS17, BellBGL020} with DP to train differentially private models in FL. 
	In particular, each client adds a small amount of DP noise to its gradient update and then engages in SecAgg with the server to securely produce only the aggregated noisy gradient update, for which the aggregated noise ensures a meaningful central DP guarantee. 
	These works do not consider the application of gradient sparsification for communication efficiency optimization.
	
	%
	
	\revise{
		Several existing works \cite{KerkoucheEuroSP21,KerkoucheUAI21,HuGG24,LuLLGY23,Ergun21} have attempted to combine gradient compression techniques like compressive sensing (CS) \cite{KerkoucheEuroSP21} and gradient sparsification \cite{KerkoucheUAI21,HuGG24,LuLLGY23,Ergun21} with SecAgg.
		For instance, Kerkouche et al.~\cite{KerkoucheEuroSP21} explore the linearity of CS to enable SecAgg over CS-compressed gradient updates. 
		However, it makes a strong assumption that the gradient update is inherently sparse in the time domain, as also noted in \cite{HuGG24}. 
		Among works employing gradient sparsification, Kerkouche et al. \cite{KerkoucheUAI21} combine SecAgg and DP in sparsification-enabled FL, but require the server to choose a common set of top-$k$ indices for all clients using an auxiliary public dataset. 
		Similarly, Hu et al.~\cite{HuGG24} propose two schemes that combine SecAgg and DP: \textsf{FedSMP-top$_k$}, which also relies on an auxiliary public dataset to identify common top-$k$ indices, and \textsf{FedSMP-rand$_k$}, in which the server randomly selects a common set of $k$ indices for all clients.
		While \textsf{FedSMP-rand$_k$} avoids the need for auxiliary data, it corresponds to random-$k$ sparsification, which is less effective than top-$k$ sparsification in preserving model utility~\cite{olive,LuLLGY23}. 
		In \cite{LuLLGY23}, Lu et al. propose to generate a globally shared top-$k$ index set by taking a union set of individual clients' top-$k$ indices, which incurs extra communication costs. 
		Nevertheless, there is usually very little overlap among the top-$k$ indices for each client \cite{Ergun21}, so using a globally shared top-$k$ index set would impair training quality and is impractical. 
		Ergun et al. \cite{Ergun21} propose to let each pair of users share a common set of randomly selected $k$ indices to facilitate the combination of gradient sparsification and SecAgg, which again corresponds to the less effective random-$k$ sparsification and suffers from degraded utility.
		In short, these works either need to rely on impractical assumption, or do not support standard top-$k$ sparsification. 
	}

	The state-of-the-art work that is closest to this paper is by Kato et al. \cite{olive}.
	They propose a trusted hardware-based framework that integrates native top-$k$ sparsification into FL while ensuring DP guarantees. 
	They first demonstrate a label inference attack exploiting top-$k$ indices of individual gradient updates, underscoring the need to protect both the values and indices of top-$k$ gradients.
	They also propose an Intel SGX-based solution for securing top-$k$ gradients.
	Such a trusted hardware-based solution is orthogonal and requires trust in hardware vendors. 
	Meanwhile, trusted hardware is known to be vulnerable to various side-channel attacks \cite{sgx-attacks}.

	An orthogonal line of work \cite{RuanXFWWH23,PETS25} studies multi-party differentially private model training. 
	Unlike FL, which keeps data local and shares only gradient updates, these frameworks follow a different workflow: 
	(1) each party secret-shares its dataset with others, and 
	(2) distributed DP-SGD \cite{DPSGD} is performed on the secret-shared dataset via MPC. 
	Notably, these works \cite{RuanXFWWH23,PETS25} also only consider semi-honest adversary and do not offer security against malicious adversary. 
	While some other works target the FL setting and can provide malicious security against servers (e.g., ELSA \cite{ELSA}), they do not consider the use of gradient sparsification and achieving strong DP guarantees while preserving good model utility. 
	To the best of our knowledge, our work is the first to harness sparsification in FL for high communication efficiency while simultaneously ensuring security against a malicious server and formal DP guarantees. 
	

	\section{Background}
	\label{sec:preliminary}
	
	\subsection{Federated Learning}
	\label{sec:pre:FL}
	
	Federated learning (FL) enables multiple clients to collaboratively train a global model without sharing their local training data. 
	Consider an FL system consisting of $n$ clients, where each client $\mathcal{C}_i$ holds a local dataset $\mathcal{D}_i$. 
	The set $\mathcal{D} = \cup_{i=0}^{n-1}\mathcal{D}_i$ denotes the full training set. 
	Formally, the overarching goal in FL can be formulated as follows: $
	\min_{\mathbf{w} \in \mathbb{R}^d} F(\mathbf{w}, \mathcal{D}) = \sum_{i=0}^{n-1} \frac{1}{n} F_i\left(\mathbf{w}, \mathcal{D}_i\right), $
	\noindent where $F_i(\cdot)$ denotes the local loss function for client $\mathcal{C}_i$ over a model $\mathbf{w}$ and a local dataset $\mathcal{D}_i$, which captures how well the parameters $\mathbf{w}$ (treated as a $d$-dimensional flattened vector) model the local dataset. 
	The standard FL algorithm to solve this optimization problem is \textsf{FedAvg} \cite{fedavg}, which runs in multiple rounds (say $T$ rounds) and the $t$-th round ($0\leq t\leq T-1$) proceeds as follows:
	
	(1) \textit{Model broadcasting}: The server samples a subset of clients $\mathcal{P}^t$ and sends the current global model $\mathbf{w}^{t}$ to these clients. 
	
	(2) \textit{Local update}: Each selected client $\mathcal{C}_i$ performs $E$ iterations of stochastic gradient descent (SGD) to train a local model using its local training dataset $\mathcal{D}_i$. 
	In FL, a single iteration over the entire local dataset is referred to as an epoch. The update process proceeds as: $\mathbf{w}^{e+1}_i = \mathbf{w}^e_i - \eta_l \nabla F_i \left(\mathbf{w}^e_i , \beta_i \right)$ for $0\leq e\leq E-1$, where $\beta_i$ represents a batch randomly sampled from the local dataset $\mathcal{D}_i$, $\eta_l$ is the learning rate, and $\mathbf{w}^0_i$ is initialized with $\mathbf{w}^{t}$. 
	%
	
	(3) \textit{Global model aggregation}: After completion of the local update process, each selected client shares the individual gradient update $\Delta_i=\mathbf{w}^{E}_i-\mathbf{w}^{0}_i$ with the server.
	The server then produces an updated global model for the next round via: $\mathbf{w}^{t+1} = \mathbf{w}^{t} + \frac{1}{|\mathcal{P}^t|}\sum_{\mathcal{C}_i\in\mathcal{P}^t}  \Delta_i$, where $|\mathcal{P}^t|$ denotes the number of selected clients.

	\subsection{Differential Privacy}
	Differential privacy (DP) is a rigorous definition for formalizing the privacy guarantees of algorithms that process sensitive data. 
	We say that two datasets $\mathcal{D},\mathcal{D}'\in\mathcal{X}^n$ are neighboring datasets if they differ in one data record. 
	The formal definition of DP is as follows.
	
	\begin{defn}
		\textit{\textbf{(Differential Privacy - ($\varepsilon,\delta$)-DP \cite{DworkR14}).}}
		\emph{
			A randomized mechanism $\mathcal{M}$: $\mathcal{X}^n\rightarrow \mathcal{Y}$ satisfies $(\varepsilon,\delta)$-DP if for any two neighboring datasets $\mathcal{D},\mathcal{D}^{\prime} \in \mathcal{X}^n$ and for any subset of outputs $\mathcal{S}\subseteq\mathcal{Y}$ it holds that $\operatorname{Pr}[\mathcal{M}(\mathcal{D}) \in \mathcal{S}] \leq e^{\varepsilon} \operatorname{Pr}\left[\mathcal{M}\left(\mathcal{D}^{\prime}\right) \in \mathcal{S}\right]+\delta$.
		}
	\end{defn}
	
	\noindent Here, the parameter $\varepsilon$ denotes the privacy budget, where a smaller value results in stronger privacy guarantees, but with a corresponding reduction in data utility. 
	The parameter $\delta$ denotes the probability that privacy is violated and is typically selected to be very small. 
	The mechanism is said to satisfy $\varepsilon$-DP or pure DP when $\delta = 0$.
	
	In this paper, we also consider the notion of R\'enyi differential privacy (RDP) \cite{RDP} as a generalization of differential privacy, which allows for tight tracking of privacy loss over multiple iterations in differentially private learning algorithms. 
	
	\begin{defn}
		\textit{\textbf{(R\'enyi Differential Privacy - ($\alpha,\tau$)-RDP {\cite{RDP}}).}}
		\emph{
			A randomized mechanism $\mathcal{M}$: $\mathcal{X}^n\rightarrow \mathcal{Y}$ satisfies $(\alpha,\tau)$-RDP if for any two neighboring datasets $\mathcal{D},\mathcal{D}^{\prime}  \in \mathcal{X}^n$, we have 
			$
			D_\alpha\left(\mathcal{M}(\mathcal{D})\|\mathcal{M}\left(\mathcal{D}^{\prime}\right)\right) \leq \tau,
			$ 
			where $D_\alpha(P\| Q)$ is the R\'enyi divergence between two probability distributions P and Q and is given by
			$$
			D_\alpha(P\| Q) \triangleq \frac{1}{\alpha-1} \log \left(\mathbb{E}_{x \sim Q}\left[\left(\frac{P(x)}{Q(x)}\right)^\alpha\right]\right).
			$$
		}
		
	\end{defn}
	
	RDP offers a key advantage over the notion of ($\varepsilon,\delta$)-DP due to its composition property, which is highlighted as follows.
	
	\begin{lem}\textit{\textbf{(Adaptive Composition of RDP \cite{RDP}).}}
		\label{lem:sequantial}
		Let $\mathcal{M}_1: \mathcal{D}\rightarrow\mathcal{R}_1$ be a mechanism satisfying ($\alpha,\tau_1$)-RDP and $\mathcal{M}_2: \mathcal{D}\times\mathcal{R}_1\rightarrow\mathcal{R}_2$ be a mechanism satisfying ($\alpha,\tau_2$)-RDP. Define their combination $\mathcal{M}_{1,2}:\mathcal{D}\rightarrow\mathcal{R}_2$ by $\mathcal{M}_{1,2}(\mathcal{D})=\mathcal{M}_2(\mathcal{D},\mathcal{M}_1(\mathcal{D}))$. Then $\mathcal{M}_{1,2}$ satisfies ($\alpha, \tau_1+\tau_2$)-RDP.
	\end{lem}
	
	Since RDP is a generalization of DP, it can be easily converted back to ($\varepsilon,\delta$)-DP using the standard conversion lemma as follows.
	
	\begin{lem}\textit{\textbf{(From RDP to DP \cite{wang2019subsampled}).}}
		\label{lem:rdp_to_dp}
		If a randomized mechanism $\mathcal{M}$ is ($\alpha,\tau$)-RDP, then the mechanism is also ($\varepsilon,\delta$)-DP, where $\varepsilon=( \tau+\frac{\log (1 / \delta)}{\alpha-1})$ for a given $\delta\in(0,1)$. 
	\end{lem}
	
	A basic approach to achieving RDP is through the \textit{Gaussian mechanism}, which adds calibrated noise from a Gaussian distribution to the output of the data analytics.
	
	\begin{lem}
		\label{thm:gaussian}
		\textit{\textbf{(Gaussian Mechanism{\cite{RDP}}).} 
			Let $f:\mathcal{X}^n\rightarrow \mathbb{R}^d$ be a non-private algorithm taking a dataset as input and outputting a $d$-dimensional vector. 
			For neighboring datasets $\mathcal{D}$ and $\mathcal{D}^{\prime}$, define $f$'s sensitivity as 
			$\Delta_f = \max _{\mathcal{D}, \mathcal{D}^{\prime}} \left\|f(\mathcal{D}) - f\left(\mathcal{D}^{\prime}\right)\right\|_2$.
			The Gaussian mechanism satisfying $(\alpha,\alpha/2\sigma^2)$-RDP is of the form $\mathcal{M}(\mathcal{D})=f(\mathcal{D})+\mathcal{N}\left(0,\sigma^2\Delta_f^2\mathbf{I}_d\right).$ 
		}
	\end{lem}
	
	\revise{
		In the context of FL, DP can be categorized into record-level DP and client-level DP, depending on how neighboring datasets are defined. 
		In this paper, following existing works \cite{McMahanRT018, KerkoucheUAI21,KerkoucheEuroSP21,olive,HuGG24,CVPR22,CVPR23}, we focus on client-level DP,  where two datasets $\mathcal{D}$ and $\mathcal{D}^{\prime}$ are said to be neighboring if $\mathcal{D}^{\prime}$ can be formed by adding or removing all data records associated with a single client from $\mathcal{D}$. 
	}
	
	\subsection{Replicated Secret Sharing}
	\label{sec:background:Replicated-Secret-Sharing}
	Given a private value $x\in\mathbb{Z}_p$, replicated secret sharing (RSS) \cite{RSS} works by splitting it into three secret shares $\langle x \rangle_0$, $\langle x \rangle_1$ and $\langle x \rangle_2\in\mathbb{Z}_p$ such that $x = \langle x \rangle_0+\langle x \rangle_1+\langle x \rangle_2\: \bmod p$, where $p$ is a prime number. 
	The shares are then distributed among three parties, denoted by $\mathcal{P}_0$, $\mathcal{P}_1$ and $\mathcal{P}_2$, where each party holds a pair of shares.
	In particular, $\mathcal{P}_0$ holds $(\langle x \rangle_0,\langle x \rangle_1)$, $\mathcal{P}_1$ holds $(\langle x \rangle_1,\langle x \rangle_2)$, and $\mathcal{P}_2$ holds $(\langle x \rangle_0,\langle x \rangle_2)$.
	%
	%
	For simplicity, we write $i\pm1$ to represent the next ($+$) party (or secret share) or the previous ($-$) party (or secret share) with wrap around. 
	For example, $\mathcal{P}_{2+1}$ (or $\langle x \rangle_{2+1}$) refers to $\mathcal{P}_0$ (or $\langle x \rangle_0$) and $\mathcal{P}_{0-1}$ (or $\langle x \rangle_{0-1}$) refers to $\mathcal{P}_{2}$ (or $\langle x\rangle_{2}$).
	In this way, the secret shares held by $\mathcal{P}_{i}$ can be represented as $(\langle x \rangle_i,\langle x \rangle_{i+1})$ for $i\in\{0,1,2\}$. 
	We use $\llbracket x \rrbracket$ to denote such way of secret sharing of $x$. 
	Unless otherwise stated, we omit ($\bmod$ $p$) for brevity in the subsequent description of RSS-based operations.

	The basic operations related to RSS are as follows. 
	(1) \textit{Reconstruction}. To reconstruct ($\mathsf{Rec(\cdot)}$) $x$ from $\llbracket x \rrbracket$, $\mathcal{P}_i$ sends $\langle x \rangle_{i+1}$ to $\mathcal{P}_{i-1}$, which reconstructs $x = \langle x \rangle_0 + \langle x \rangle_1 + \langle x \rangle_2$ for $i\in\{0,1,2\}$. 
	(2) \textit{Addition/subtraction}. Addition/subtraction of secret-shared values is completed by $\mathcal{P}_i$ locally for $i\in\{0,1,2\}$: To securely compute $\llbracket z \rrbracket = \llbracket x\pm y \rrbracket$, each party $\mathcal{P}_i$ locally computes $\langle z \rangle_i = \langle x \rangle_i\pm\langle y \rangle_i$ and $\langle z \rangle_{i+1} = \langle x \rangle_{i+1}\pm\langle y \rangle_{i+1}$. 
	(3) \textit{Multiplication}. 
	To securely compute $\llbracket z \rrbracket = \llbracket x \cdot y \rrbracket$, $\mathcal{P}_i$ first locally computes $\langle z \rangle_i = \langle x \rangle_i\cdot\langle y \rangle_i + \langle x \rangle_i\cdot\langle y \rangle_{i+1} + \langle x \rangle_{i+1}\cdot\langle y \rangle_i$. 
	This local computation produces a 3-out-of-3 additive secret sharing of $z$ among the three parties, i.e., each party $\mathcal{P}_i$ only holds $\langle z \rangle_i$.
	To generate an RSS of $z$ for subsequent computations, a \textit{re-sharing} operation can be carried out as follows: $\mathcal{P}_i$ sends $\langle z^\prime \rangle_i = \langle z \rangle_i + \langle \alpha \rangle_i$ to $\mathcal{P}_{i-1}$, where $\langle \alpha \rangle_i$ is a share from a 3-out-of-3 additive secret sharing of zero, i.e., $\langle \alpha \rangle_0 + \langle \alpha \rangle_1 + \langle \alpha \rangle_2 = 0$. 
	Such a secret sharing of zero can be generated in an offline phase via pseudorandom functions (PRFs) \cite{RSS}. 
	
	\subsection{Permutation}
	\label{sec:preliminary:permutation}
	%
	A permutation $\pi$ is a rearrangement of $n$ elements in a definite order. 
	In this paper, $\pi$ is explicitly represented as an ordered tuple whose elements are distinct
	from each other and belong to $[n]$, where $[n]$ denotes the set $\{0,\cdots,n-1\}$ for $n\in\mathbb{N}$.
	For an $n$-dimensional vector $\boldsymbol{v}$, applying $\pi$ to $\boldsymbol{v}$ rearranges its elements such that the element at position $i$ is moved to position $\pi(i)$. 
	Here $\pi(i)$ denotes the destination index of the $i$-th element of $\boldsymbol{v}$. 
	The resulting permuted vector is denoted as $\pi(\boldsymbol{v})$. 
	%
	%
	%
	%
	The composition of two permutations $\pi_1$ and $\pi_2$ is denoted by $\pi_1 \circ \pi_2$, which satisfies $\pi_1 \circ \pi_2(\boldsymbol{v}) = \pi_1(\pi_2(\boldsymbol{v}))$. 
	The inverse of a permutation $\pi$ is denoted as $\pi^{-1}$, and it satisfies $\pi \circ \pi^{-1}(\boldsymbol{v}) = \pi^{-1} \circ \pi(\boldsymbol{v}) = \boldsymbol{v}$.
	In this paper, we also work with permutations that are secret-shared across the parties, in a way similar to RSS. Specifically, we use $\llangle \pi \rrangle$ to denote a sharing of permutation $\pi$. 
	Let $\pi=\pi_0\circ\pi_1\circ\pi_2$, where $\pi_0,\pi_1$ are random permutations, and $\pi_2 = \pi_1^{-1}\circ\pi_0^{-1}\circ\pi$.
	Then $\pi$ is secret-shared across the parties in a replicated way: $\mathcal{P}_i$ holds the pair $(\pi_i,\pi_{i+1})$ for $i\in\{0,1,2\}$. 
	
	\section{Problem Statement}
	
	\subsection{System Model}
	\label{sec:problem}
	%
	{\main} is aimed at the realization of sparsification-enabled secure and differentially private federated learning.
	There are two types of actors in {\main}: clients and servers.
	Clients hold their datasets locally and want to collaboratively train a model, under the coordination of the servers.
	The design of {\main} considers a distributed trust setup where three servers from different trust domains collaborate in providing the FL service for the clients. 
	Such distributed trust model has also gained adoption in existing works on secure FL \cite{ELSA, GehlharM0SWY23, tang2024flexible, camel}, as well as in other secure systems and applications \cite{prio,DautermanRPS22, NDSS22, RatheeZCP24, PLASMA, bruggemann2024don,he2024rhombus}.
	For ease of presentation, we denote the three servers $\mathcal{S}_0$, $\mathcal{S}_1$, and $\mathcal{S}_2$ collectively as $\mathcal{S}_{\{0,1,2\}}$. 
	
	At a high level, in each training round of {\main}, clients train their local models using their private datasets and then \textit{securely} share the resulting gradient updates with $\mathcal{S}_{\{0,1,2\}}$. 
	Then $\mathcal{S}_{\{0,1,2\}}$ aggregate the received gradient updates to refine the global model. 
	To reduce client-server communication costs, we let each client first sparsify its gradient update (treated as a vector) using the top-$k$ sparsification mechanism before sharing it to the servers. 
	To securely and efficiently aggregate the top-$k$ sparsified gradient updates, we introduce a communication-efficient secure sparse vector aggregation mechanism. 
	In addition to preserving the privacy of individual sparsified gradient updates, {\main} also enforces lightweight integrity checks on the server-side aggregation process.  
	Meanwhile, to ensure that the aggregated gradient update satisfies DP, we leverage the strategy of lightweight distributed noise generation  \cite{kairouz21,Canonne0S20,google2020secure,PETS25} to enable servers to securely and efficiently inject DP noise into the aggregated gradient update, and devise a lightweight verification mechanism to ensure the correctness of the injected DP noise.

	\subsection{Threat Model and Security Guarantees}
	\label{sec:threat_model}
	
	Similar to prior works under the three-server setting \cite{NDSS22,PLASMA,camel,AsharovHIKNPTT22}, we consider a non-colluding and honest-majority threat model against the servers. 
	Specifically, we assume that each server of $\mathcal{S}_{\{0,1,2\}}$ may individually try to deduce private information during the protocol execution. 
	At most one of the three servers may maliciously deviate from our protocol specification. 
	{\main} aims to ensure that the servers learn no information about individual gradient updates during aggregation. 
	What the servers observe in each training round is a noisy aggregate that satisfies DP (which will be validated through formal analysis). 
	\revise{
		{\main} detects and outputs abort if the noise is heavily deviating/not drawn from the target distribution, or if a malicious server tampers with the integrity of other server-side computations. 
	}
	It is noted that similar to prior works on privacy-aware FL \cite{kairouz21,AgarwalKL21,ChenCKS22,camel}, {\main} focuses on the threats primarily from the servers and provides privacy protection for the clients. 
	Defense against adversarial attacks from clients is an orthogonal line of research and is out of the scope of this work. 
	
	\section{Communication-Efficient Secure Sparse Vector Aggregation}\label{sec:sparvecagg}
	
	In this section, we introduce \textsf{SparVecAgg}, a customized mechanism for communication-efficient secure sparse vector aggregation, which will serve as the basis of {\main} for securely and efficiently harnessing top-$k$ gradient sparsification in FL.

	\subsection{Design Rationale}
	\label{sec:sparvecagg:design_rationale}
	
	The problem of secure sparse vector aggregation is described as follows. 
	Each client $\mathcal{C}_i$ holds a $d$-dimensional sparse vector $\boldsymbol{x}_i$ as defined by:
		\begin{equation}
			\label{eq:top-k}
			\boldsymbol{x}_i[j] = 
			\begin{cases} 
				\boldsymbol{x}_i[j], & \text{if } j \in \mathcal{I}_i, \\ 
				0, & \text{otherwise},
			\end{cases}
		\end{equation}
	where $\boldsymbol{x}_i[j]$ denotes the $j$-th element in $\boldsymbol{x}_i$, only $k$ elements in $\boldsymbol{x}_i$ are non-zero values, and $\mathcal{I}_i\subseteq [d]$ is a set that contains the indices associated with them. 
	%
	The density of the sparse vector $\lambda$ is defined as $\lambda=k/d$. 
	To aggregate the sparse vectors across $n$ clients and produce $\boldsymbol{x} = \sum_{i\in[n]} \boldsymbol{x}_i$ (a dense vector), it is required that each client's input $\boldsymbol{x}_i$ (with respect to the non-zero values and their positions) should not be disclosed for privacy protection.
	In the meantime, for communication efficiency, it is desirable for the aggregation process to be sparsity-aware, meaning that each client's communication cost ideally should depend on the number $k$ of non-zero values, rather than the vector size $d$ which could be quite large in some applications like FL.
	
	\noindent\textbf{Strawman approach.} A strawman solution is to directly secret-share the individual vectors' non-zero elements and corresponding indices, and then rely on distributed ORAM \cite{ORAM1,ORAM2}---which can support oblivious write over secret-shared data---to aggregate the secret-shared index-value pairs. 
	Specifically, the servers collectively maintain a secret-shared vector. 
	For each pair of index and value in secret-shared form, the servers perform oblivious write on that vector to obliviously write the value to the target index. 
	However, as shown in Section \ref{sec:experiments}, this strawman approach is highly inefficient for supporting sparsity-aware oblivious aggregation in FL, due to the expensive oblivious write operation being performed over the large-sized vector for every index-value pair at the server side.
	
	\begin{table}[!t]
		\centering
		\caption{Comparison of per-vector client communication cost between different approaches for oblivious sparse vector aggregation (with semi-honest security). Here, $|\mathbb{Z}_p|$ denotes the bit length of a field element, $|s|$ is the length of a small-sized random seed (e.g., 128 bits), and $k \ll d$.}
		\label{tab:client-server-communication-cost}
		\begin{tabular}{cc}
			\hline
			& Client Comm. Cost (bits) \\ \hline
			\textsf{SparVecAgg} (w/o compression)  &  $6d |\mathbb{Z}_p| + 6k|\mathbb{Z}_p|$\\ 
			\textsf{SparVecAgg} (w/ compression) & $8k |\mathbb{Z}_p| + 4|s|$ \\ 
			Strawman approach & $12k|\mathbb{Z}_p|$ \\ \hline
		\end{tabular}
		\vspace{-9pt}
	\end{table}
	
	\noindent\textbf{Our insight.} We present a customized design that significantly outperforms the strawman solution, in terms of client-server communication, inter-server communication, as well as server-side computation. Our first insight is to re-organize a sparse vector so that the sparse vector can be appropriately encoded by the $k$ non-zero values and a specific permutation.
	For example, if $\mathcal{C}_i$ re-arranges the elements in $\boldsymbol{x}_i$ by placing the non-zero elements at the beginning of $\boldsymbol{x}_i$, this will produce a new vector $\boldsymbol{x}'_i$ and a corresponding permutation $\pi_i$ such that $\boldsymbol{x}_i = \pi_i(\boldsymbol{x}'_i)$.
	The first $k$ values in $\boldsymbol{x}'_i$ and the permutation $\pi_i$ are then \emph{appropriately} secret-shared among the servers.
	With the customized encoding on sparse vectors, the problem is then transformed to oblivious application of the permutation $\pi_i$ to $\boldsymbol{x}'_i$ so that $\llbracket \boldsymbol{x}_i \rrbracket = \llbracket \pi_i(\boldsymbol{x}'_i) \rrbracket$ can be produced for secure aggregation/summation, without the servers learning $\pi_i$, the first $k$ values in $\boldsymbol{x}'_i$, and $\boldsymbol{x}_i$.
	We make an observation that this can be achieved by lightweight secret-shared shuffle techniques \cite{AsharovHIKNPTT22,Araki0OPRT21}, which can securely shuffle data in the secret sharing domain, without revealing the underlying data or permutation.
	We resort to the protocol in \cite{AsharovHIKNPTT22}, which is the state-of-the-art secret-shared shuffle protocol working with RSS and well suits our purpose.
	
	It is noted that directly secret-sharing a permutation can also lead to $O(d)$ communication cost for each client.
	We introduce a mechanism for permutation compression, allowing the client's communication cost to be independent of $d$.
	As shown in Table~\ref{tab:client-server-communication-cost}, our customized design can even result in about \textasciitilde33\% reduction in client-server communication, compared to the way of directly secret-sharing indices and values as in the strawman approach.

	\subsection{Detailed Construction}
	\label{sec:sparvecagg:details}
	\noindent\textbf{Permutation-based encoding of sparse vectors.}
	%
	For $\mathcal{C}_i$'s sparse vector $\boldsymbol{x}_i$, we can follow Eq. \eqref{eq:re-order} to re-order $\boldsymbol{x}_i$ such that the non-zero elements are placed at the beginning, and the remaining positions are filled with zeros: 
		\begin{equation}
			\label{eq:re-order}
			\boldsymbol{x}_i^{\prime}[j] = 
			\begin{cases} 
				\boldsymbol{x}_i[{L_i(j)}], & \text{if } 0\leq j \leq k-1, \\ 
				0, & \text{if } k\leq j\leq d-1.
			\end{cases}
		\end{equation}
	%
	Here $L_i$ is a list that stores the indices of non-zero values in $\boldsymbol{x}_i$ in order with respect to their occurrence in $\boldsymbol{x}_i$ and $L_i(j)$ denotes the $j$-th element of $L_i$. 
	For example, for a sparse vector $\boldsymbol{x} = (0, x_1, 0, x_3, 0, x_5,$ $0,0,0,0)$, where the non-zero elements are $x_1, x_3, x_5$ (i.e., $L = [1, 3, 5]$ and $k=3$), then the re-ordered vector $\boldsymbol{x}^{\prime}$ is: $\boldsymbol{x}^{\prime} = (x_1, x_3, x_5, 0, 0,0,0,$ $0,0, 0)$. 
	Such re-ordering defines a permutation $\pi$ that maps the re-ordered vector $\boldsymbol{x}^{\prime}$ to the original vector $\boldsymbol{x}$, i.e., $\boldsymbol{x} = \pi (\boldsymbol{x}^{\prime})$.
	Let $E_i$ be a list that stores the indices of zero values in $\boldsymbol{x}_i$ in ascending order.
	In particular, the permutation $\pi_i$ for $\boldsymbol{x}_i$ is derived as follows:  
	%
		\begin{equation}
			\label{eq:comp_pi}
			\pi_i(j) = 
			\begin{cases} 
				L_i(j), & 0\leq j \leq k-1, \\ 
				E_i(j-k), & k\leq j\leq d-1.
			\end{cases}
		\end{equation}	
	With the above encoding, $\mathcal{C}_i$ can then secret-share the two components among $\mathcal{S}_{\{0,1,2\}}$ under RSS: (1) the first $k$ non-zero elements of the re-ordered vector $\boldsymbol{x}^\prime$ (denoted as $\boldsymbol{r}_i = (\boldsymbol{x}_i^\prime[0], \dots, \boldsymbol{x}_i^\prime[k-1])$), and (2) the permutation $\pi_i$ (constructed via Eq.~\eqref{eq:comp_pi}).
	For the $k$ non-zero values, $\mathcal{C}_i$ can just generate and distribute the RSS shares $(\langle \boldsymbol{r}_i \rangle_j, \langle \boldsymbol{r}_i \rangle_{j+1})$ to $\mathcal{S}_j$ for $j \in \{0, 1, 2\}$. 
	Regarding $\pi_i$, its three shares can be generated as follows.
	$\mathcal{C}_i$ first locally samples two random permutations $\pi_{i,0},\pi_{i,1}$, and then computes $\pi_{i,2} = \pi_{i,1}^{-1}\circ\pi_{i,0}^{-1}\circ\pi_i$, which ensures that $\pi_i=\pi_{i,0}\circ\pi_{i,1}\circ\pi_{i,2}$.
	After generating the shares of the permutation $\pi_i$, we then need to consider how to appropriately distribute them to the three servers $\mathcal{S}_{\{0,1,2\}}$.
	
	Directly distributing $(\pi_{i,j}, \pi_{i,j+1})$ to server $\mathcal{S}_j$ for $j \in \{0, 1, 2\}$ would be highly inefficient and incur much more client-server communication overhead than the strawman approach, as demonstrated in Table~\ref{tab:client-server-communication-cost}.
	Specifically, with directly distributing the shares of the non-zero values and permutation, each client would have a communication cost of $6d |\mathbb{Z}_p| + 6k|\mathbb{Z}_p|$ bits ($6k|\mathbb{Z}_p|$ for sharing non-zero elements and $6d|\mathbb{Z}_p|$ for sharing permutations), where $|\mathbb{Z}_p|$ is the bit length of representing an element in the field. 
	In contrast, in the ORAM-based strawman approach, each client only has a communication cost of $12k|\mathbb{Z}_p|$ bits ($6k|\mathbb{Z}_p|$ for sharing non-zero elements and $6k|\mathbb{Z}_p|$ for sharing indices). 
	This raises the problem regarding how to distribute the permutation shares to the servers while retaining communication efficiency.

	\begin{figure}[!t]
		\centering
		\includegraphics[scale=0.16]{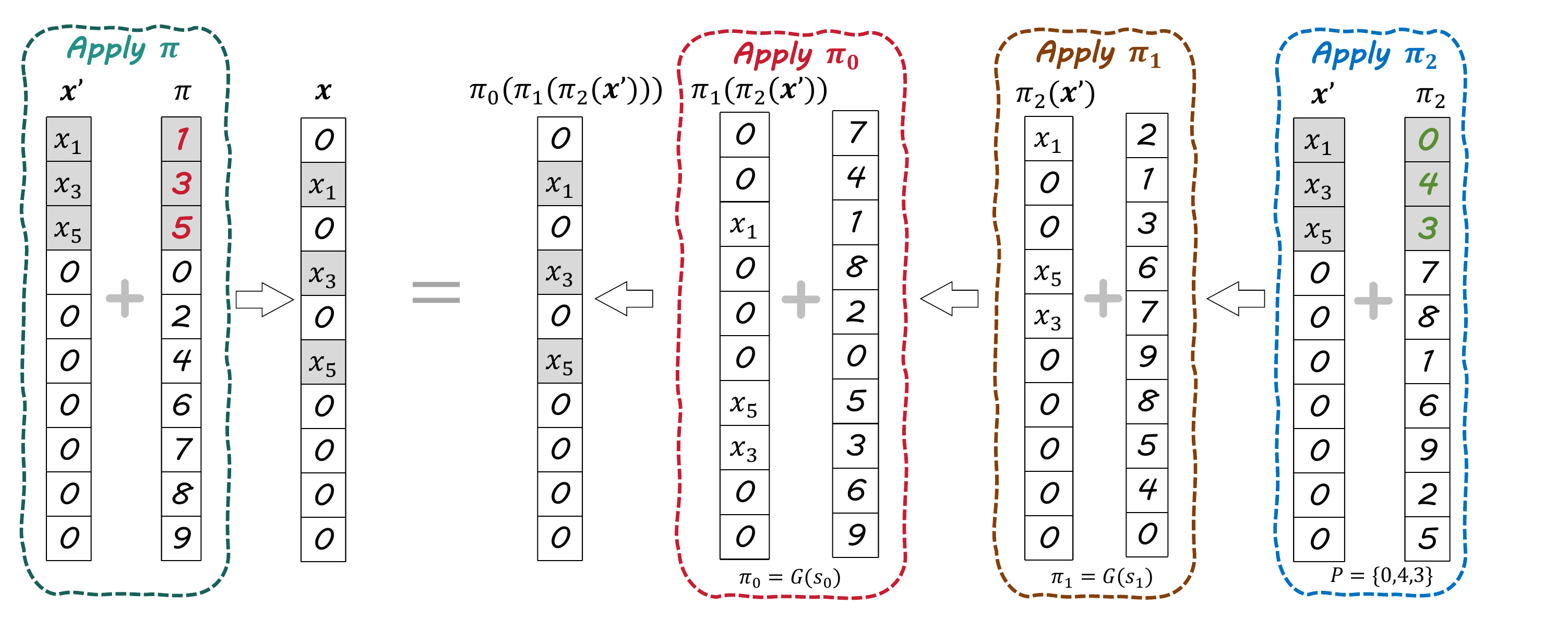}
		\caption{An intuitive example of decomposing the process of applying $\pi$ on $\boldsymbol{x}^{\prime}$ into applying $\pi_2,\pi_1,\pi_0$ sequentially on $\boldsymbol{x}^{\prime}$, where $\pi_1,\pi_0$ are random permutations and $\pi_2 = \pi_1^{-1}\circ\pi_0^{-1}\circ\pi$. }
		\label{fig:decomp_perm}
	\end{figure}
	
	\noindent \textbf{Compressing the permutation for communication efficiency.} 
	Our key idea is to compress the permutation shares $\pi_{i,0}, \pi_{i,1}$ and $\pi_{i,2}$ for $\mathcal{C}_i$.
	Through our compression mechanism, the communication cost for each client can be reduced to $8k |\mathbb{Z}_p| + 4|s|$ bits (where $|s|$ is the length of a small-sized random seed (e.g., 128 bits)), outperforming the strawman approach (about \textasciitilde33\% reduction).
	Our mechanism for the compression of a permutation $\pi_i$ works as follows:

	\begin{itemize}
		\item \textit{Permutation Compression by PRG}: Recall that the two permutation shares $\pi_{i,0}$ and $\pi_{i,1}$ are locally sampled random permutations. 
		So we can compress them into small-sized seeds $\mathsf{s}_{i,0},\mathsf{s}_{i,1}$ such that $\pi_{i,0} \leftarrow G(\mathsf{s}_{i,0}), \pi_{i,1} \leftarrow G(\mathsf{s}_{i,1})$ using a PRG $G$. 
		This allows the servers to reconstruct $\pi_{i,0}$ and $\pi_{i,1}$ using the same seeds and PRG. 
		To facilitate the subsequent presentation, we use $\pi'_{i,0} ,\pi'_{i,1}$ to denote the server-side reconstructed permutations for $\pi_{i,0}$ and $\pi_{i,1}$, respectively. 
		
		\item \textit{Permutation Compression by Distillation}: We make an observation that the third permutation share, i.e., $\pi_{i,2}$, actually only needs to encode information about the mapping of the first $k$ non-zero elements of $\boldsymbol{x}_i^{\prime}$, while the remaining $d-k$ elements of $\pi_{i,2}$ are indeed not useful. 
		Fig.~\ref{fig:decomp_perm} illustrates the decomposition for a generic permutation $\pi$, analogous to each client $\mathcal{C}_i$'s $\pi_i$. 
		The overall permutation $\pi$ is applied to $\boldsymbol{x}'$ sequentially as $\pi = \pi_0 \circ \pi_1 \circ \pi_2$.
		Indeed only the first $k$ elements of $\pi_2$ contribute to the re-arrangement of the first $k$ non-zero elements in $\boldsymbol{x}^{\prime}$, while the remaining $d-k$ elements merely shuffle zeros and thus carry no meaningful information. 
		Hence, our insight for compressing $\pi_{i,2}$ is to have $\mathcal{C}_i$ only upload the first $k$ elements in $\pi_{i,2}$ (denoted as a list $P_{i}$). 
		Note that since $\pi_i=\pi_{i,0}\circ\pi_{i,1}\circ\pi_{i,2}$, the randomly sampled $\pi_{i,0},\pi_{i,1}$ serve as the random masks of $\pi_i$, ensuring that $P_{i}$ does not reveal any private information about $\pi_i$. 
		Let $R_i$ be a list containing elements from the complement set $\{j \in [d] \mid j \notin P_{i}\}$ in ascending order. 
		Upon receiving $P_{i}$, the servers compute $\pi^{\prime}_{i,2}$ by 
			\begin{equation}
				\label{eq:comp_pi2_server}
				\pi^{\prime}_{i,2}(j) = 
				\begin{cases} 
					P_i(j), & 0\leq j \leq k-1, \\ 
					R_i(j-k), & k\leq j\leq d-1.
				\end{cases}
			\end{equation}
		\noindent Although the server-side reconstructed $\pi^{\prime}_{i,2}$ may not be exactly the same as the client-side original $\pi_{i,2}$, it guarantees $\pi^{\prime}_{i,2}(\boldsymbol{x}_i^{\prime}) = \pi_{i,2}(\boldsymbol{x}_i^{\prime})$.  
		This reduces the communication cost of transmitting $\pi_{i,2}$ from $d |\mathbb{Z}_p|$ bits to $k |\mathbb{Z}_p|$ bits.

	\end{itemize}
	
	Under the above compression mechanism, the correctness of reconstructing $\mathcal{C}_i$'s sparse vector at the server side is as follows: 
	{\small
		\begin{equation}
			\begin{aligned} 
				& \langle \pi'_{i,0}(\pi'_{i,1}(\pi'_{i,2}(\boldsymbol{x}'_i))) \rangle_0 + \langle \pi'_{i,0}(\pi'_{i,1}(\pi'_{i,2}(\boldsymbol{x}'_i))) \rangle_1 + \langle \pi'_{i,0}(\pi'_{i,1}(\pi'_{i,2}(\boldsymbol{x}'_i))) \rangle_2 \\ 
				& = \pi'_{i,0}\left(\pi'_{i,1}(\pi'_{i,2}(\boldsymbol{x}'_i))\right) \stackrel{(a)}{=} \pi_{i,0}\left(\pi_{i,1}(\pi'_{i,2}(\boldsymbol{x}'_i))\right) \stackrel{(b)}{=} \pi_{i,0}\left(\pi_{i,1}(\pi_{i,2}(\boldsymbol{x}'_i))\right) \\ 
				& = \pi_{i,0} \circ \pi_{i,1} \circ \pi_{i,2}(\boldsymbol{x}'_i) = \pi_i(\boldsymbol{x}'_i) = \boldsymbol{x}_i, 
			\end{aligned}
			\nonumber
			\vspace{-4pt}
		\end{equation}
	}
	where step (a) holds because $\pi'_{i,0} = \pi_{i,0}$ and $\pi'_{i,1} = \pi_{i,1}$ (both are generated from the same PRG seeds). 
	Step (b) follows because $P_{i}$ contains the destinations of the first $k$ elements of $\boldsymbol{x}'_i$ under $\pi_{i,2}$, ensuring $\pi'_{i,2}(\boldsymbol{x}'_i) = \pi_{i,2}(\boldsymbol{x}'_i)$.
	
	\noindent\textbf{Leveraging secret-shared shuffle.} 
	With server $\mathcal{S}_j$ ($j \in \{0, 1, 2\}$) holding the shares $(\langle \boldsymbol{r}_i \rangle_j, \langle \boldsymbol{r}_i \rangle_{j+1})$ of the non-zero values in $\boldsymbol{x}'_i$ and the reconstructed permutation shares $(\pi'_{i,j}, \pi'_{i,j+1})$, we then resort to the technique of secret-shared shuffle to achieve oblivious permutation of $\boldsymbol{x}'_i$ so as to produce $\llbracket \boldsymbol{x}_i \rrbracket = \llbracket \pi_i(\boldsymbol{x}'_i) \rrbracket$.
	This works as follows. 
	First, $\mathcal{S}_{\{0,1,2\}}$ pad $\llbracket \boldsymbol{r}_i\rrbracket$ with $d-k$ secret-shared zeros, yielding $\llbracket \boldsymbol{x}_i^{\prime}\rrbracket$. 
	This padding step can be accomplished by leveraging PRG-based techniques non-interactively in an offline phase \cite{RSS}. 
	Then, for $j = 1,0,2$, servers $\mathcal{S}_j$ and $\mathcal{S}_{j+1}$ locally apply $\pi^{\prime}_{i,j+1}$ to their respective shares, i.e., $(\langle\boldsymbol{x}_i^{\prime}\rangle_j, \langle\boldsymbol{x}_i^{\prime}\rangle_{j+1})$ and $(\langle\boldsymbol{x}_i^{\prime}\rangle_{j+1}, \langle\boldsymbol{x}_i^{\prime}\rangle_{j-1})$, and re-share the permuted results with the third party $\mathcal{S}_{j-1}$.
	The re-sharing is done in two steps: $\mathcal{S}_j$ and $\mathcal{S}_{j+1}$ first choose $\langle \alpha \rangle_0, \langle \alpha \rangle_1$ and $\langle \alpha \rangle_2$ such that $\langle \alpha \rangle_0+\langle \alpha \rangle_1+\langle \alpha \rangle_2=\mathbf{0}$ and set 
	\begin{equation}
		\notag
		\begin{aligned}
			&\pi^{\prime}_{i,j+1}(\langle\boldsymbol{x}_i^{\prime}\rangle_0) = \pi^{\prime}_{i,j+1}(\langle\boldsymbol{x}_i^{\prime}\rangle_0)+\langle \alpha \rangle_0,\\
			&\pi^{\prime}_{i,j+1}(\langle\boldsymbol{x}_i^{\prime}\rangle_1) = \pi^{\prime}_{i,j+1}(\langle\boldsymbol{x}_i^{\prime}\rangle_1)+\langle \alpha \rangle_1, \text{and}\\
			&\pi^{\prime}_{i,j+1}(\langle\boldsymbol{x}_i^{\prime}\rangle_2) = \pi^{\prime}_{i,j+1}(\langle\boldsymbol{x}_i^{\prime}\rangle_2)+\langle \alpha \rangle_2. 
		\end{aligned}
	\end{equation}
	Then, $\mathcal{S}_j$ sends $\pi^{\prime}_{i,j+1}(\langle\boldsymbol{x}_i^{\prime}\rangle_j)$ to $\mathcal{S}_{j-1}$ and $\mathcal{S}_{j+1}$ sends $\pi^{\prime}_{i,j+1}(\langle\boldsymbol{x}_i^{\prime}\rangle_{j-1})$ to $\mathcal{S}_{j-1}$. 
	The re-sharing procedure guarantees that the third party $\mathcal{S}_{j-1}$ receives random shares, ensuring that $\mathcal{S}_{j-1}$ learns no information about the permutation $\pi^{\prime}_{i,j+1}$. 
	In this way, no single server could view all three permutations $\pi^{\prime}_{i,0},\pi^{\prime}_{i,1},\pi^{\prime}_{i,2}$. 
	Therefore, the permutation $\pi_i$ remains unknown to the servers. 
	
	Once $\pi^{\prime}_{i,0},\pi^{\prime}_{i,1},\pi^{\prime}_{i,2}$ are applied, $\mathcal{S}_{\{0,1,2\}}$ obtain the secret-shared sparse vector $\llbracket \boldsymbol{x}_i \rrbracket=\llbracket \pi_i(\boldsymbol{x}_i^{\prime}) \rrbracket$. 
	For each client, this process is repeated in parallel, and the resulting secret-shared sparse vectors are securely aggregated on the server side.
	Algorithm \ref{alg:SparVecAgg} summarizes our communication-efficient secure sparse vector aggregation protocol.

	\begin{algorithm}[!t]
		\caption{Our Proposed Mechanism \textsf{SparVecAgg} for Communication-Efficient Secure Sparse Vector Aggregation} 
		\label{alg:SparVecAgg}
		\begin{algorithmic}[1]
			\Require Each client $\mathcal{C}_i$ holds a sparse vector $\boldsymbol{x}_i$. 
			\Ensure $\mathcal{S}_{\{0,1,2\}}$ obtain a secret-shared dense vector $\llbracket \boldsymbol{x} \rrbracket$, where $\boldsymbol{x} = \sum_{i \in [n]} \boldsymbol{x}_i$. 
			
			\State $\mathcal{S}_{\{0,1,2\}}$ initialize $\llbracket \boldsymbol{x} \rrbracket \leftarrow \llbracket 0,\cdots,0 \rrbracket$. 
			\For{each client $\mathcal{C}_i$ \textbf{in parallel}}
			\State \underline{\textit{// Client-side computation ($\mathcal{C}_i$):}}
			\State Initialize empty lists $L_i, E_i, P_i$. 
			\For{$j\in[d]$}
			\State $L_i.\mathsf{append}(j)$ if $\boldsymbol{x}_i[j] \neq 0$ else $E_i.\mathsf{append}(j)$.
			\EndFor
			\State $\boldsymbol{x}_i^{\prime} \stackrel{Eq. \eqref{eq:re-order}}{\longleftarrow} \boldsymbol{x}_i,L_i$; \, $\pi_i \stackrel{Eq. \eqref{eq:comp_pi}}{\longleftarrow} L_i,E_i$. 
			\State Locally sample seeds ${\mathsf{s}_{i,0},\mathsf{s}_{i,1}}$. 
			\State \multiline{$\pi_{i,0} \leftarrow G(\mathsf{s}_{i,0})$, $\pi_{i,1} \leftarrow G(\mathsf{s}_{i,1})$. \Comment\emph{{Perm. Compression.}}}
			\State $\pi_{i,2} \leftarrow \pi_{i,1}^{-1} \circ \pi_{i,0}^{-1} \circ \pi_i$.
			\State $P_{i}.\mathsf{append}\left(\pi_{i,2}(j)\right)$ for $j \in [k]$.\Comment\emph{{Perm. Distillation.}} 
			\State $\boldsymbol{r}_i \leftarrow (\boldsymbol{x}_i^\prime[0], \boldsymbol{x}_i^\prime[1], \dots, \boldsymbol{x}_i^\prime[k-1])$
			\State \multiline{Distribute $\mathsf{s}_{i,0}$ and $\mathsf{s}_{i,1}$ to $\mathcal{S}_0$, $\mathsf{s}_{i,1}$ and $P_{i}$ to $\mathcal{S}_1$, $\mathsf{s}_{i,0}$ and $P_{i}$ to $\mathcal{S}_2$. Secret-share $\boldsymbol{r}_i$ across $\mathcal{S}_{\{0,1,2\}}$.}
			\State \underline{\textit{// Server-side computation:}}
			\State $\mathcal{S}_{\{0,2\}}$ locally compute $\pi^{\prime}_{i,0} \leftarrow G(\mathsf{s}_{i,0})$.
			\State $\mathcal{S}_{\{0,1\}}$ locally compute $\pi^{\prime}_{i,1} \leftarrow G(\mathsf{s}_{i,1})$. 
			\State $\mathcal{S}_{\{1,2\}}$ locally initialize an empty list $R_i$. 
			\For{$j\in[d]$}
			\State $R_i.\mathsf{append}(j)$ if {$j \notin P_i$}. \Comment\emph{{Computation on $\mathcal{S}_{\{1,2\}}$.}} 
			\EndFor
			\State $\mathcal{S}_{\{1,2\}}$ locally compute $\pi_{i,2}^{\prime}\stackrel{Eq. \eqref{eq:comp_pi2_server}}{\longleftarrow}P_i,R_i$.
			\State \multiline{$\mathcal{S}_{\{0,1,2\}}$ locally compute $\llbracket \boldsymbol{x}_i^{\prime} \rrbracket \leftarrow \llbracket \boldsymbol{r}_i \rrbracket \| \llbracket 0,\cdots,0 \rrbracket$. \Comment \emph{Pad $d-k$ secret-shared zeros.}}
			\For {$j = 1, 0, 2$} \Comment \emph{Secret-shared shuffle.} \label{code:shuffle-start}
			\State $\mathcal{S}_j,\mathcal{S}_{j+1}$ locally apply $\pi^{\prime}_{i,j+1}$ over the shares of $\llbracket \boldsymbol{x}_i^{\prime} \rrbracket$. 
			\State \multiline{$\mathcal{S}_j, \mathcal{S}_{j+1}$ locally sample $\langle \alpha \rangle_0, \langle \alpha \rangle_1, \langle \alpha \rangle_2$ s.t. $\langle \alpha \rangle_0 + \langle \alpha \rangle_1 + \langle \alpha \rangle_2 = \mathbf{0}$.}
			\State \multiline{$\mathcal{S}_j,\mathcal{S}_{j+1}$ compute: \\$\pi^{\prime}_{i,j+1}(\langle\boldsymbol{x}_i^{\prime}\rangle_0)\leftarrow \pi^{\prime}_{i,j+1}(\langle\boldsymbol{x}_i^{\prime}\rangle_0)+\langle \alpha \rangle_0, $\\$\pi^{\prime}_{i,j+1}(\langle\boldsymbol{x}_i^{\prime}\rangle_1) \leftarrow \pi^{\prime}_{i,j+1}(\langle\boldsymbol{x}_i^{\prime}\rangle_1)+\langle \alpha \rangle_1, $\\$\pi^{\prime}_{i,j+1}(\langle\boldsymbol{x}_i^{\prime}\rangle_2) \leftarrow \pi^{\prime}_{i,j+1}(\langle\boldsymbol{x}_i^{\prime}\rangle_2)+\langle \alpha \rangle_2$. }
			
			\State \multiline{$\mathcal{S}_j$ sends $\pi^{\prime}_{i,j+1}(\langle\boldsymbol{x}_i^{\prime}\rangle_j)$ to $\mathcal{S}_{j-1}$; $\mathcal{S}_{j+1}$ sends $\pi^{\prime}_{i,j+1}(\langle\boldsymbol{x}_i^{\prime}\rangle_{j-1})$ to $\mathcal{S}_{j-1}$.} 
			\EndFor \label{code:shuffle-end}
			\State $\mathcal{S}_{\{0,1,2\}}$ compute $\llbracket \boldsymbol{x} \rrbracket \leftarrow \llbracket \boldsymbol{x} \rrbracket + \llbracket\boldsymbol{x}_i\rrbracket$. \Comment \emph{Aggregation.}
			
			\EndFor
		\end{algorithmic}
	\end{algorithm}

	\section{Basic Design of {\main} with Semi-Honest Security}
	\label{sec:semi-honest-construction}
	
	In this section, we introduce the basic design of {\main} with semi-honest security, which builds on our \textsf{SparVecAgg} protocol to enable efficient secure aggregation for FL harnessing top-$k$ sparsification, and provide DP guarantees for the aggregated gradient update in each round. 
	We will later show how to extend the basic design to achieve malicious security for {\main} in Section \ref{sec:malicious}.

	%

	\begin{algorithm}[!t]
		\caption{The Basic Design of {\main} with Semi-Honest Security} 
		\label{alg:semi-FL}
		\begin{algorithmic}[1]
			\Require Each client $\mathcal{C}_i$ holds a local dataset $\mathcal{D}_i$ for $i\in[n]$.
			\Ensure The servers $\mathcal{S}_{\{0,1,2\}}$ and clients $\mathcal{C}_i$ ($i \in [n]$) obtain a global model $\mathbf{w}$ that satisfies $(\varepsilon,\delta)$-DP. 
			
			\Procedure{Train}{$\{\mathcal{D}_i\}_{i\in[n]}$}
			\State Initialize $ {\mathbf{w}}^0$. 
			\For {$t\in [T]$}\
			\State \multiline{Sample a set $\mathcal{P}^t$ of clients.} 
			\State // \emph{\underline{Client-side computation:}}
			\For{$\mathcal{C}_i \in \mathcal{P}^t$ \textbf{in parallel}}
			\State $\Delta^{t+1}_i \leftarrow \mathsf{LocalUpdate}(\mathcal{D}, \mathbf{w}^t)$. 
			\State $\mathcal{I}_i \leftarrow \text{argtop}_k \big(\mathsf{abs}(\Delta^{t+1}_i)\big).$ \label{alg:semi-FL:argtop}
			
			\State $\boldsymbol{x}^{t+1}_i \leftarrow \mathbf{0}^d$. 
			\State ${\boldsymbol{x}}^{t+1}_i[j]\leftarrow\Delta^{t+1}_i[j]$ for $j\in\mathcal{I}_i$. \Comment\emph{{Sparsification.}} \label{alg:semi-FL:topk-end}
			\State ${\boldsymbol{x}}^{t+1}_i \leftarrow {\boldsymbol{x}}^{t+1}_i / \max \{1, \|{\boldsymbol{x}}^{t+1}_i\|_2/C\}$. \Comment\emph{{Clipping.}} \label{alg:semi-FL:clipping}
			\EndFor
			
			\State // \emph{\underline{Bridge the clients and the servers:}}
			\State $\llbracket \Delta^{t+1} \rrbracket \leftarrow \mathsf{SparVecAgg}(\{{\boldsymbol{x}}^{t+1}_i\}_{\mathcal{C}_i\in \mathcal{P}^t})$.
			\State // \emph{\underline{Server-side computation:}}
			\State $\llbracket {\Delta}^{t+1} \rrbracket \leftarrow \mathsf{SecNoiseAdd}(\llbracket \Delta^{t+1} \rrbracket, \sigma, C)$.
			\State \multiline{$ \Delta^{t+1} \leftarrow \mathsf{Rec}(\llbracket \Delta^{t+1} \rrbracket)$. } 
			\State $\mathbf{w}^{t+1} \leftarrow \mathbf{w}^{t} + \Delta^{t+1}/|\mathcal{P}^t|$. 
			\EndFor
			\EndProcedure
			
			\Procedure{\textsf{LocalUpdate}}{$\mathcal{D}_i, {\mathbf{w}}^t$}
			\State $\mathbf{w}^0_i \leftarrow {\mathbf{w}}^t$.
			\State $\mathcal{B}_i\leftarrow \mathcal{D}_i$. \Comment{\emph{Split the dataset $\mathcal{D}_i$ into batches.}}
			\For{each local epoch $e \in [E]$}
			\For{each batch $\beta_i \in \mathcal{B}_i$}
			\State $\mathbf{w}^{e+1}_i \leftarrow \mathbf{w}^e_i - \eta_l \nabla F_i \left(\mathbf{w}^e_i , \beta_i \right)$. 
			\EndFor
			\EndFor
			\State $\Delta^{t+1}_i \leftarrow \mathbf{w}^{E}_i-\mathbf{w}^0_i$.
			\State \textbf{Return} $\Delta^{t+1}_i$. 
			\EndProcedure
		\end{algorithmic}
		
	\end{algorithm}


	\noindent\textbf{Overview.} Recall that in each training round of FL, a selected client locally trains a model and produces a gradient update. 
	Instead of directly transmitting the original gradient update (treated as a dense vector), each selected client can first locally sparsify this gradient update with the top-$k$ sparsification mechanism. 
	Then the clients and the servers can proceed to collaboratively run our proposed \textsf{SparVecAgg} protocol to produce an aggregated gradient update in secret-shared form at the server side. 
	Then we consider how to securely add calibrated DP noise to the aggregated gradient update so that the servers can only obtain a noisy version of the aggregated gradient update to update the current global model.
	One possible way here is to have the servers collaboratively execute secure noise sampling protocols \cite{FuW24,WeiYFCW23,ARES24}. 
	These protocols can directly produce a noise secret-shared among the servers while keeping the noise value hidden from individual servers, but would incur prohibitively high performance overheads for applications like distributed model training as shown in \cite{PETS25}. 
	Aiming at secure and efficient noise addition, we instead follow an alternative strategy---distributed noise generation \cite{kairouz21,Canonne0S20,google2020secure,PETS25}, where a random noise is locally sampled by each server from the appropriate discrete Gaussian distribution and secret-shared among the servers. 
	Subsequently, the noises from all servers are efficiently added to the aggregated gradient update in the secret sharing domain.
	In this way, DP guarantees can be achieved for the aggregated gradient update (and thus the global model) with high efficiency. 
	\revise{
		Note that the strategy of distributed noise addition is also applied in prior works \cite{kairouz21,Canonne0S20,google2020secure,PETS25}.  Our novelty lies in newly designing a lightweight noise verification mechanism, as will be presented in Section \ref{sec:malicious}.
	}

	\noindent\textbf{Construction.} Algorithm \ref{alg:semi-FL} gives the construction of the basic design of {\main}.
	At the beginning, servers $\mathcal{S}_{\{0,1,2\}}$ agree upon an initialized global model $\mathbf{w}^0$. 
	The $t$-th FL round starts with $\mathcal{S}_{\{0,1,2\}}$ selecting a common subset of clients (denoted as $\mathcal{P}^t$) and broadcasting the current global model $\mathbf{w}^t$ to them. 
	Each selected client $\mathcal{C}_i$ then performs local training, which is denoted by \textsf{LocalUpdate}. 
	To sparsify the produced gradient update $\Delta^{t+1}_i$ for reducing communication costs, $\mathcal{C}_i$ begins by identifying the set $\mathcal{I}_i$ containing indices associated with the top-$k$ largest absolute values of $\Delta^{t+1}_i$. 
	Next, $\mathcal{C}_i$ initializes a vector $\boldsymbol{x}_i^{t+1}$ of size $d$, filled with zeros. 
	For each index $j \in \mathcal{I}_i$, the corresponding value from $\Delta_i^{t+1}$ is then assigned to $\boldsymbol{x}_i^{t+1}[j]$, producing a top-$k$ sparsified gradient update $\boldsymbol{x}_i^{t+1}$. 
	Finally, $\mathcal{C}_i$ clips $\boldsymbol{x}_i^{t+1}$ to bound its $\ell_2$-norm, facilitating the use of the Gaussian mechanism to achieve DP. 
	
	With the top-$k$ sparsified gradient updates $\{\boldsymbol{x}^{t+1}_i\}_{\mathcal{C}_i\in \mathcal{P}^t}$ computed for the $t$-th round of FL, our proposed \textsf{SparVecAgg} mechanism can be employed to efficiently and securely bridge the clients and the servers, with the input of each client $\mathcal{C}_i \in \mathcal{P}^t$ being its top-$k$ sparsified gradient update vector. 
	From the execution of \textsf{SparVecAgg}, the servers produce the secret-shared aggregated gradient update $\llbracket \Delta^{t+1} \rrbracket$. 
	Then the servers securely inject DP noise to $\llbracket \Delta^{t+1} \rrbracket$ through \textsf{SecNoiseAdd}, which works as follows.

	First, each server $\mathcal{S}_i$ ($i\in\{0,1,2\}$) samples discrete Gaussian noise $\eta_i$ from the discrete Gaussian distribution $\mathcal{N}_{\mathbb{Z}}(0, \frac{1}{2}\sigma^2C^2\mathbf{I}_d)$ . 
	Then $\mathcal{S}_i$ ($i\in\{0,1,2\}$) secret-shares its noise $\eta_i$ among all servers. 
	Subsequently, the servers securely compute the summation of the noise and the aggregated gradient update via: $\llbracket \Delta^{t+1} \rrbracket \leftarrow \llbracket \Delta^{t+1} \rrbracket + \sum_{i=0}^{2}\llbracket \eta_i \rrbracket $. 
	So the output of \textsf{SecNoiseAdd} is the secret-shared noisy aggregated gradient update $\llbracket \Delta^{t+1} \rrbracket$.
	Note that the discrete Gaussian noise sampled independently by each server ensures that even if a server attempts to subtract its own noise contribution, the resulting noise-subtracted output still retains noise distributed as $\mathcal{N}_{\mathbb{Z}}(0, \sigma^2 C^2 \mathbf{I}_d)$. 
	This meets the requirement of the Gaussian mechanism (Definition~\ref{thm:gaussian}) and ensures DP. 
	At the end of each training round, $\mathcal{S}_{\{0,1,2\}}$ reconstruct $\Delta^{t+1}$  and integrate it to update the global model via: $\mathbf{w}^{t+1} = \mathbf{w}^{t} + \Delta^{t+1}/|\mathcal{P}^t|$. 
	Algorithm \ref{alg:semi-FL} satisfies $(\varepsilon,\delta)$-DP, where the privacy analysis is provided in Appendix \ref{appendix:proof:privacy}.

	\section{Achieving Malicious Security for {\main}}
	\label{sec:malicious}
	In this section, we extend our basic design to achieve malicious security.
	With respect to the basic design with semi-honest security in Algorithm \ref{alg:semi-FL}, achieving security against a maliciously acting server \emph{essentially} needs to ensure that: (1) the permutation shares are correctly decompressed and applied for shuffling, (2) the secret-shared discrete Gaussian noise is correctly sampled by each server, and (3) the secret-shared noises are correctly aggregated with individual client's gradient updates. 
	%
	%
	In what follows, we elaborate on how to ensure each of these integrity guarantees. 
	And we give the complete maliciously secure protocol of {\main} in Algorithm~\ref{alg:mali-FL} in Appendix~\ref{sec:complete-protocol}, with security proofs provided in Appendix~\ref{sec:security-analysis}.

	\subsection{Maliciously Secure Decompression and Shuffle}
	\label{sec:malicious:shuffle}
	
	Recall that for each selected client, the servers first decompress the permutation shares, and then obliviously apply the permutation to the (re-ordered) gradient update based on secret-shared shuffle. 
	For integrity verification, it is necessary to verify that: (1) the permutation shares are correctly decompressed, and (2) the secret-shared shuffle is correctly performed. 
	We observe that verifying the integrity of the \emph{final output} of the shuffle process, $\llbracket \boldsymbol{x}_i \rrbracket = \llbracket \pi'_i(\boldsymbol{x}'_i) \rrbracket$, is sufficient to detect errors in either phase. 
	This is because the application of each permutation $\pi'_{i,j}$ in Algorithm \ref{alg:SparVecAgg} requires computations involving a specific pair of servers. 
	If one of these servers is malicious and uses an incorrect permutation (e.g., due to faulty decompression), while the other server is honest and applies the correct one, the permutation is applied inconsistently across their respective shares.
	This inconsistency causes the shares to be misaligned, essentially leading to an incorrect final secret-shared output that would cause the integrity check to fail. 
	%

	A na\"ive approach to verifying the integrity of the shuffle output is to attach an information-theoretic MAC for each value of the gradient update (treated as a vector), which can then be used to detect any malicious behavior \cite{AsharovHIKNPTT22}. 
	However, this incurs the cost of two secure dot product computations per shuffle for verification \cite{AsharovHIKNPTT22}. 
	To improve verification efficiency, we propose a tailored mechanism built on \textit{blind MAC verification}, which requires only one secure dot product computation per shuffle. 

	At a high level, each client locally generates a MAC for its re-ordered sparse gradient update, then secret-shares the MAC value and the MAC key; the MAC key is securely shuffled alongside the gradient update. 
	The servers perform verification by computing the expected MAC through a secure dot product on the shuffled key and gradient update shares. 
	Finally, the computed MACs of all gradient updates are aggregated and verified together in the secret-sharing domain.
	It is noted that while recent works \cite{NDSS22, camel} also adopt a similar strategy of blind MAC verification for secure shuffle, they focus on additive secret sharing and two-party protocols based on \cite{ChaseGP20}. 
	In contrast, our approach is tailored for the RSS-based three-party secure shuffle protocol in \cite{AsharovHIKNPTT22} and provides a tailored mechanism for achieving malicious security.

	\noindent\textbf{Construction.} Our blind MAC verification scheme leverages the Carter-Wegman one-time MAC \cite{WegmanC81}. 
	For a re-ordered sparse gradient update $\boldsymbol{x}^{\prime}_i$ at client $\mathcal{C}_i$, the MAC is computed as
	%
		\begin{equation}
			\label{eq:MAC}
			t_i = \sum_{j=0}^{d-1} \boldsymbol{k}_i[j] \cdot \boldsymbol{x}^{\prime}_i[j], 
		\end{equation}
	\noindent where $\boldsymbol{k}_i\in\mathbb{Z}^d_p$ is a MAC key and $\boldsymbol{k}_i[j]$ denotes the $j$-th element of $\boldsymbol{k}_i$. 
	The security of the MAC follows directly from the DeMillo-Lipton-Schwartz-Zippel lemma \cite{Zippel79,Schwartz80,DemilloL78}. 
	Instead of transmitting the full MAC key, $\mathcal{C}_i$ generates $\boldsymbol{k}_i$ from three key seeds $\mathsf{ks}_{i,0}, \mathsf{ks}_{i,1}, \mathsf{ks}_{i,2}$ using a PRG $G$: 
	\begin{equation}
		\notag
		\begin{aligned}
			\boldsymbol{k}_i  = (\boldsymbol{k}_i[0], \cdots, \boldsymbol{k}_i[d-1]) 
			= {G}(\mathsf{ks}_{i,0}) + {G}(\mathsf{ks}_{i,1}) + {G}(\mathsf{ks}_{i,2}). 
		\end{aligned}
	\end{equation}
	This enables the client to transmit only the small key seeds for communication efficiency. 
	%
	%
	Then, $\mathcal{C}_i$ distributes $(\langle t_i \rangle_j, \langle \boldsymbol{r}_i \rangle_j, \mathsf{ks}_{i,j})$ and $(\langle t_{i} \rangle_{j+1}, \langle \boldsymbol{r}_{i} \rangle_{j+1}, \mathsf{ks}_{{i},j+1})$ to server $\mathcal{S}_j$ ($j\in\{0,1,2\}$). 

	Upon receiving shares, $\mathcal{S}_j$ ($j\in\{0,1,2\}$) locally runs $\langle \boldsymbol{k}_i \rangle_{j} = G(\mathsf{ks}_{i,j})$ and $\langle \boldsymbol{k}_i \rangle_{j+1} = G(\mathsf{ks}_{i,j+1})$ to expand the key shares.
	$\mathcal{S}_{\{0,1,2\}}$ also pad $\llbracket \boldsymbol{r}_i\rrbracket$ with $d-k$ secret-shared zeros to get $\llbracket \boldsymbol{x^{\prime}}_i\rrbracket$. 
	As a result, $\mathcal{S}_{\{0,1,2\}}$ hold a secret-shared tuple $(\llbracket t_i\rrbracket, \llbracket \boldsymbol{x^{\prime}}_i \rrbracket, \llbracket \boldsymbol{k}_i\rrbracket)$. 
	Then the servers follow the secret-shared shuffle process in Algorithm \ref{alg:SparVecAgg} (lines~\ref{code:shuffle-start}-\ref{code:shuffle-end}) to apply the secret-shared permutation $\llangle \pi_i \rrangle$ to both the re-ordered gradient update $\llbracket \boldsymbol{x}^{\prime}_i \rrbracket$ and the MAC key $\llbracket \boldsymbol{k}_i \rrbracket$. 
	To verify integrity, the servers compute the expected MAC $\llbracket {t}^\prime_i \rrbracket$ as the secure dot product of the shuffled gradient update $\llbracket \pi_i(\boldsymbol{x}^{\prime}_i) \rrbracket$ and the shuffled key $\llbracket \pi_i(\boldsymbol{k}_i) \rrbracket$.
	%
	%
	Subsequently, the servers compute the difference $\llbracket f_i \rrbracket = \llbracket {t}_i \rrbracket - \llbracket {t}^\prime_i \rrbracket$.
	Any tampering by a malicious server during this process results in $f_i \neq 0$.

	Finally, the servers verify all MACs in a batch by computing and reconstructing $ f = \mathsf{Rec}\left( \llbracket r \rrbracket\cdot(\sum_{\mathcal{C}_i\in\mathcal{P}^t} \llbracket f_i \rrbracket) \right)$, where $\llbracket r \rrbracket$ is a secret-shared random value sampled by the servers to prevent a malicious server from learning anything from a non-zero sum. 
	If $f=0$, the servers accept the MAC, confirming no malicious behavior in the process. 
	Otherwise, the honest servers output abort and stop. 
	In this way, all the MACs are blindly verified by the servers as one batch in the secret sharing domain, without revealing any particular shuffled value or MAC. 
	To prevent a malicious server from manipulating its share to force $f$ to be zero, the servers could compare the reconstructed results and output abort if inconsistency is detected. 
	Details regarding the above process of integrity check on the permutation decompression and secret-shared shuffle are summarized in Algorithm~\ref{alg:mali-FL} (lines~\ref{code:malicious-shuffle-start}-\ref{code:malicious-shuffle-end}) in Appendix~\ref{sec:complete-protocol}.

	
	\vspace{-5pt}

	\revise{\subsection{Efficient and Secure Noise Sampling with Verifiability}}
	\label{sec:malicious:noise-sampling}
	
	
	\noindent Recall that our secure noise addition protocol \textsf{SecNoiseAdd} requires each server to locally sample and secret-share noise. 
	However, under a malicious adversary, a compromised server may sample incorrect noise (e.g., with excessively large magnitude \cite{FuW24}) which can severely degrade model utility. 
	Ensuring the sampled noise follows the correct distribution is thus critical. 
	While existing works \cite{WeiYFCW23,FuW24} support secure sampling of discrete Gaussian noise against both semi-honest and malicious adversaries, as previously discussed, their performance overhead is prohibitively high for practical applications like FL---even in the semi-honest adversary setting. 
	
	To address this challenge, we propose a tailored lightweight verifiable noise sampling mechanism. 
	Our mechanism leverages observation on a fundamental property of discrete Gaussian distributions: the sum of two independent discrete Gaussian samples $\eta, \xi \sim \mathcal{N}_\mathbb{Z}(0, \frac{1}{2}\sigma^2C^2\mathbf{I}_d)$ is proven to be extremely close to another discrete Gaussian sample $\kappa \sim \mathcal{N}_\mathbb{Z}(0, \sigma^2C^2\mathbf{I}_d)$ \cite{kairouz21,FuW24}.
	We leverage this property to transform the traditional costly MPC-based noise verification process into a lightweight local verification process using the Kolmogorov-Smirnov (KS) test \cite{ks-test}. 
	At a high level, to verify whether a server's sampled noise conforms to the distribution $\mathcal{N}_\mathbb{Z}(0, \frac{1}{2}\sigma^2 C^2 \mathbf{I}_d)$, another server samples noise from the same distribution and then secret-shares it. 
	The two secret-shared noises are added up, and the third server reconstructs the sum. 
	The KS test is then applied locally to efficiently verify if the resulting noise conforms to the distribution $\mathcal{N}_\mathbb{Z}(0, \sigma^2 C^2 \mathbf{I}_d)$, while ensuring the privacy of the raw noise values.

	\noindent\textbf{Construction.} 
	We start with verifying the discrete Gaussian noise $\eta_0 $ sampled by server $\mathcal{S}_0$, where $\eta_0 \sim \mathcal{N}_{\mathbb{Z}}(0, \frac{1}{2}\sigma^2 C^2 \mathbf{I}_d)$. 
	The same procedure applies to noise verification for other servers ($\mathcal{S}_1$ and $\mathcal{S}_2$). 
	To validate whether $\eta_0$ conforms to $\mathcal{N}_{\mathbb{Z}}(0, \frac{1}{2}\sigma^2 C^2 \mathbf{I}_d)$, another server (either $\mathcal{S}_1$ or $\mathcal{S}_2$) first independently samples masking noise $\xi \sim \mathcal{N}_{\mathbb{Z}}(0, \frac{1}{2}\sigma^2 C^2 \mathbf{I}_d)$, and then secret-shares it across $\mathcal{S}_{\{0,1,2\}}$. 
	This masking noise preserves the distributional properties while ensuring $\eta_0 + \xi$ achieves statistical indistinguishability---disclosing $\kappa = \eta_0 + \xi$ reveals no information about $\eta_0$.
	Without loss of generality, we designate $\mathcal{S}_1$ as the masking noise generator when verifying $\mathcal{S}_0$'s samples. 
	The servers then securely compute $\llbracket \kappa \rrbracket = \llbracket \eta_0 \rrbracket + \llbracket \xi \rrbracket$ via secure addition. 
	Finally, $\mathcal{S}_0$ and $\mathcal{S}_1$ send $\langle \kappa \rangle_1$ to $\mathcal{S}_2$ to reconstruct $\kappa$ on $\mathcal{S}_2$. 
	$\mathcal{S}_2$ first compares the $\langle \kappa \rangle_1$ sent from $\mathcal{S}_0$ and $\mathcal{S}_1$. 
	If inconsistency is detected, $\mathcal{S}_2$ outputs abort.
	Next, $\mathcal{S}_2$ performs distributional verification through a two-sample KS test between the masked noise $\kappa$ and reference sample $\kappa' \sim \mathcal{N}_{\mathbb{Z}}(0, \sigma^2C^2\mathbf{I}_d)$. The empirical CDFs are computed as:
	$F_{\kappa}(y) = \frac{1}{d}\sum_{i=1}^d \mathbb{I}(\kappa[i] \leq y), 
	F_{\kappa'}(y) = \frac{1}{d}\sum_{j=1}^d \mathbb{I}(\kappa'[j] \leq y),$ 
	where $\mathbb{I}(\cdot)$ denotes the indicator function. The KS statistic calculates the maximum divergence: $D_{\text{KS}} = \sup_{y} |F_{\kappa}(y) - F_{\kappa'}(y)|$. 
	Under the null hypothesis that \(\kappa\) and \(\kappa'\) are drawn from the same distribution, the critical value at significance level $\alpha$ is:
	$D_{\text{crit}} = \sqrt{-\frac{1}{2} \ln\left(\frac{\alpha}{2}\right)} \cdot \sqrt{\frac{2}{d}}$. 
	In this paper, we adopt a widely used significance level of \(\alpha = 0.05\) \cite{FuW24,Xiang0L023}. 

	\revise{
		If the KS statistic exceeds the critical value, \(\mathcal{S}_2\) rejects the null hypothesis, indicating one of the following malicious behaviors:  
		(1) $\mathcal{S}_0$ sampled incorrect $\eta_0$,  
		(2) $\mathcal{S}_1$ sampled incorrect $\xi$, or  
		(3) $\mathcal{S}_2$ manipulated the test.  
		Note that there exists a possibility that $\eta_0$ or $\xi$ is manipulated a bit yet still passes the KS test. 
		However, this does not matter in practice because such manipulation is constrained to be insignificant by the KS test. 
		That is, a significant manipulation would cause the KS test to fail, implying that $\eta_0$ or $\xi$ is heavily deviating/not drawn from the target distribution, resulting in abort. 
		Full protocol details are given in Algorithm~\ref{alg:malicious-noise-sampling} in Appendix \ref{sec:complete-protocol}. 
	}
	
	\noindent \revise{\textbf{Remark. }
		There exist some DP mechanisms that adopt noise distributions with bounded support (e.g., the truncated Laplacian mechanism~\cite{GengDGK20}), which could potentially limit the influence of a malicious server during noise sampling. 
		However, we note that using such distributions does not directly enable a simpler version of verifiable noise sampling.
		This is because it just mitigates malicious influence and does not address the verification of correct noise sampling. 
		Indeed, how to securely and efficiently verify that the noise conforms to the target distribution is a key challenge tackled in this paper. 
	}
	
	\subsection{Maliciously Secure Aggregation of Noises and Gradient Updates}
	\label{sec:malicious:aggregation}
	%
	
	So far we have introduced how to achieve efficient integrity check on permutation decompression and secret-shared shuffle (Section \ref{sec:malicious:shuffle}) and on DP noise sampling (Section \ref{sec:malicious:noise-sampling}).
	%
	%
	%
	Building on these guarantees, the final step towards achieving security against a malicious server is to ensure the correct computation of the final aggregate from the verified gradient updates and noise shares.

	The aggregation phase only admits a malicious server to inject an additive error \cite{ChidaGHIKLN18} into the secret-shared values during local addition or reconstruction \cite{ELSA}. 
	Such an attack can compromise the correctness of the final aggregated model, thereby affecting utility. 
	%
	%
	Specifically, consider that a malicious server $\mathcal{S}_j$ ($j\in\{0,1,2\}$) introduces an additive error $\mathbf{e}$ during the local addition phase: $
	\langle \Delta^{t+1} \rangle_j = \sum_{\mathcal{C}_i\in\mathcal{P}^t} \langle \boldsymbol{x}_i^{t+1} \rangle_j + \sum_{i=0}^{2}\langle \eta_i \rangle_j + \boldsymbol{e}$. 
	This error would propagate to the subsequent reconstruction phase. 
	To detect the additive error attack, we let servers $\mathcal{S}_{\{0,1,2\}}$ reconstruct $\Delta^{t+1}$, update the global model via: $\mathbf{w}^{t+1} = \mathbf{w}^{t} + {\Delta^{t+1}}/{|\mathcal{P}^t|}$, and then compute hash of $\mathbf{w}^{t+1}$ locally.
	Next, each server $\mathcal{S}_j$ exchanges its hash with the others and checks consistency across all three hashes. 
	If any inconsistency is detected, the protocol aborts.
	The security derives from the replication of RSS: 
	A malicious server $\mathcal{S}_j$ can only tamper with its own shares $\langle \Delta^{t+1} \rangle_j$ and $\langle \Delta^{t+1} \rangle_{j+1}$. 
	Since $\mathcal{S}_{j+1}$ reconstructs $\Delta^{t+1}$ using $\langle \Delta^{t+1} \rangle_{j}$ from $\mathcal{S}_{j-1}$ (not $\mathcal{S}_j$), any additive error introduced by $\mathcal{S}_j$ will not affect $\mathcal{S}_{j+1}$'s hash. 
	This mismatch guarantees the detection of malicious deviations.
	
	\vspace{-10pt}
	
	\revise{
		\subsection{Remark}
		\label{sec:malicious_remark} 
	}
	\noindent \revise{
		For clarity, we now summarize the malicious deviations for {\main} in the presence of a malicious server.  
		Specifically, the abort is triggered upon:
		(1) MAC verification failure, which occurs if a server incorrectly decompresses a permutation or performs the secret-shared shuffle protocol; 
		(2) failure of the KS test, which happens if the computed KS statistic exceeds its critical value, indicating the distribution of the sampled noise significantly deviates from the target discrete Gaussian distribution; 
		and (3) mismatch of hashes, which occurs if a server introduces an additive error in the aggregation phase, causing the hashes of the final updated model, computed independently by each server, to differ.
	}
	
	\revise{
		Regarding the sum of sparsified gradient updates in {\main}, it is noted that a malicious server cannot manipulate it due to enforcement of blind MAC verification. 
		For the noisy aggregated gradient update, a malicious server could potentially introduce a bias in the noise part, yet it is constrained to be insignificant by the KS test. 
	}

	\section{Theoretical Analysis of Privacy, Communication, and Convergence}
	\label{sec:theoretical-results}
	%
	Theorem \ref{thm:analysis_overview} gives the key analytical results regarding the DP guarantees, communication complexity, and convergence of {\main}. 
	We provide the proof for Theorem \ref{thm:analysis_overview} in Appendix \ref{appendix:analysis_overview}. 
	
	\begin{thm}
		\label{thm:analysis_overview} For sampling rate $q=|\mathcal{P}^t|/n$, if we run {\main} over $T$ rounds, then we have:
		\begin{itemize}
			\item \textbf{Privacy}: For any $\delta \in (0, 1)$ and $\varepsilon < 2 \log(1/\delta)$, {\main} guarantees ($\varepsilon, \delta$)-DP after $T$ training rounds if
			$$
			\sigma^2 \geq \frac{14q^2T\log(1/\delta)}{\varepsilon^2} + \frac{7q^2T}{\varepsilon}.
			$$
			
			\item \textbf{Communication}: 
			With malicious security, the communication cost per client per round is $(8k+6) |\mathbb{Z}_p| + 10|s|$ bits under local top-$k$ sparsification, where $|\mathbb{Z}_p|$ is the bit-length of a field element and $|s|$ is the bit-length of a PRG seed.
			The total inter-server communication complexity per round is $\mathcal{O}(|\mathcal{P}^t| \cdot d)$.

			\item \textbf{Convergence}: 
			%
			Under standard assumptions (see Appendix~\ref{appendix:proof:convergence}), the average expected squared gradient update norm of the global loss function over \( T \) training rounds is bounded by \scalebox{0.9}{$\frac{1}{T} \sum_{t=0}^{T-1} \mathbb{E} \|\nabla F(\mathbf{w}^t)\|^2 \leq \mathcal{O} ( \frac{1}{\eta_l E T} + \frac{L \eta_l E (1 + \phi) (G^2 + \sigma_l^2)}{qn}) +  \mathcal{O} \left( \frac{d \sigma^2 C^2}{q^2 n^2} \right)$}.

			%

		\end{itemize} 
	\end{thm}

	\begin{figure}[!t]
		
		\centering
		
		\begin{minipage}[t]{0.45\linewidth}
			\centering
			\includegraphics[width=\linewidth]{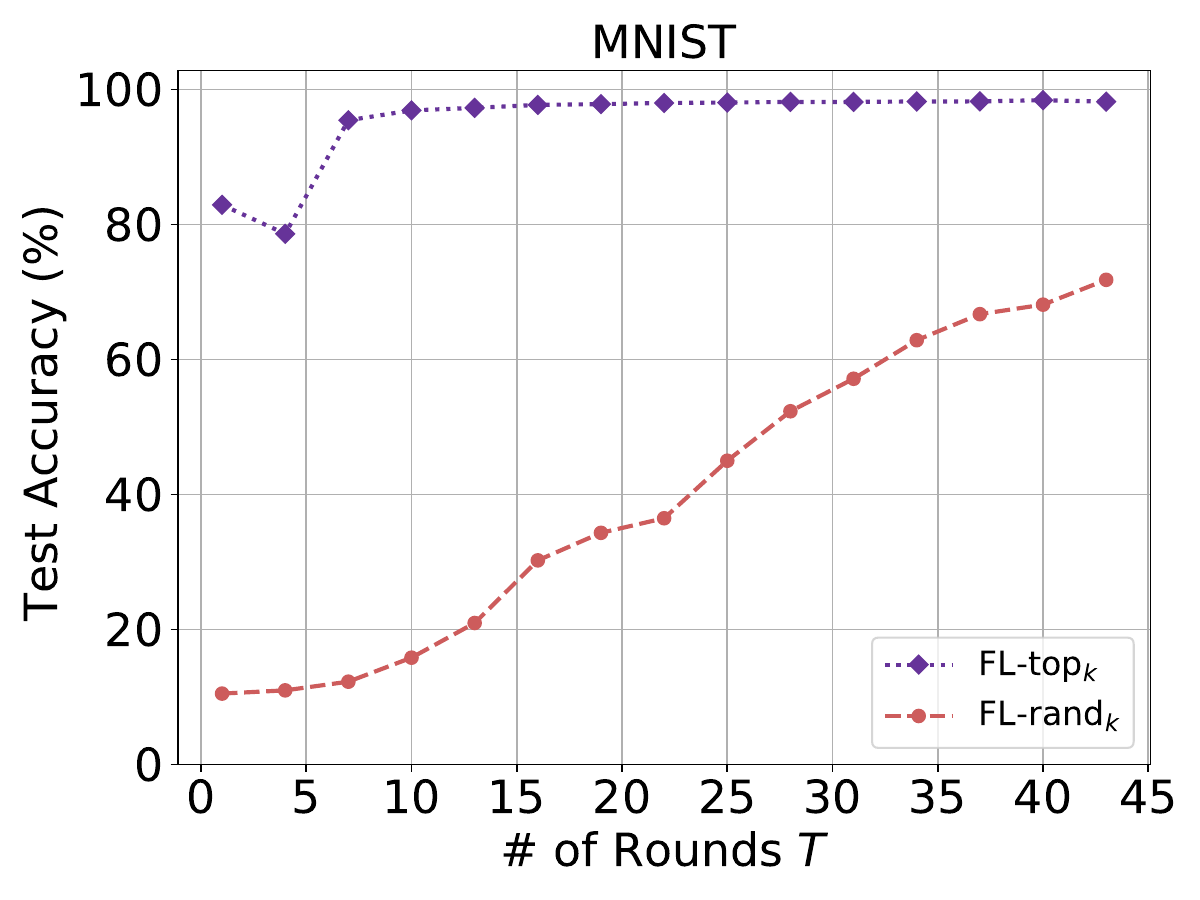}\\(a)
		\end{minipage}
		\begin{minipage}[t]{0.45\linewidth}
			\centering
			\includegraphics[width=\linewidth]{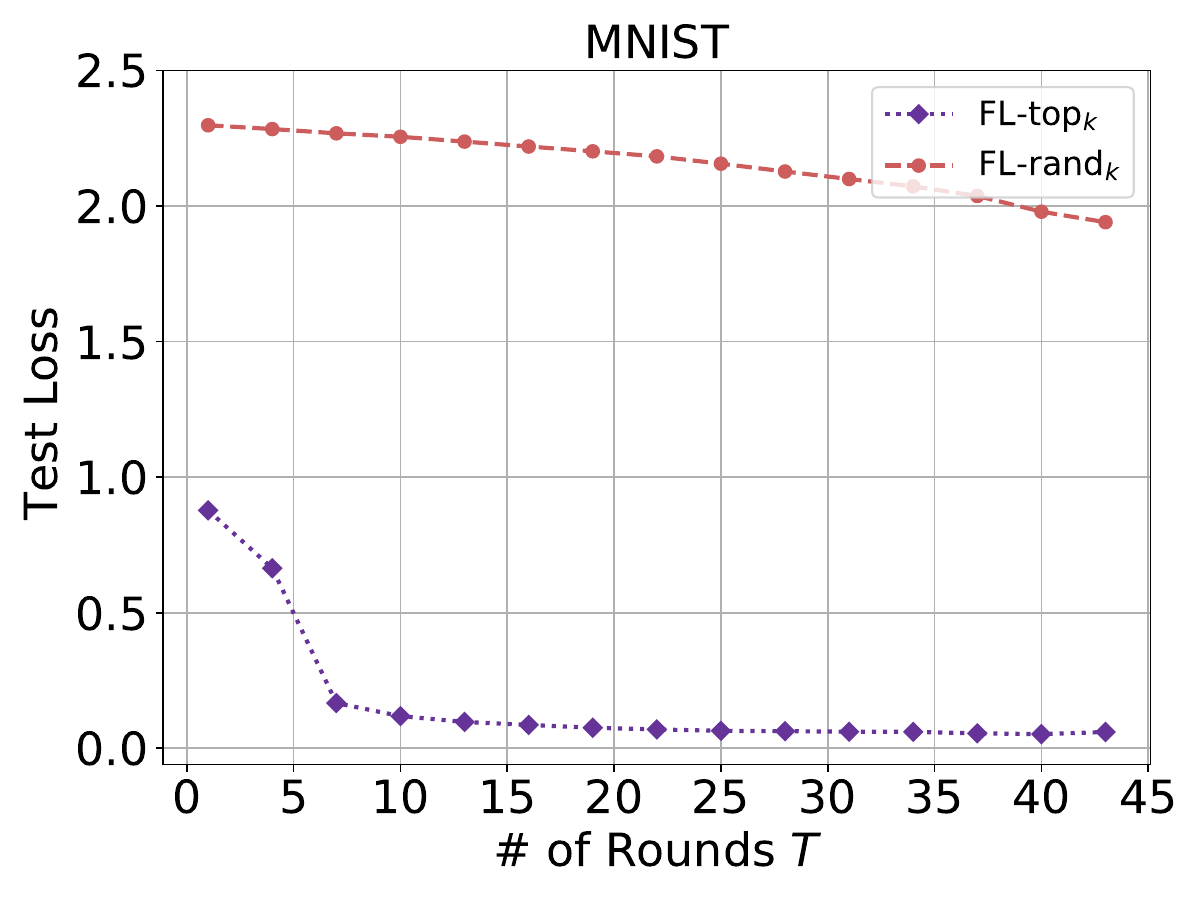}\\(b)
		\end{minipage}
		\\
		\begin{minipage}[t]{0.45\linewidth}
			\centering
			\includegraphics[width=\linewidth]{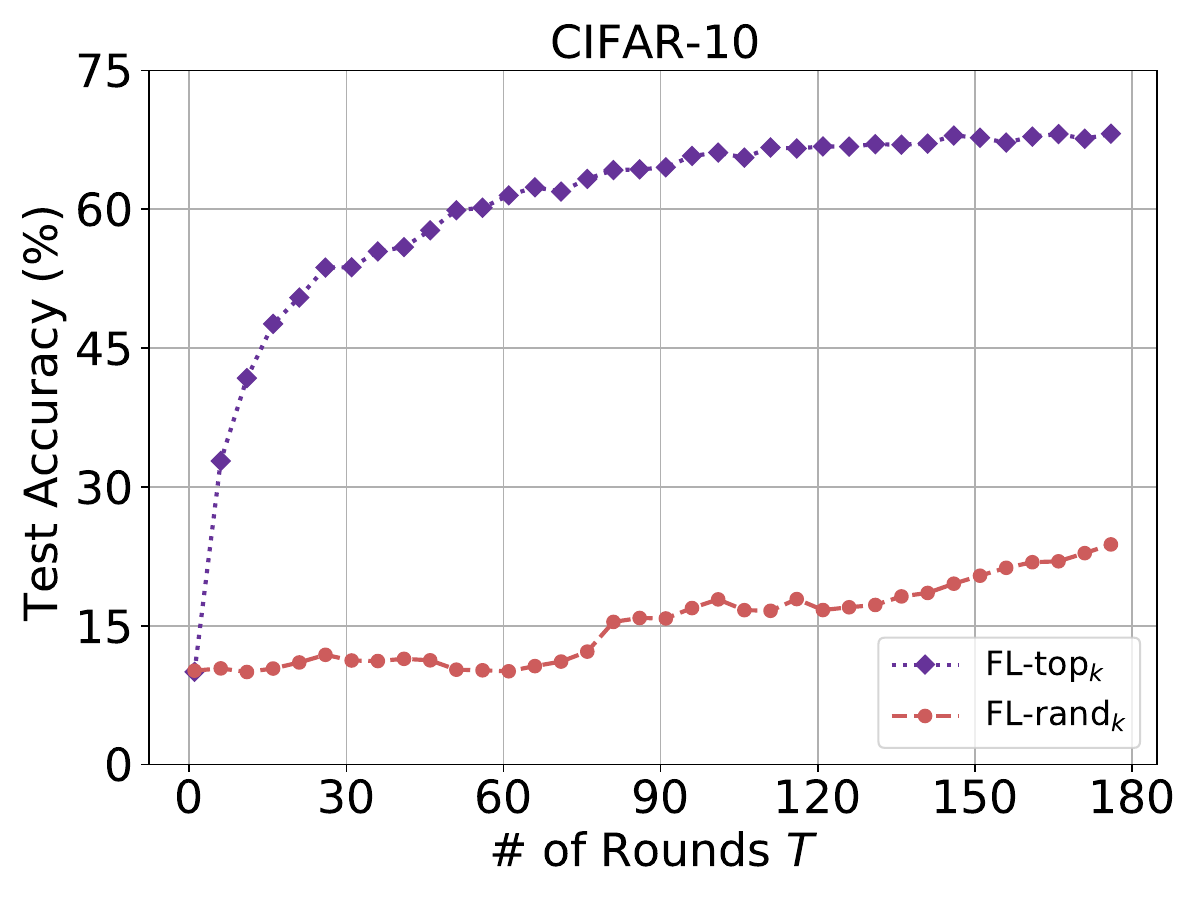}\\(c)
		\end{minipage}
		\begin{minipage}[t]{0.45\linewidth}
			\centering
			\includegraphics[width=\linewidth]{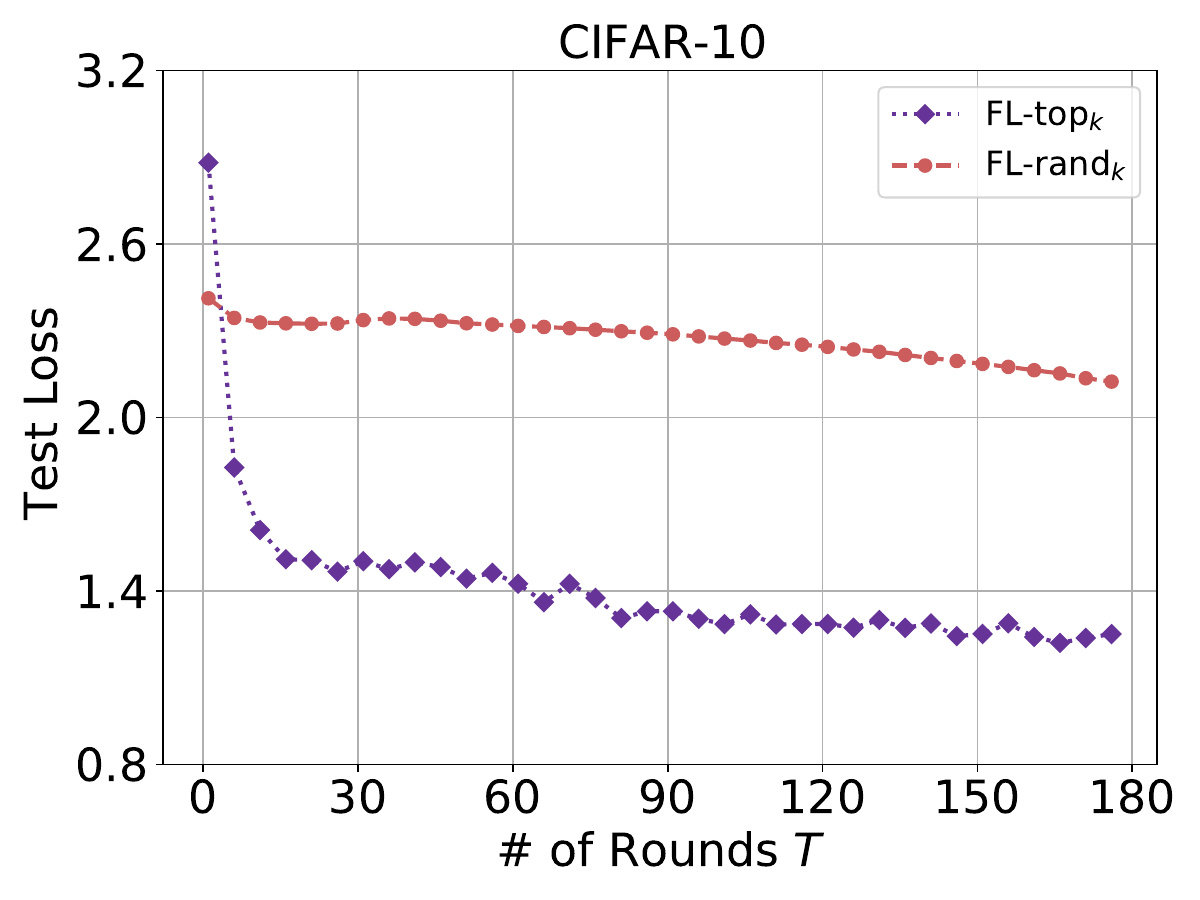}\\(d)
		\end{minipage}
		\\
		\begin{minipage}[t]{0.45\linewidth}
			\centering
			\includegraphics[width=\linewidth]{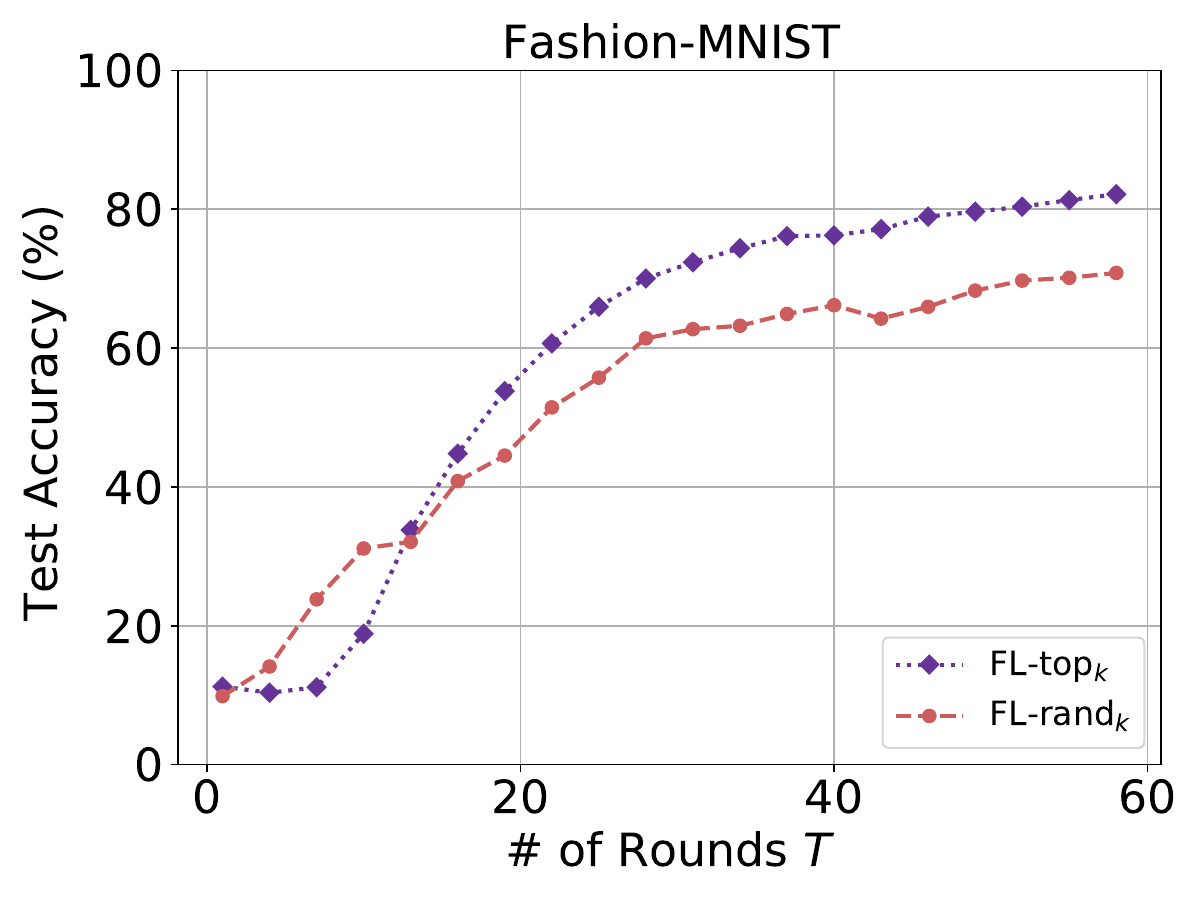}\\(e)
		\end{minipage}
		\begin{minipage}[t]{0.45\linewidth}
			\centering
			\includegraphics[width=\linewidth]{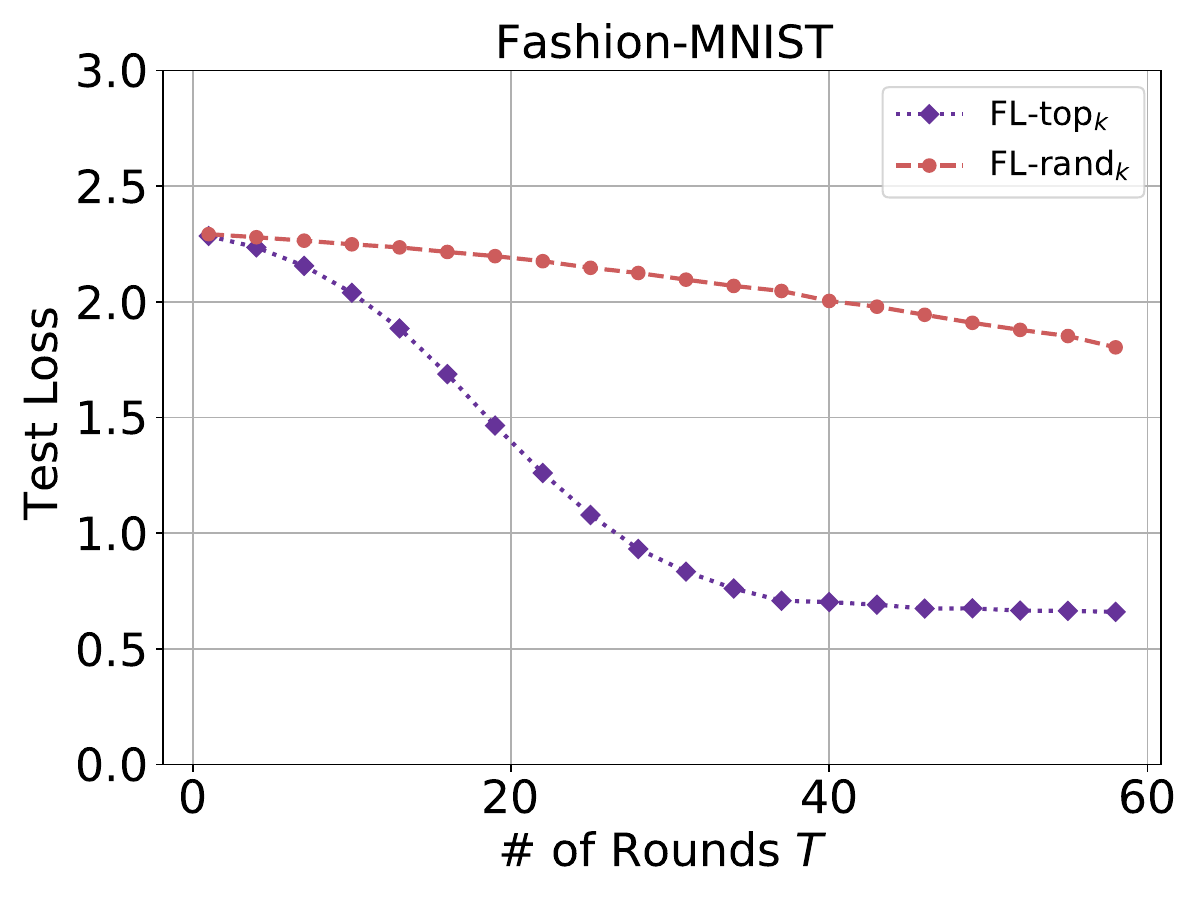}\\(f)
		\end{minipage}
		\caption{\revise{Comparison of test accuracy and test loss versus the number of training rounds $T$ for \textsf{FL-top}$_k$ and \textsf{FL-rand}$_k$ on different datasets, with the density of the sparsified gradient update $\lambda=0.5\%$.} }
		\label{fig:utility_nodp_vary_T}
		\vspace{-18pt}
		
	\end{figure}

	\section{Experiments}
	\label{sec:experiments}
	\noindent \textbf{Implementation.} We implement a prototype system of {\main} in Python. 
	Our code is open-source\footnote{https://github.com/Shuangqing-Xu/Clover}. 
	For federated model training, PyTorch is utilized. 
	We employ AES in CTR mode as PRG. 
	All experiments are conducted on a workstation equipped with an Intel Xeon CPU boasting 64 cores running at 2.20GHz, 3 NVIDIA RTX A6000 GPUs, 256G RAM, and running Ubuntu 20.04.6 LTS.
	To emulate client-server and inter-server communication, we use the loopback filesystem, with both communication delays set to 40 ms and the bandwidth limited to 100 Mbps to simulate a WAN setting.
	
	\noindent \textbf{Datasets and models.} We conduct experiments on three widely-used datasets with the following model architectures: 
	(1) \emph{MNIST} \cite{LeCunBBH98}: MNIST contains 60,000 training and 10,000 test images of handwritten digits (0-9). We use a variant of the LeNet-5~\cite{LeCunBBH98} architecture with 431,080 trainable parameters, commonly adopted in distributed learning~\cite{CaPC,ACORN}. 
	(2) \emph{Fashion-MNIST} \cite{FMNIST}: Fashion-MNIST contains 60,000 training and 10,000 testing examples of Zalando's article images. Each is a grayscale image labeled among 10 classes.
	We utilize the same LeNet-5 architecture as above.  
	(3) \emph{CIFAR-10} \cite{krizhevsky2009learning}: CIFAR-10 contains 50000 training color images and 10000 testing color images, with 10 classes. 
	We adopt the ResNet-18 \cite{HeZRS16} backbone as in \cite{LiuCCS24,ELSA}, which consists of 11,173,962 trainable parameters.
	For all the datasets, we randomly shuffle the training data records and evenly distribute them across $n$ clients. 

	\noindent \revise{\textbf{Baselines.} 
		We compare {\main}'s utility against following baselines: 
		\vspace{-12pt}
		\begin{itemize}
			\item  \textsf{DP-FedAvg} \cite{McMahanRT018}: the vanilla differentially private FedAvg algorithm, which operates under central DP and does not apply gradient sparsification. 
			\item   \textsf{FedSMP-rand$_k$} \cite{HuGG24}: an FL algorithm that combines SecAgg with DP and uses random sparsification for communication reduction, where each client only adds DP noise to the gradient values associated with the common random indices selected by the server, and then use secure aggregation to produce a noisy gradient update at server. 
			\item   \textsf{FedSel} \cite{liu2020fedsel}:  a two-stage LDP-FL algorithm combining gradient sparsification with LDP, consisting of a dimension selection (DS) stage---where each client privately selects one important dimension from its local top-$k$ index set---and a value perturbation (VP) stage, which perturbs the selected value using an LDP mechanism.
		\end{itemize}
	}
	\vspace{-3pt}
	
	To highlight the efficiency of our secure sparse vector aggregation mechanism, \textsf{SparVecAgg}, we construct a baseline that uses distributed ORAM to aggregate the sparse vectors, as described in Section \ref{sec:sparvecagg:design_rationale}. 
	This baseline is implemented using MP-SPDZ \cite{Keller20}, a state-of-the-art framework for general-purpose MPC that has been widely used for establishing baselines in recent works \cite{DautermanRPS22,WatsonWP22}. 
	In this baseline, we use the same three-server honest-majority system model as {\main} but consider a semi-honest threat model. 
	For each sparse vector, the non-zero elements and their corresponding indices are secret-shared among the servers in a 61-bit prime field with RSS. 
	Then the servers collaboratively execute the secure protocol \textsf{OptimalORAM} from MP-SPDZ for the secure aggregation of the secret-shared index-value pairs. 
	For simplicity, we denote this baseline as {\baseline}.

	\noindent \textbf{Parameters.} In all experiments, unless otherwise specified, the client sampling ratio is set to $0.1$, i.e., $|\mathcal{P}^t| / n=0.1$, where $n=100$. 
	The local learning rate $\eta_l$ is initialized to $0.1$ with a decay rate of $0.005$ and a momentum of $0.5$. 
	The number of local training epochs is set to $30$, and the local batch size is $50$. 
	For privacy parameters, the noise multiplier $\sigma$ is set to $0.8$, and the privacy failure probability $\delta$ is set to $\frac{1}{n}$. 
	The $\ell_2$-norm clipping threshold $C$ is set as $0.1$ for the MNIST and Fashion-MNIST datasets and $0.2$ for the CIFAR-10 dataset. 
	The optimal privacy parameter $\alpha$ is determined by minimizing the RDP parameter $\tau$. 
	For both {\main} and {\baseline}, we consider all elements are represented in fixed-point format over a 61-bit Mersenne prime field $p = 2^{61} - 1$ to enable efficient modular arithmetic, using 15 bits for fractional precision.

	\begin{figure}[!t]
		
		\centering
		
		\begin{minipage}[t]{0.49\linewidth}
			\centering
			\includegraphics[width=\linewidth]{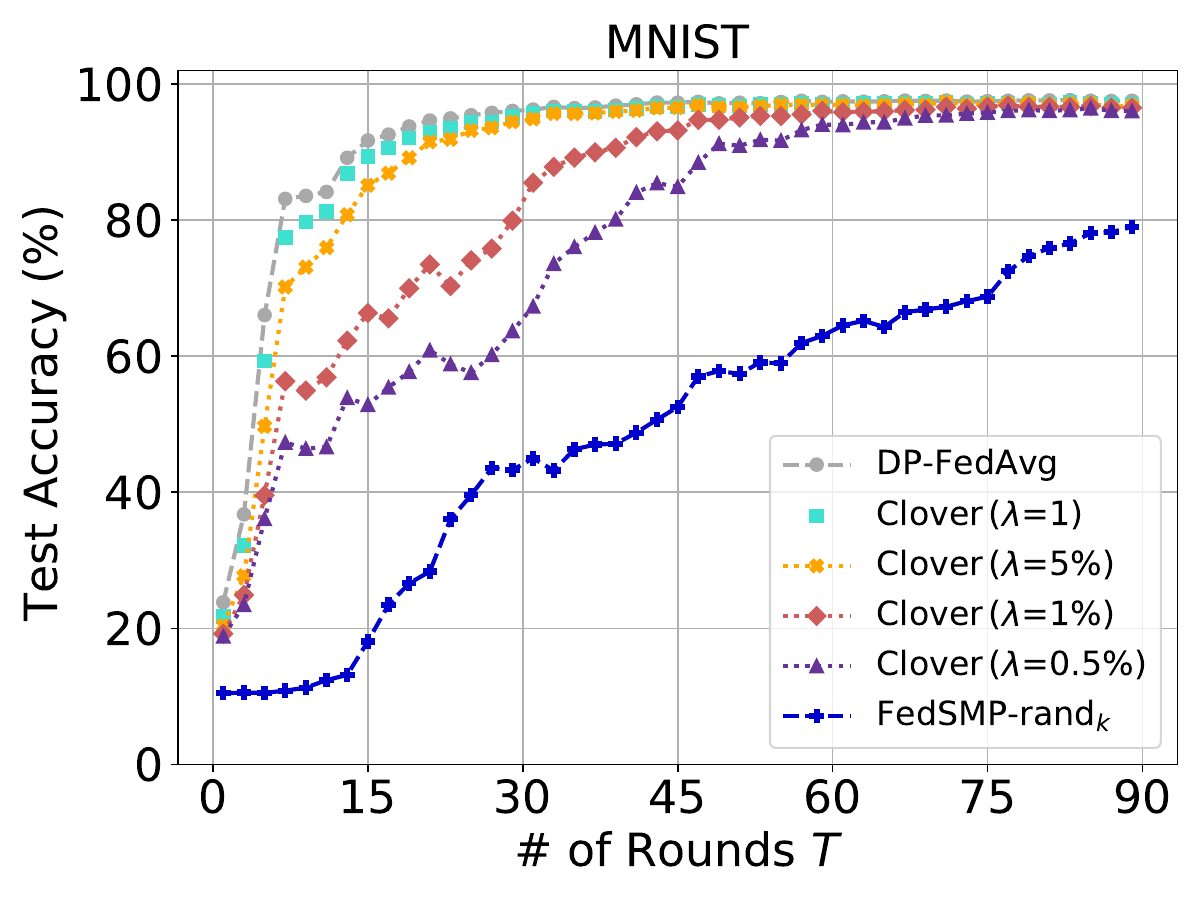}\\(a)
		\end{minipage}
		\begin{minipage}[t]{0.49\linewidth}
			\centering
			\includegraphics[width=\linewidth]{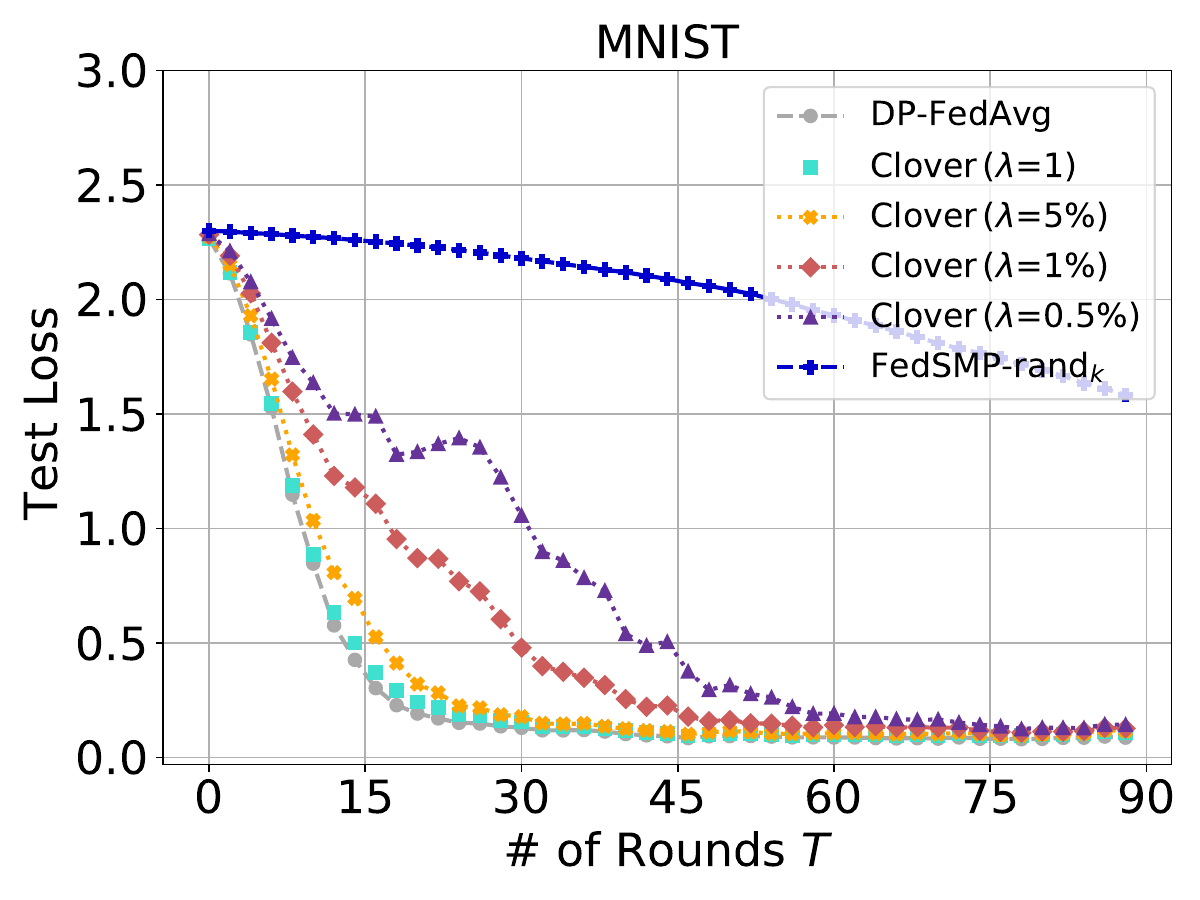}\\(b)
		\end{minipage}
		\\
		\begin{minipage}[t]{0.49\linewidth}
			\centering
			\includegraphics[width=\linewidth]{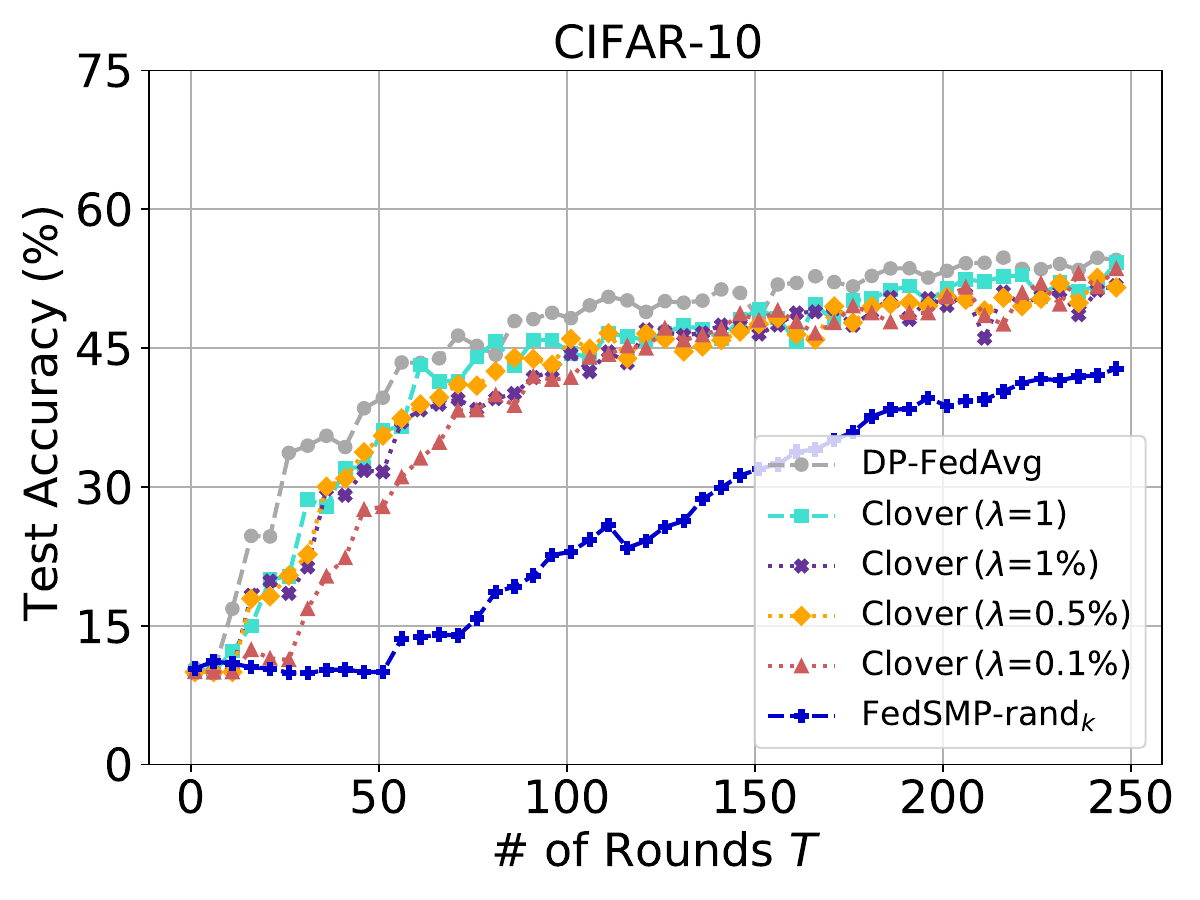}\\(c)
		\end{minipage}
		\begin{minipage}[t]{0.49\linewidth}
			\centering
			\includegraphics[width=\linewidth]{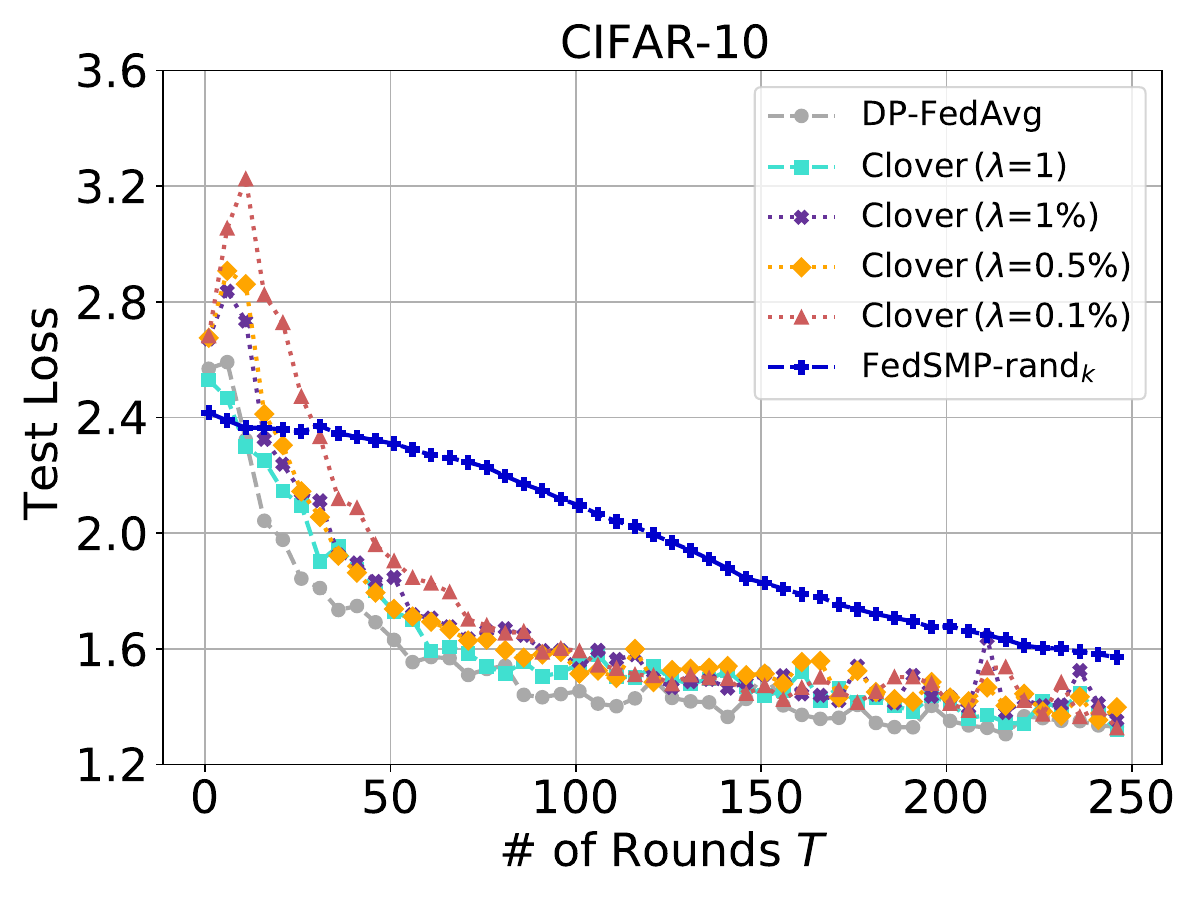}\\(d)
		\end{minipage}
		\\
		\begin{minipage}[t]{0.49\linewidth}
			\centering
			\includegraphics[width=\linewidth]{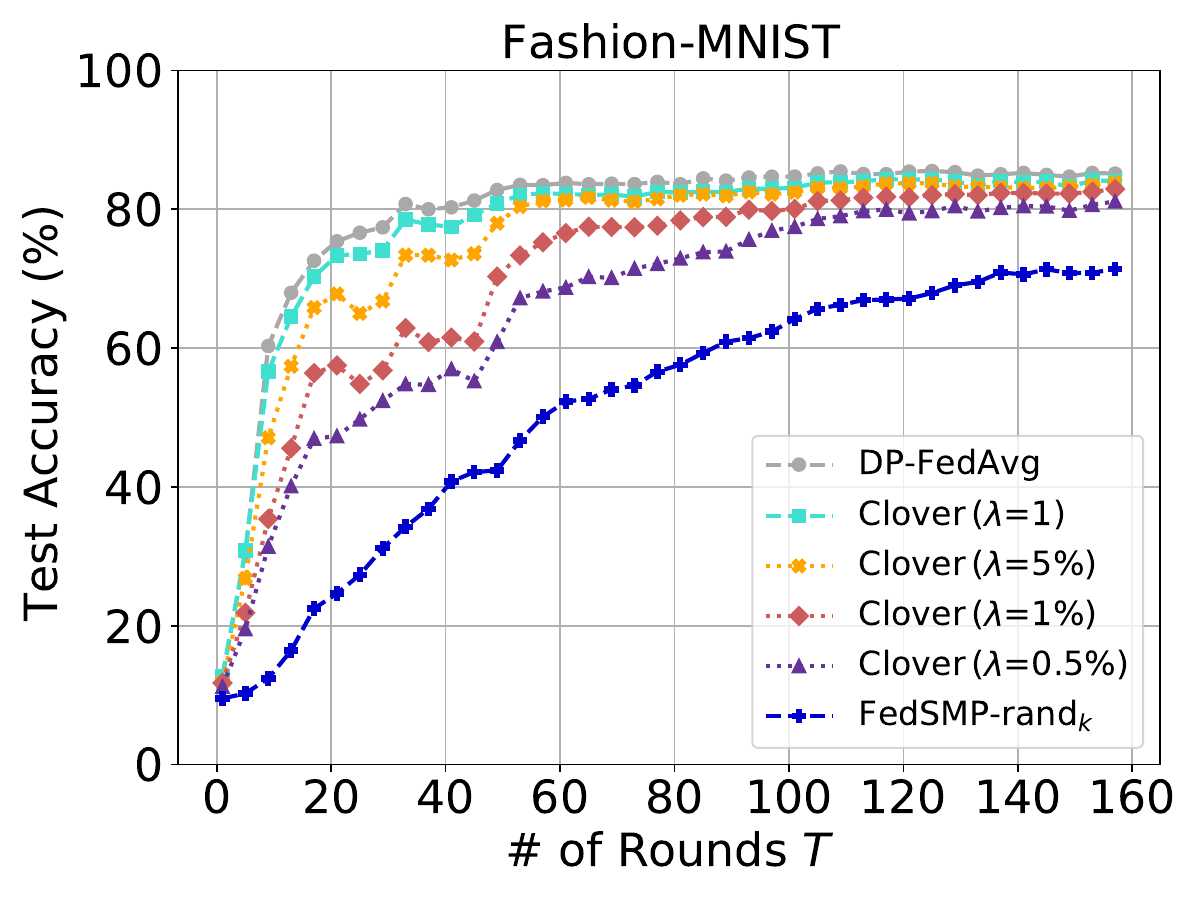}\\(e)
		\end{minipage}
		\begin{minipage}[t]{0.49\linewidth}
			\centering
			\includegraphics[width=\linewidth]{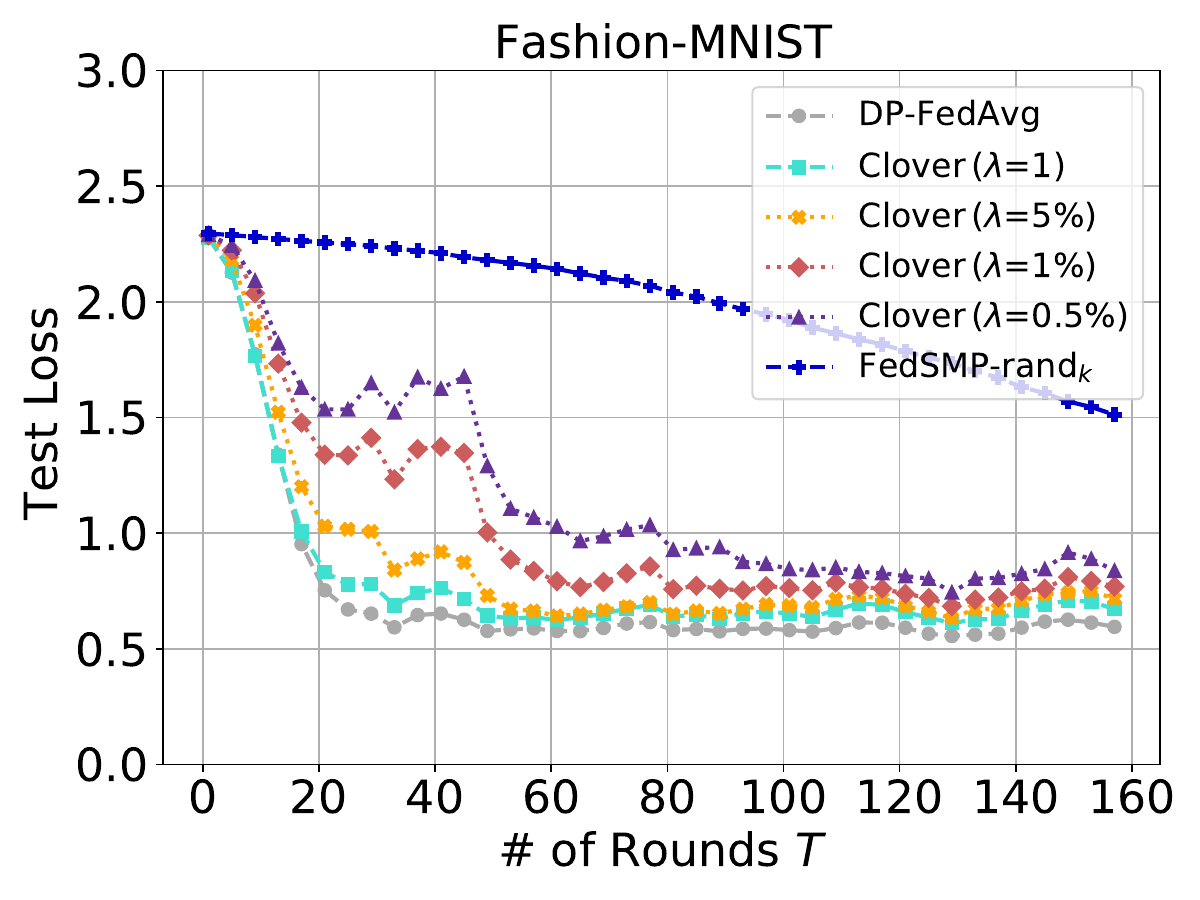}\\(f)
		\end{minipage}
		\caption{\revise{Test accuracy and test loss of {\main} and the baselines on different datasets versus the number of training rounds $T$. For all methods, the overall privacy budget is set to $\varepsilon=6.59$ for MNIST, $\varepsilon=11.96$ for CIFAR-10, and $\varepsilon=9.89$ for Fashion-MNIST.}} 
		\label{fig:utility_vary_T}
		\vspace{-15pt}
	\end{figure}

	\subsection{Utility}
	\label{sec:exp:privacy-utility-trade-off}

	\noindent\revise{\textbf{Validating the necessity of standard top-$k$ sparsification.} We start with validating the necessity of enforcing standard top-$k$ sparsification for communication reduction in FL while retaining good model utility, as opposed to random-$k$ sparsification (which is adopted by some existing secure aggregation protocols for FL to reconcile secure aggregation and gradient sparsification \cite{HuGG24,Ergun21}). 
		Specifically, we compare FL using the standard top-$k$ sparsification mechanism (denoted as \textsf{FL-top}$_k$) with FL using a na\"ive random-$k$ sparsification mechanism that randomly selects $k$ values (denoted as \textsf{FL-rand}$_k$). 
		For \textsf{FL-top}$_k$, we let each client locally sparsify its gradient update using standard top-$k$ sparsification (for different clients their local top-$k$ indices are different). 
		For \textsf{FL-rand}$_k$, we incorporate random sparsification by letting clients share a common random set of indices for sparsifying their local gradient updates, and this set is re-sampled at the beginning of each training round, as in \cite{KerkoucheEuroSP21,KerkoucheUAI21}. 
		For both approaches, we adopt the same default parameters as described at the beginning of Section \ref{sec:experiments}, with the number of training rounds $T$ set as follows: $T=45$ for MNIST, $T=180$ for CIFAR-10, $T=60$ for Fashion-MNIST, and the density of the sparsified gradient update $\lambda=0.5\%$. 
		Note that in this comparison, we do not apply $\ell_2$-norm clipping or noise addition to either \textsf{FL-top}$_k$ or \textsf{FL-rand}$_k$, as our focus is on examining the importance of enforcing standard top-$k$ sparsification for retaining good model utility.
		%
		As shown in Fig.~\ref{fig:utility_nodp_vary_T}, the results demonstrate the superior utility of \textsf{FL-top}$_k$ over \textsf{FL-rand}$_k$ across all three datasets. 
		For instance, on the MNIST dataset (Fig.~\ref{fig:utility_nodp_vary_T} (a) and (b)), \textsf{FL-top}$_k$ achieves over 90\% test accuracy within the first 10 training rounds, with the test loss decreasing correspondingly to near-zero. 
		In contrast, \textsf{FL-rand}$_k$ exhibits very slow learning progress, struggling to reach approximately 70\% accuracy after 45 rounds.
		Similar trends are observed on the CIFAR-10 and Fashion-MNIST datasets (Figs.~\ref{fig:utility_nodp_vary_T} (c)--(f)). 
		These results have validated the necessity of enforcing standard top-$k$ sparsification in FL for communication reduction while retaining good model utility.
		This in turn showcases the significance and value of our first-of-its-kind delicate design for securely applying standard top-$k$ sparsification in FL with strong privacy assurance.
	}

	\noindent\revise{\textbf{Comparing {\main} with prior works.} 
		Then we investigate the utility of {\main} by varying the number of training rounds $T$ with different density $\lambda$, comparing with \textsf{FedSMP-rand$_k$} and \textsf{DP-FedAvg}. 
		For \textsf{FedSMP-rand$_k$}, we set $\lambda=0.5\%$, and the other parameters identical to those used in {\main}. 
		From Fig.~\ref{fig:utility_vary_T}, we can observe that {\main} consistently outperforms \textsf{FedSMP-rand$_k$} across all three datasets. 
		For instance, on MNIST (Fig.~\ref{fig:utility_vary_T}(a)), {\main} ($\lambda=0.5\%$) achieves a higher final accuracy of approximately 97\%, whereas \textsf{FedSMP-rand$_k$} (with the same density $\lambda=0.5\%$) only reaches about 80\% accuracy after 90 rounds. 
		%
		%
		Moreover, from Fig.~\ref{fig:utility_vary_T}, we observe that {\main} with different $\lambda$ values achieves model utility comparable to \textsf{DP-FedAvg}: 
		although sparsification slightly delays the initial training progress, all configurations achieve similar accuracy after sufficient rounds. 
		This observation aligns with theoretical results in distributed learning \cite{WangniWLZ18}, which show that sparsification maintains the same convergence rate as dense gradient updates. 
	}
	
	\revise{
		We further investigate the utility of {\main} by varying overall privacy budget $\varepsilon$ of {\main} with different density $\lambda$. 
		We include \textsf{FedSMP-rand$_k$} (with $\lambda=0.5\%$), \textsf{FedSel}, and \textsf{DP-FedAvg} for this comparison. 
		For the LDP-FL baseline \textsf{FedSel}, we follow \cite{liu2020fedsel} and allocate the per-round privacy budget $\varepsilon_0$ in a 1:9 ratio between the DS stage (using the exponential mechanism \cite{DworkR14}) and the VP stage (using the piecewise mechanism \cite{WangXYZHSS019}). 
		We set $\varepsilon_0 = 2$ for MNIST and Fashion-MNIST, and $\varepsilon_0 = 10$ for CIFAR-10. The overall privacy budget is accounted as $T \varepsilon_0$, following \cite{liu2020fedsel}. 
		The other parameters of \textsf{FedSel} are set the same as {\main}. 
		Fig.~\ref{fig:utility-vary-eps} demonstrates the better utility of {\main} against \textsf{FedSMP-rand$_k$} and \textsf{FedSel} under different privacy budget. 
		First, when compared to \textsf{FedSMP-rand$_k$}, for any given privacy budget $\varepsilon$, {\main} consistently achieves higher test accuracy across all three datasets, indicating a better balance between privacy and utility. 
		Furthermore, when compared with the LDP-based \textsf{FedSel} method, {\main} shows a vast advantage. 
		For example, {\main} achieves over 96\% accuracy with a practical privacy budget of $\varepsilon < 10$ and fewer than 90 training rounds on MNIST. 
		In contrast, \textsf{FedSel} only achieves a final accuracy of  approximately 70\% with an impractically large privacy budget $\varepsilon = 10^4$ and a high number of training rounds $T=5000$ due to the heavy noise required by LDP.
	}

	\begin{figure}[!t]
		
		\centering
		\begin{minipage}[t]{0.45\linewidth}
			\centering
			\includegraphics[width=\linewidth]{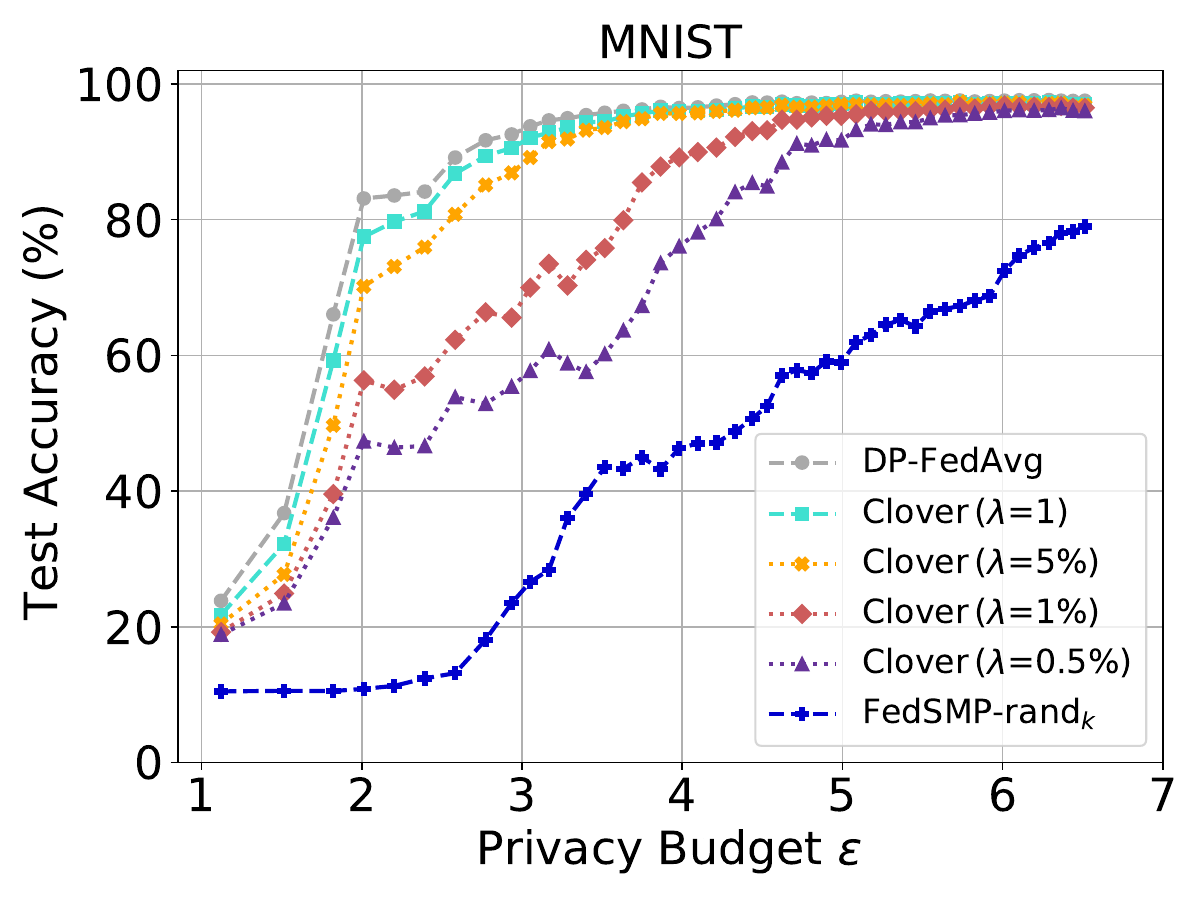}\\(a)
		\end{minipage}
		\begin{minipage}[t]{0.45\linewidth}
			\centering
			\includegraphics[width=\linewidth]{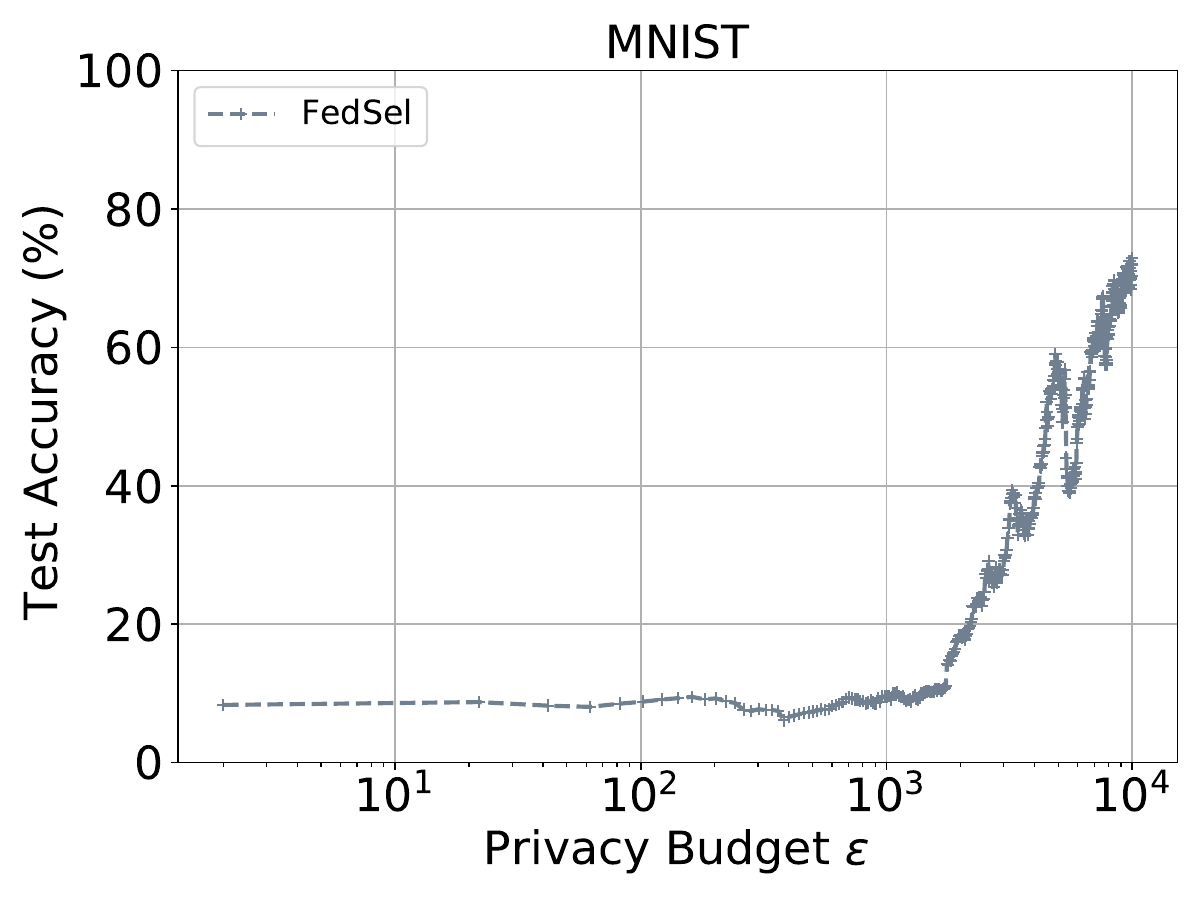}\\(b)
		\end{minipage}
		\\
		\begin{minipage}[t]{0.45\linewidth}
			\centering
			\includegraphics[width=\linewidth]{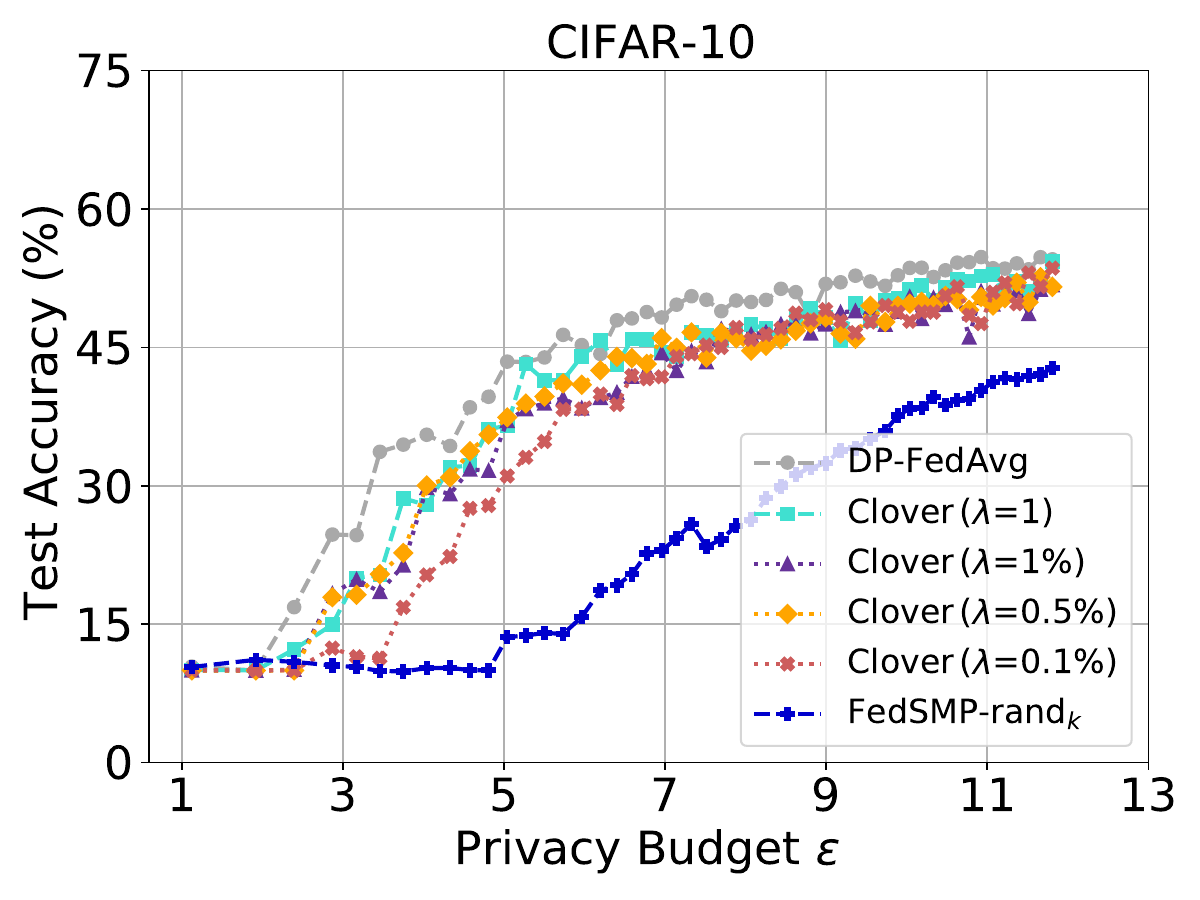}\\(c)
		\end{minipage}
		\begin{minipage}[t]{0.45\linewidth}
			\centering
			\includegraphics[width=\linewidth]{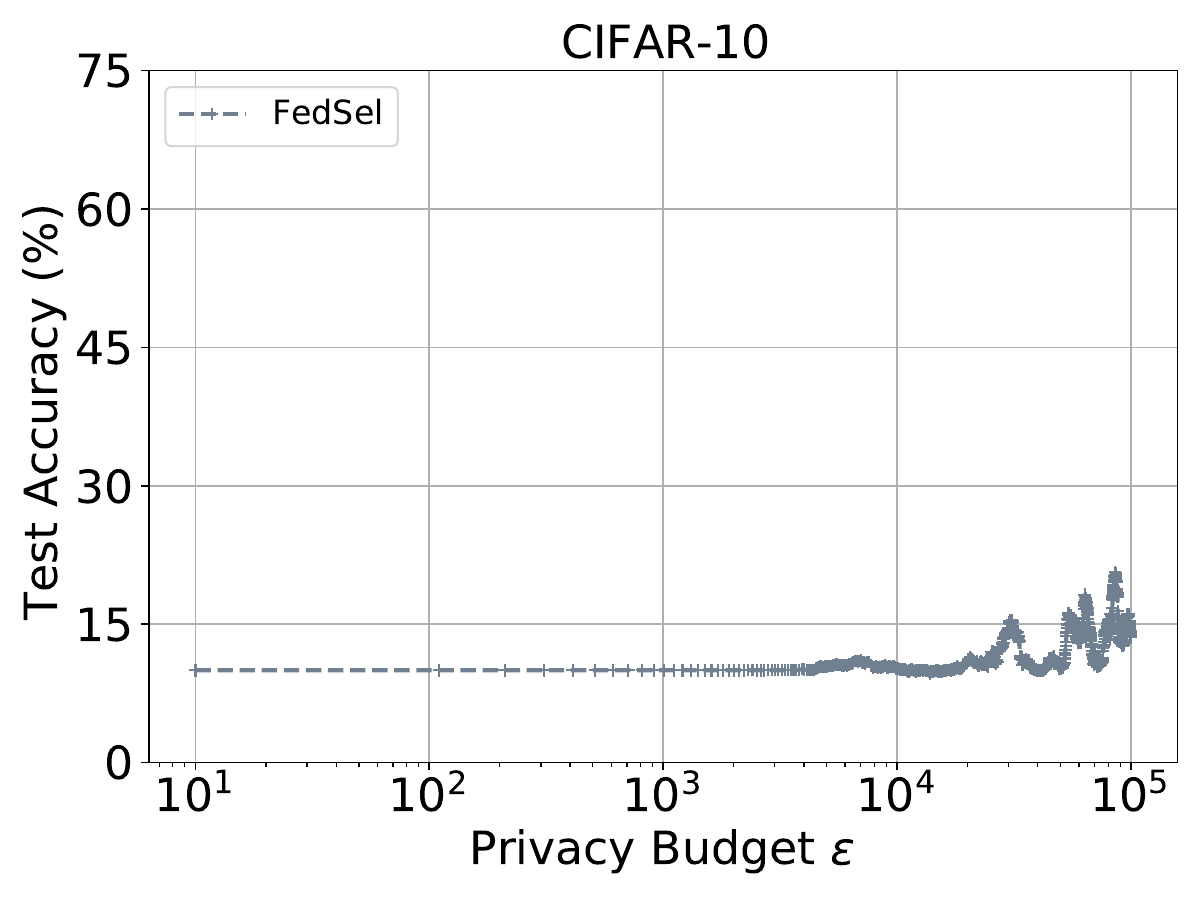}\\(d)
		\end{minipage}
		\\
		\begin{minipage}[t]{0.45\linewidth}
			\centering
			\includegraphics[width=\linewidth]{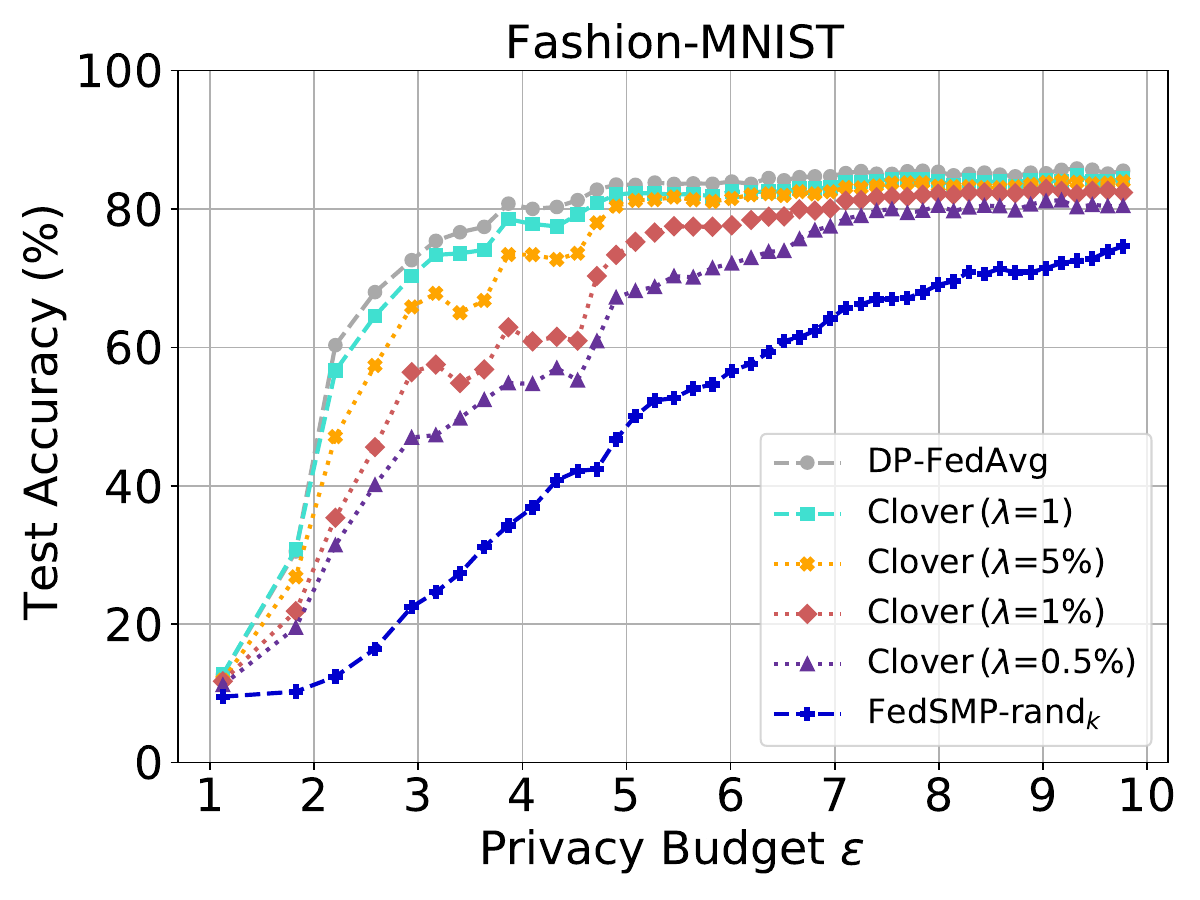}\\(e)
		\end{minipage}
		\begin{minipage}[t]{0.45\linewidth}
			\centering
			\includegraphics[width=\linewidth]{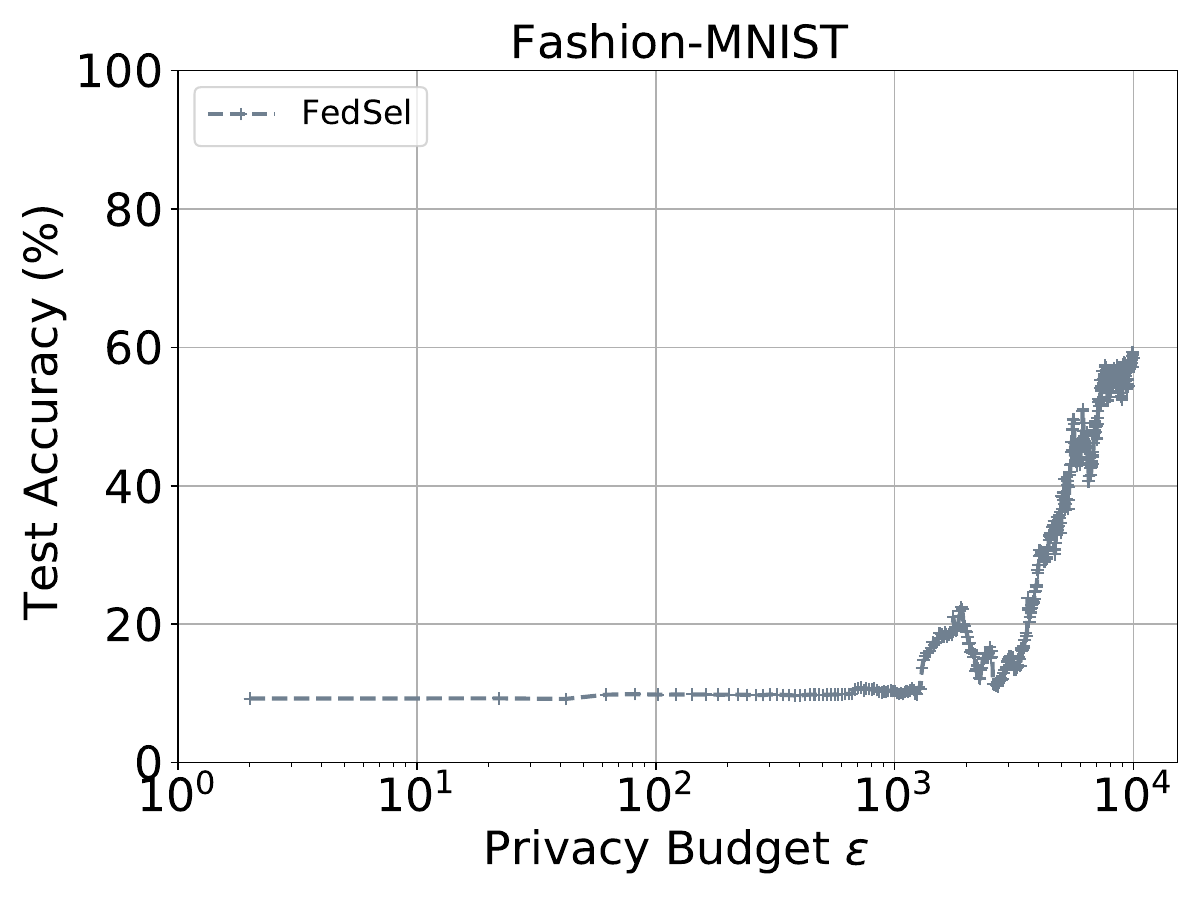}\\(f)
		\end{minipage}
		\caption{\revise{Test accuracy of {\main} and the baselines on different datasets versus the privacy budget $\varepsilon$. 
				The number of training rounds $T$ for {\main}, \textsf{DP-FedAvg}, and \textsf{FedSMP-rand$_k$} is limited at $T \le 90$ for MNIST, $T \le 250$ for CIFAR-10, and $T \le 180$ for Fashion-MNIST. 
				For \textsf{FedSel}, we limit $T \le 5000$ for MNIST and Fashion-MNIST, and $T \le 10^4$ for CIFAR-10.}} 
		\label{fig:utility-vary-eps}
		
	\end{figure}

	\begin{table*}[!t]
		\centering
		\caption{\revise{Training costs of {\main} with semi-honest security and malicious security respectively on the MNIST, Fashion-MNIST, and CIFAR-10 datasets evaluated per training round under varying density $\lambda$.}}
		\label{tab:effi_real_datasets}
		\setlength\tabcolsep{3pt}
		\scalebox{0.8}{\revise{
				\begin{tabular}{ccccccc}
					\hline 
					\multirow{2}{*}{Dataset} & \multirow{2}{*}{Threat Model} & \multirow{2}{*}{Density $\lambda$} & \multirow{2}{*}{\begin{tabular}{c}Per-Client \\ Comp. Cost (s)\end{tabular}} & \multirow{2}{*}{\begin{tabular}{c}Overall Client \\ Comm. Cost (MB)\end{tabular}} & \multirow{2}{*}{\begin{tabular}{c}Inter-Server \\ Comm. Cost (MB)\end{tabular}} & \multirow{2}{*}{\begin{tabular}{c}Server-Side \\ Overall Runtime (s)\end{tabular}} \\
					\\
					\hline
					\multirow{6}{*}{MNIST} 
					& \multirow{3}{*}{Semi-honest} & 0.5\% & 2.786 & 1.254 & 75.233 & 8.211 \\
					& & 1\%  & 2.824      & 2.508     & 75.233      & 8.168     \\
					& & 5\%  & 3.072      & 12.539     & 75.233      & 8.184     \\
					& \multirow{3}{*}{Malicious}    & 0.5\% & 3.708 & 1.255 & 200.621 & 28.082 \\
					&    & 1\%  & 3.955       & 2.509      & 200.621       & 29.362      \\
					&    & 5\%  & 3.402      & 12.540     & 200.621       & 28.307      \\
					\cline{1-7} 
					\multirow{6}{*}{CIFAR-10} 
					& \multirow{3}{*}{Semi-honest} & 0.1\% & 14.467 & 6.501 & 1950.106 & 184.581 \\
					&  & 0.5\%  & 14.464     & 32.502    & 1950.106        & 183.534       \\
					&  & 1\%  & 14.336      & 65.004     & 1950.106         &  183.572     \\
					& \multirow{3}{*}{Malicious}   & 0.1\% & 31.711 & 6.502 & 5200.284 & 675.069 \\
					&   & 0.5\%  & 31.434      & 32.503     & 5200.284        & 670.772       \\
					&   & 1\%  & 31.517      & 65.005     & 5200.284        & 673.342       \\
					\cline{1-7} 
					\multirow{6}{*}{Fashion-MNIST} 
					& \multirow{3}{*}{Semi-honest} & 0.5\% & 2.876 & 1.254 & 75.233 & 8.803 \\
					& & 1\%  & 2.478     & 2.508     & 75.233      & 8.362     \\
					& & 5\%  & 2.737     & 12.539     & 75.233      & 8.181     \\
					& \multirow{3}{*}{Malicious}    & 0.5\% & 3.250 & 1.255 & 200.621 & 28.295 \\
					&   & 1\%  & 3.743     & 2.509     & 200.621        & 28.095     \\
					&   & 5\%  & 3.198     & 12.540     & 200.621        & 28.108     \\
					\hline
			\end{tabular}}
		}
		\vspace{-10pt}
	\end{table*}
	%
	%

	\begin{figure}[t!]
		
		\centering
		
		\begin{minipage}[t]{0.45\linewidth}
			\centering
			\includegraphics[width=\linewidth]{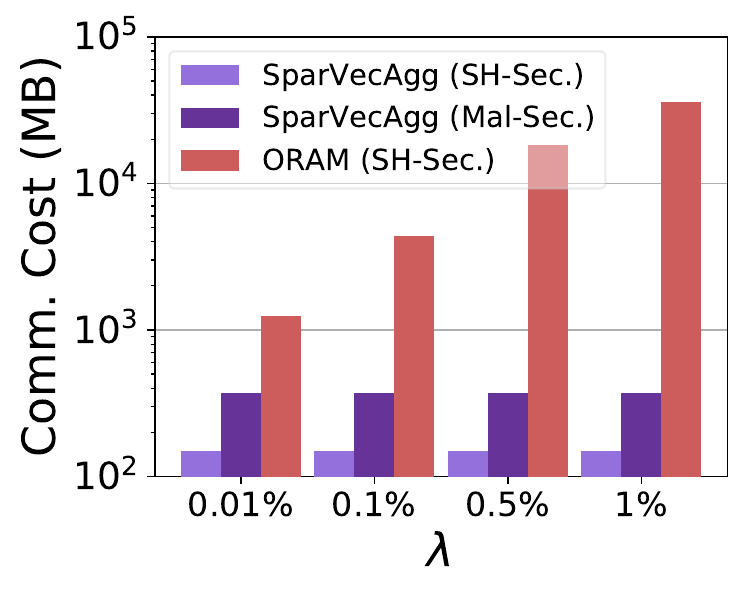}\\(a)
		\end{minipage}
		\begin{minipage}[t]{0.45\linewidth}
			\centering
			\includegraphics[width=\linewidth]{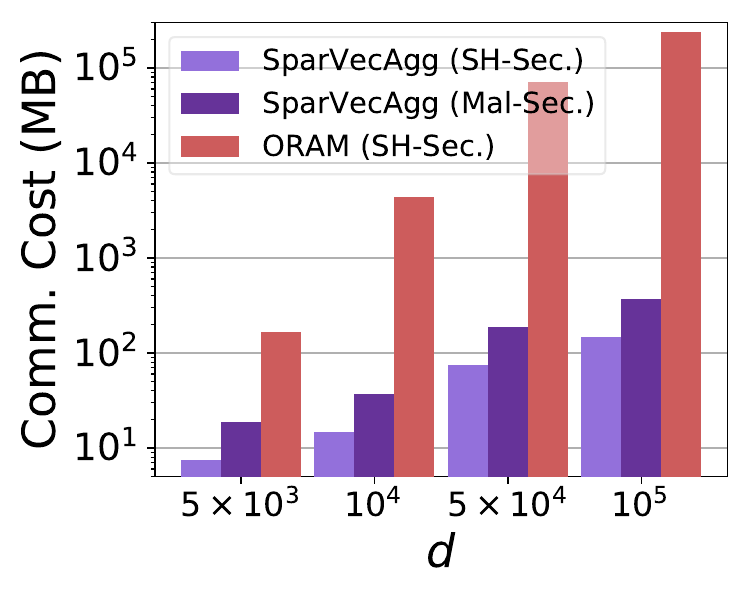}\\(b)
		\end{minipage}
		
		\caption{Comparison of inter-server communication costs for securely aggregating 100 sparse vectors under different approaches. (a) Varying the density $\lambda\in\{0.01\%,0.1\%,0.5\%,1\%\}$ and fixing $d=10^5$. (b) Varying the vector dimension $d\in\{5\times10^3,10^4,5\times10^4,10^5\}$ and fixing $\lambda=1\%$.}
	\label{fig:effi_comm_cost}
	\vspace{-10pt}
	
\end{figure}

\begin{figure}[t!]
	
	\centering
	
	\begin{minipage}[t]{0.45\linewidth}
		\centering
		\includegraphics[width=\linewidth]{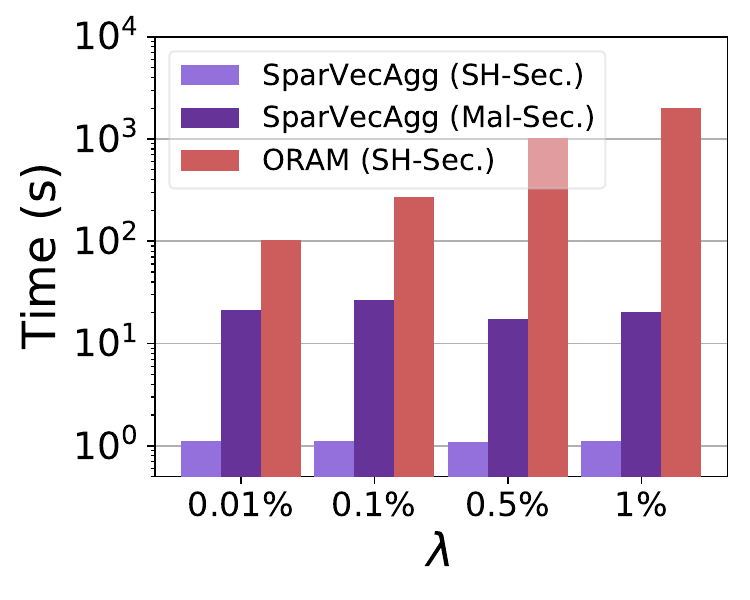}\\(a)
	\end{minipage}
	\begin{minipage}[t]{0.45\linewidth}
		\centering
		\includegraphics[width=\linewidth]{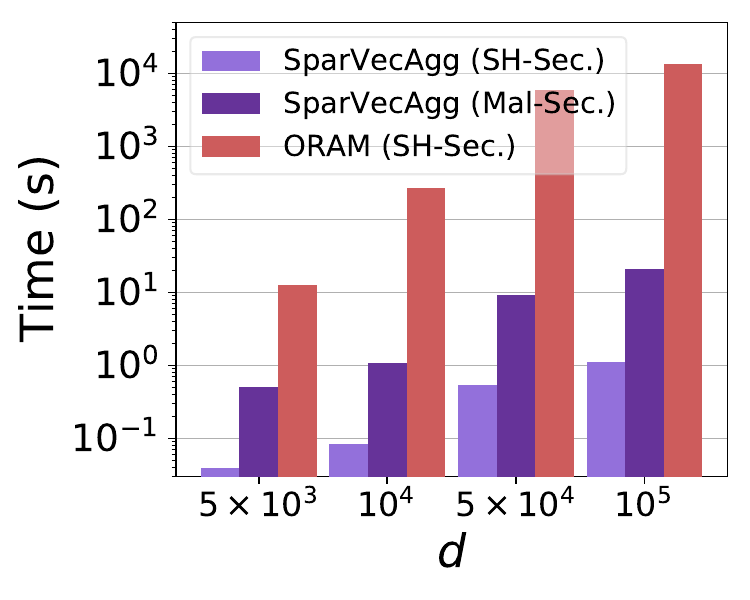}\\(b)
	\end{minipage}
	\caption{Comparison of server-side computation costs for securely aggregating 100 sparse vectors under different approaches. (a) Varying the density $\lambda\in\{0.01\%,0.1\%,0.5\%,1\%\}$ and fixing $d=10^5$. (b) Varying the vector dimension $d\in\{5\times10^3,10^4,5\times10^4,10^5\}$ and fixing $\lambda=1\%$.}
	
\label{fig:effi_run_time}
\vspace{-10pt}

\end{figure}

\vspace{-10pt}

\subsection{Efficiency}
\label{sec:exp:efficiency}

\revise{
We first evaluate the per-round training costs of {\main} under both semi-honest and malicious security on MNIST, Fashion-MNIST, and CIFAR-10, using the default parameters from Section~\ref{sec:experiments} and varying $\lambda$.
Table~\ref{tab:effi_real_datasets} summarizes the results. The per-client computation cost includes local execution time, while the client communication cost refers to the total data volume transmitted by all selected clients per round. The inter-server communication cost measures per-server data exchange during protocol execution. The server-side overall runtime captures server computation, latency, and data transfer time. 
As shown in Table~\ref{tab:effi_real_datasets}, the client communication cost scales approximately linearly with $\lambda$. For instance, on MNIST with semi-honest security, it increases from $1.25$ MB to $12.54$ MB as $\lambda$ grows from $0.5\%$ to $5\%$. Similar trends hold under malicious security and across other datasets, confirming that {\main} achieves significant bandwidth savings for clients at low $\lambda$ values. 
Moreover, this reduction is robust across different security levels, with malicious security only introducing a slight additional overhead in client communication cost compared with semi-honest security. 
Furthermore, the inter-server communication cost of the maliciously secure {\main} is only about 2.67$\times$ higher, and the server-side overall runtime is at most 3.67$\times$ that of {\main} in the semi-honest setting, which results from our proposed lightweight integrity check mechanisms. 
}


\revise{
Furthermore, we also conduct experiments on synthetic data to compare the efficiency of our sparse vector aggregation mechanism for the secure aggregation of 100 sparse vectors with both semi-honest security (denoted as \textsf{SparVecAgg (SH-Sec.)}) and malicious security (denoted as \textsf{SparVecAgg (Mal-Sec.)}) against the semi-honest baseline {\baseline}. 
The sparse vectors in each experiment are randomly generated with a specified density $\lambda$.
First, we vary $\lambda \in \{0.01\%,0.1\%,0.5\%,1\%\}$ while keeping $d = 10^5$ fixed to evaluate the inter-server communication costs and server-side computation costs of \textsf{SparVecAgg (SH-Sec.)}, \textsf{SparVecAgg (Mal-Sec.)}, and {\baseline}. 
Fig. \ref{fig:effi_comm_cost}(a) and Fig. \ref{fig:effi_run_time}(a) reveal two critical insights: (1) both the communication costs and computation costs of {\baseline} exhibit a strong positive correlation with $\lambda$, whereas \textsf{SparVecAgg (SH-Sec.)} and \textsf{SparVecAgg (Mal-Sec.)} maintain stable performance across all tested $\lambda$ values, and (2) both \textsf{SparVecAgg (SH-Sec.)} and \textsf{SparVecAgg (Mal-Sec.)} achieve significantly lower costs than {\baseline} due to our novel permutation-based computation paradigm. 
For example, with $\lambda = 1\%$, \textsf{SparVecAgg (SH-Sec.)} incurs $\sim243\times$ less communication cost and $\sim1805\times$ less computation cost than {\baseline}. 
Similarly, \textsf{SparVecAgg (Mal-Sec.)} incurs $\sim96\times$ less communication cost and $\sim99\times$ less computation cost than {\baseline}. 
}

Furthermore, we fix the density $\lambda = 1\%$ and vary the vector dimension $d \in \{5 \times 10^3, 10^4, 5 \times 10^4, 10^5\}$ to evaluate the efficiency of \textsf{SparVecAgg (SH-Sec.)}, \textsf{SparVecAgg (Mal-Sec.)}, and {\baseline} for secure aggregation of 100 sparse vectors. 
The evaluation results in Fig. \ref{fig:effi_comm_cost}(b) and Fig. \ref{fig:effi_run_time}(b) demonstrate that both \textsf{SparVecAgg (SH-Sec.)} and \textsf{SparVecAgg (Mal-Sec.)} achieve significantly lower inter-server communication costs and server-side computation costs than {\baseline}. 
For instance, when $d = 10^5$, \textsf{SparVecAgg (SH-Sec.)} incurs $\sim1602\times$ less communication cost and $\sim12041\times$ less communication cost than {\baseline}. 
Similarly, \textsf{SparVecAgg (Mal-Sec.)} incurs $\sim636\times$ less communication cost and $\sim653\times$ less communication cost than {\baseline}. 
These findings highlight {\main}'s superior suitability for high-dimensional vector aggregation in large-scale FL applications.

\section{Conclusion and Discussion}
\label{sec:conlcusion}

In this paper, we present {\main}, a novel sparsification-enabled secure and differentially private FL framework. 
Through a delicate synergy of gradient sparsification, lightweight cryptography, differential privacy, and customized data encoding, {\main} is the first system framework that can simultaneously (1) support gradient sparsification while protecting both the indices and values of top-$k$ gradients in a lightweight and secure manner, (2) offer DP guarantees on the trained models with utility comparable to central DP, and (3) ensure robustness against a malicious server through efficiently checking the integrity of server-side computation. 
%
%
Extensive evaluation results demonstrate that {\main} achieves utility comparable to \textsf{DP-FedAvg} (which works under central DP).
Compared to an ORAM-based strawman approach with only semi-honest security, our secure sparse vector aggregation mechanism, with either semi-honest security or malicious security, substantially reduces both inter-server communication costs and server-side computation costs by several orders of magnitude. 
Our experiments also demonstrate that due to our lightweight integrity check mechanisms, {\main} achieves malicious security with only modest performance overhead (limited to 2.67$\times$ increase in inter-server communication cost and 3.67$\times$ increase in server-side overall runtime), compared to the semi-honest security counterpart.

\revise{
{\main} is currently designed and built in the three-server distributed trust setting because we leverage RSS as the foundation for lightweight secure computation, which typically operates in the three-server model. 
Recall that we also use RSS-based secret-shared shuffle, of which the state-of-the-art protocol \cite{AsharovHIKNPTT22} is currently designed for the three-server setting. 
Also, RSS nicely facilitates the design for malicious security. 
In principle, it is possible to extend {\main} beyond three servers with RSS adapted to operate under $>$3 servers.  
However, this is not immediately clear. 
Doing so would need a careful protocol redesign and would inherently incur significant overhead, which requires a careful treatment.
}


\section*{Acknowledgments}

We sincerely thank the shepherd and the anonymous reviewers for their constructive and invaluable feedback. This work was supported in part by the National Natural Science Foundation of China under Grant 62572150, by the Guangdong Basic and Applied Basic Research Foundation under Grant 2024A1515012299, and by the Shenzhen Science and Technology Program under Grant JCYJ20230807094411024.

\bibliographystyle{ACM-Reference-Format}
\bibliography{reference}

\appendix

\revise{
\section{Utility Evaluation of {\main} on Differentially Private Sparse Vector Summation}
\label{appendix:utility_vector_sum}
In addition to evaluating the utility of {\main} in FL tasks, we also use synthetic data to directly evaluate the utility of {\main} on the task of differentially private sparse vector summation. 
Specifically, we consider a scenario with $n=100$ clients, where each client $\mathcal{C}_i$ holds a $d=10^4$ dimensional sparse vector $\boldsymbol{v}_i$. 
Each vector is generated with a specific density $\lambda$, and is subsequently normalized to have an $\ell_2$-norm of $C=1$.
The local sparse vectors are then processed by our protocol, which outputs an aggregated dense vector $\bar{\boldsymbol{v}}$ that satisfies $(\varepsilon, \delta)$-DP.
We measure the utility using the Mean Squared Error (MSE). 
Given $\hat{\boldsymbol{v}}=\frac{1}{n} \sum_{i \in[n]} \boldsymbol{v}_i$ as the true mean of the input vectors and $\bar{\boldsymbol{v}}$ as the estimated vector from our protocol, the MSE is defined as: $\mathrm{MSE}=\frac{1}{d} \|\bar{\boldsymbol{v}}-\hat{\boldsymbol{v}}\|_2^2$. 
In our experiments, we analyze the impact of both the privacy budget $\varepsilon$ and the density $\lambda$ on the MSE (with $\delta$ fixed at $10^{-5}$). 
In Fig.~\ref{fig:utility_MSE}(a), we plot the MSE as a function of the privacy budget $\varepsilon$ (from 1.0 to 10.0). Each curve in the plot corresponds to a different fixed vector density ($\lambda \in \{0.1\%, 0.5\%, 1\%\}$).
In Fig.~\ref{fig:utility_MSE}(b), we plot the MSE as a function of the vector density $\lambda$ (from 1\% to 5\%). Each curve here corresponds to a different fixed privacy budget ($\varepsilon \in \{1.0, 5.0, 10.0\}$).
As observed in Fig.~\ref{fig:utility_MSE}(a), the MSE decreases as the privacy budget $\varepsilon$ increases. 
This is expected, as a larger $\varepsilon$ allows less noise to be added. 
Notably, the curves for different density levels ($\lambda = 0.1\%, 0.5\%, 1\%$) are almost overlapping. 
This is because our protocol correctly aggregates the non-zero elements from each client's sparse vector $\boldsymbol{v}_i$. 
Therefore, the difference between our estimated mean $\bar{\boldsymbol{v}}$ and the true mean $\hat{\boldsymbol{v}}$ is precisely the DP noise. 
As a result, the MSE in this experiment becomes a direct measure of the variance of the added noise, which is solely determined by $\varepsilon$ and is independent of $\lambda$. 
From Fig.~\ref{fig:utility_MSE}(b), we can draw a similar conclusion. 
For a fixed privacy budget $\varepsilon$ (e.g., the $\varepsilon=10.0$ curve), the MSE remains almost constant as the density $\lambda$ varies. 
}

\section{Proof of Theorem \ref{thm:analysis_overview}}
\label{appendix:analysis_overview}

\subsection{Proof of Privacy}
\label{appendix:proof:privacy}
\begin{proof}
The privacy analysis follows the RDP framework \cite{RDP} combined with privacy amplification by subsampling \cite{WangBK19}. 
Subsampling in RDP reduces the privacy cost compared to applying the same privacy mechanism to the entire population, of which the amplification result is defined as follows. 

\begin{lem}
	\label{lem:subsampled_gaussian}
	\textit{\textbf{(RDP for Subsampled Gaussian Mechanism {\cite{WangBK19}}).}}
	%
	For any $\alpha \geq 2$ and $0 < q < 1$ be a subsampling ratio of subsampling operation Samp\(_q\).
	Let $\mathcal{G}'_{\mathcal{M}} = G_{\mathcal{M}} \circ \text{Samp}_q(\cdot)$ be a subsampled Gaussian mechanism, where $G_{\mathcal{M}}$ is Gaussian Mechanism. 
	Then, $\mathcal{G}'_{\mathcal{M}}$ satisfies $(\alpha, \tau'(\alpha, \sigma))$-RDP where
	\begin{equation}
		\begin{aligned} 
			\tau'(\alpha, \sigma) \triangleq \frac{1}{\alpha-1} \log (1+2q^2 \binom{\alpha}{2} &\min (2 e^{(\alpha-1)^2 / \sigma^2}-1, e^{\alpha^2 / \sigma^2} ) \\
			&+ \sum_{j=3}^{\alpha} 2q^j \binom{\alpha}{j} e^{j(j-1) / 2\sigma^2} ).
		\end{aligned}
		\nonumber
		\vspace{-4pt}
	\end{equation}
\end{lem}

Consider a single round $t$ of {\main}. The servers compute an aggregate based on inputs from a subset $\mathcal{P}^t$ of clients, where $|\mathcal{P}^t| = qn$. Let the function computed be $f(\{\boldsymbol{x}_i^{t+1}\}_{\mathcal{C}_i \in \mathcal{P}^t}) = \sum_{\mathcal{C}_i \in \mathcal{P}^t} \boldsymbol{x}_i^{t+1}$. The inputs $\boldsymbol{x}_i^{t+1}$ are the clipped, sparse gradient updates satisfying $\|\boldsymbol{x}_i^{t+1}\|_2 \leq C$.
The $\ell_2$-sensitivity $\Delta_f$ of the sum function $f$ is thus $C$. 
The mechanism adds noise ${\eta}^t$ to the sum $f$, where ${\eta}^t$ is effectively distributed as $\mathcal{N}_{\mathbb{Z}}(0, \sigma^2 C^2 \mathbf{I}_d)$ in each server's view (resulting from the sum of noise from two honest servers, see Section~\ref{sec:malicious:noise-sampling}). The use of discrete Gaussian $\mathcal{N}_{\mathbb{Z}}$ provides guarantees closely matching the continuous Gaussian $\mathcal{N}(0, \sigma^2 C^2 \mathbf{I}_d)$ \cite{kairouz21}.

The RDP of the (subsampled) Gaussian mechanism for sensitivity $\Delta_f = C$ and effective noise variance $\sigma_{\text{eff}}^2 = \sigma^2 C^2$ is analyzed. 
With Poisson subsampling at rate $q=|\mathcal{P}^t|/n$, the privacy is amplified. 
Using Lemma~\ref{lem:subsampled_gaussian} and Lemma 3.7 from \cite{UAI23}, for $\alpha \geq 2$, the $t$-th round RDP guarantee becomes: $\tau_{t} \leq \frac{3.5q^2\alpha}{\sigma^2}$. 
Over $T$ training rounds, by applying Lemma \ref{lem:sequantial} (RDP Composition), the entire {\main} framework satisfies $(\alpha, T \cdot \tau) = (\alpha, \frac{3.5q^2\alpha T}{\sigma^2})$-RDP. 
Finally, we convert from $(\alpha, \tau_{\text{total}})$-RDP to $(\varepsilon, \delta)$-DP using Lemma~\ref{lem:rdp_to_dp}:
$$
\varepsilon \le \tau_{\text{total}} + \frac{\log(1/\delta)}{\alpha-1} = \frac{3.5 q^2 \alpha T}{\sigma^2} + \frac{\log(1/\delta)}{\alpha-1}.
$$
Setting $\alpha = 1 + \frac{2\log(1/\delta)}{\varepsilon}$ and solving for $\sigma^2$, we obtain the noise requirement:
$$
\sigma^2 \geq \frac{14q^2T\log(1/\delta)}{\varepsilon^2} + \frac{7q^2T}{\varepsilon}.
$$
This ensures that {\main} satisfies $(\varepsilon, \delta)$-DP after $T$ training rounds.
This concludes the proof for the privacy guarantee of {\main}.
\end{proof}

\begin{figure}[t!]

\centering

\begin{minipage}[t]{0.49\linewidth}
	\centering
	\includegraphics[width=\linewidth]{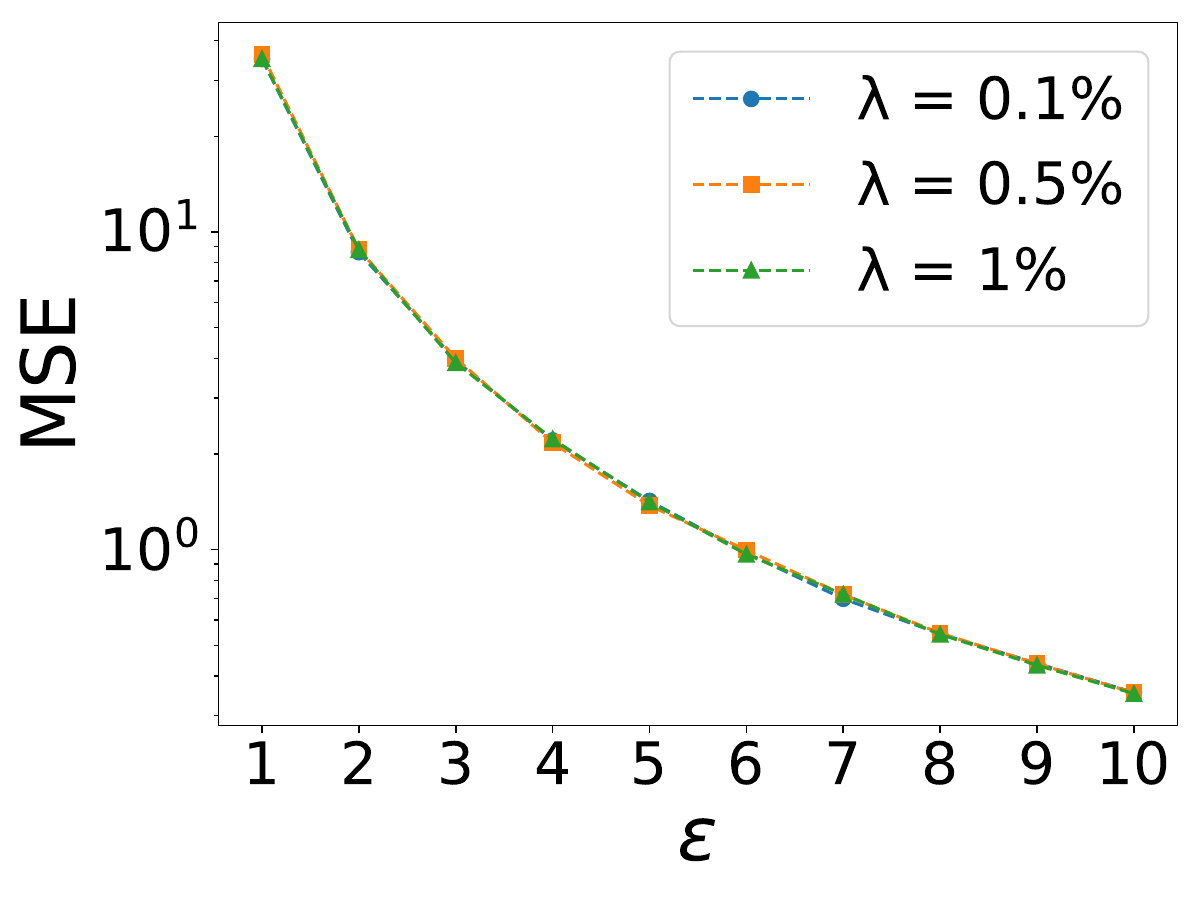}\\(a)
\end{minipage}
\begin{minipage}[t]{0.49\linewidth}
	\centering
	\includegraphics[width=\linewidth]{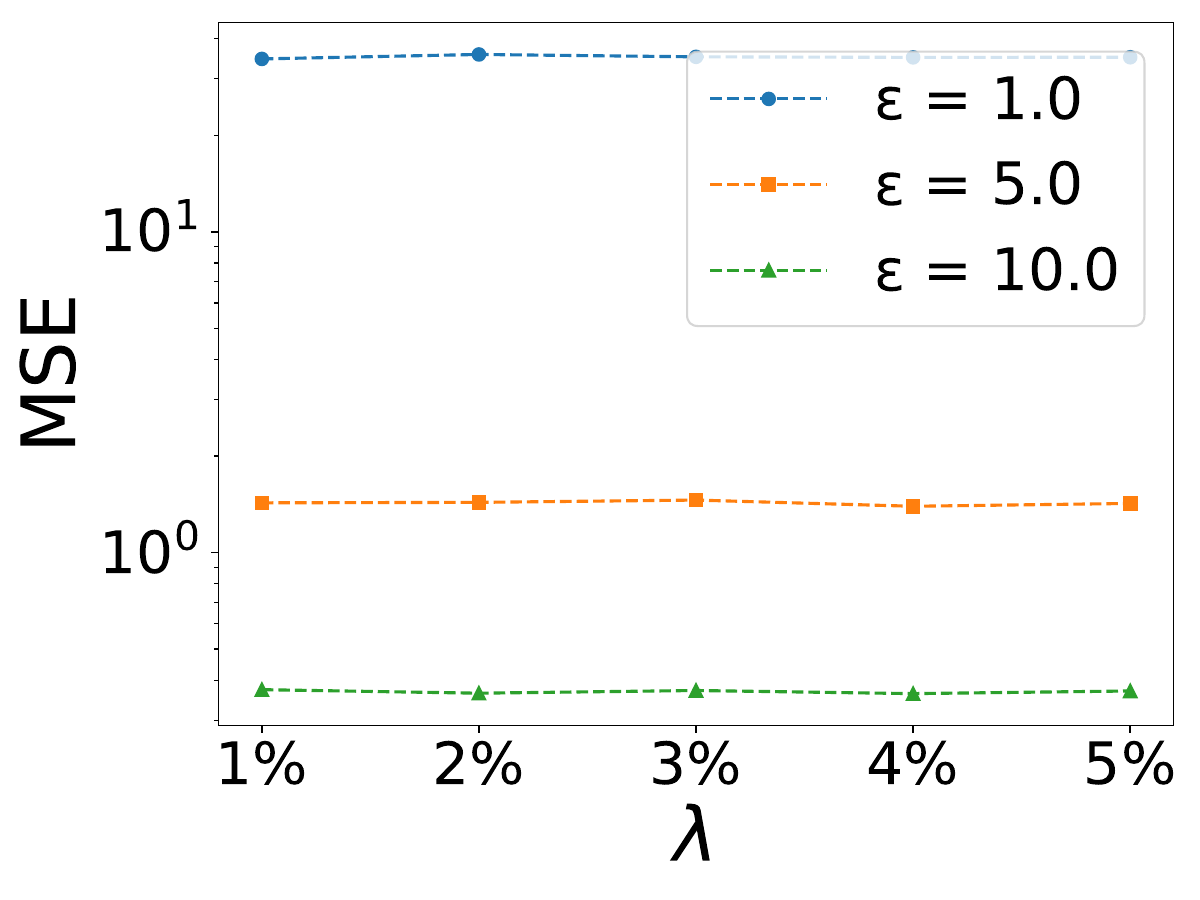}\\(b)
\end{minipage}

\caption{\revise{Utility of {\main} on differentially private sparse vector summation with 100 $10^4$-dimensional sparse vectors. \textbf{(a)} MSE vs. privacy budget $\varepsilon$ under different densities $\lambda$. \textbf{(b)} MSE vs. density $\lambda$ under different privacy budgets~$\varepsilon$.}}

\label{fig:utility_MSE}
\vspace{-5pt}

\end{figure}

\subsection{Proof of Communication}		
\begin{proof}
The communication costs are analyzed for client-server communication and inter-server communication of our maliciously secure {\main}.

\noindent\textit{Client-server Communication:}
In each round $t$, each participating client $\mathcal{C}_i$ processes its raw gradient update $\Delta_i^{t+1}$ and uploads information related to the clipped, sparse update $\boldsymbol{x}_i^{t+1}$. Based on the \textsf{SparVecAgg} protocol (Algorithm~\ref{alg:SparVecAgg}) and the associated permutation compression (Section~\ref{sec:sparvecagg:details}):
\begin{itemize}
	\item The client computes $\boldsymbol{r}_i$ containing the $k$ non-zero values of $\boldsymbol{x}_i^{t+1}$ (after potential clipping). It generates RSS shares $\llbracket \boldsymbol{r}_i \rrbracket$. Sending the necessary shares $(\langle \boldsymbol{r}_i \rangle_j, \langle \boldsymbol{r}_i \rangle_{j+1})$ to each server $\mathcal{S}_j$ requires transmitting $2k$ field elements to each of the 3 servers, totaling $6k$ field elements. Cost: $6k |\mathbb{Z}_p|$ bits.
	\item The client generates two random seeds $\mathsf{s}_{i,0}, \mathsf{s}_{i,1}$ (length $|s|$ bits each) for compressing permutations $\pi_{i,0}, \pi_{i,1}$. It sends $(\mathsf{s}_{i,0}, \mathsf{s}_{i,1})$ to $\mathcal{S}_0$, $\mathsf{s}_{i,1}$ to $\mathcal{S}_1$, and $\mathsf{s}_{i,0}$ to $\mathcal{S}_2$. Total cost: $4|s|$ bits.
	\item The client generates three random MAC key seeds and distributes them among the servers. The locally computed MAC is also secret-shared among the servers. Total cost: $6|s|+6|\mathbb{Z}_p|$ bits.
	\item The client computes the distilled third permutation $\pi_{i,2}$ and extracts the list $P_i$ of the first $k$ destination indices. It sends $P_i$ to $\mathcal{S}_1$ and $\mathcal{S}_2$. Each index can be represented as a field element (or requires $|\mathbb{Z}_p|$ bits), the cost is $k|\mathbb{Z}_p|$ to $\mathcal{S}_1$ and $k|\mathbb{Z}_p|$ to $\mathcal{S}_2$. Total cost: $2k |\mathbb{Z}_p|$ bits.
\end{itemize}
Summing these costs, the total client-to-server communication per client per round is $6k |\mathbb{Z}_p| + 4|s| + 6|s|+6|\mathbb{Z}_p| + 2k |\mathbb{Z}_p| = (8k+6) |\mathbb{Z}_p| + 10|s|$ bits.

\noindent\textit{Inter-Server Communication:}
This occurs primarily during the secure sparse gradient update aggregation for all $|\mathcal{P}^t| = qn$ clients and during the noise sampling phase.
\begin{itemize}
	\item Secure sparse gradient update aggregation (per client): The dominant cost is the secret-shared shuffle. Applying the three permutations ($\pi'_{i,0}, \pi'_{i,1}, \pi'_{i,2}$) sequentially involves communication. According to \cite{AsharovHIKNPTT22}, shuffling a vector of size $d$ using their RSS-based protocol requires communication proportional to $d$ (applying 3 permutations costs $\mathcal{O}(d)$ field elements exchanged between servers per client). Reconstruction of the final $\boldsymbol{x}_i$ also takes $\mathcal{O}(d)$ communication. Adding the MAC verification (Section~\ref{sec:malicious:aggregation}) requires a secure dot product ($\mathcal{O}(d)$ communication for multiplications) and reconstruction ($\mathcal{O}(1)$). Total per client: $\mathcal{O}(d)$.
	\item Verifiable noise sampling: Each server $\mathcal{S}_j$ samples noise $\eta_j$ (size $d$) and secret-shares it. Sharing involves $\mathcal{O}(d)$ communication. Summing the shared noises is local. Verifying the noise involves mask sharing (cost $\mathcal{O}(d)$), mask summing (local), and robust reconstruction/comparison of masked noise (cost $\mathcal{O}(d)$). Total per verified server: $\mathcal{O}(d)$. Total for all 3 servers: $\mathcal{O}(d)$. 
	\item {Final aggregation/reconstruction:} Summing the verified client aggregates is local. Summing the verified noise is local. Adding them is local. The final reconstruction of $\Delta^{t+1}$ requires $\mathcal{O}(d)$ communication.
\end{itemize}
The total inter-server communication per round involves running secure sparse gradient update aggregation for $qn$ clients and performing noise generation / verification / aggregation once. 
The cost is dominated by the call of secret-shared shuffles. Total inter-server communication complexity per round is $\mathcal{O}(|\mathcal{P}^t| \cdot d)$.
\end{proof}

\subsection{Proof of Convergence}
\label{appendix:proof:convergence}
\begin{proof}
Now we analyze the convergence of {\main} that combines FL with top-$k$ sparsification and Gaussian noise addition for achieving DP. 
The objective is to minimize the global loss function $F(\mathbf{w}, \mathcal{D}) = \sum_{i=0}^{n-1} \frac{1}{n} F_i\left(\mathbf{w}, \mathcal{D}_i\right)$, where \( F_i\) is the local loss function of client \( i \), \( n \) is the total number of clients, and \( \mathbf{w} \in \mathbb{R}^d \) is the model parameter vector of dimension \( d \). 
To understand {\main}'s behavior, we first outline its operation in each round \( t \). 
A subset \(\mathcal{P}^t\) of \( qn \) clients is sampled, where \( q \) is the sampling fraction. 
Each selected client \(\mathcal{C}_i \in \mathcal{P}^t\) computes its local update \(\Delta_i^{t+1} = \mathbf{w}_i^E - \mathbf{w}^t\), with \(\mathbf{w}_i^E\) being the local model after \( E \) steps of SGD starting from the global model \(\mathbf{w}^t\), using a local learning rate \(\eta_l\). 
Next, the client applies top-\(k\) sparsification \(\mathcal{C}(\cdot)\) to \(\Delta_i^{t+1}\), retaining the \( k \) largest-magnitude components, followed by clipping to ensure the \(\ell_2\)-norm is bounded by \( C \), producing \(\boldsymbol{x}_i^{t+1}\). 
Finally, the server securely aggregates these updates and adds Gaussian noise \(\eta^t \sim \mathcal{N}_\mathbb{Z}(0, \frac{3}{2} \sigma^2 C^2 \mathbf{I}_d)\), updating the global model as: $\mathbf{w}^{t+1} = \mathbf{w}^t + \frac{1}{qn} \left( \sum_{\mathcal{C}_i \in \mathcal{P}^t} \boldsymbol{x}_i^{t+1} + \eta^t \right).$

Before diving into the analysis, we make the following standard assumptions, which are commonly used in FL \cite{CVPR23,CVPR22,GhadimiL13a,HuGG24,ReddiCZGRKKM21,Qsparse}: 

\begin{assump}[$L$-Smoothness]
	Each local loss function \( F_i(\mathbf{w}) \) is differentiable, and its gradient \(\nabla F_i(\mathbf{w})\) is \( L \)-Lipschitz continuous, satisfying:
	\[
	\|\nabla F_i(\mathbf{w}) - \nabla F_i(\mathbf{w}')\| \leq L \|\mathbf{w} - \mathbf{w}'\|, \quad \forall \mathbf{w}, \mathbf{w}' \in \mathbb{R}^d.
	\]
\end{assump}

\begin{assump}[Bounded Local Variance]
	Since the data is IID, the expected gradient of each client matches the global gradient, i.e., \(\nabla F_i(\mathbf{w}) = \nabla F(\mathbf{w})\). The stochastic gradient \(\nabla F_i(\mathbf{w}; \beta_i)\) over batch \(\beta_i\) has bounded variance:
	\[
	\mathbb{E}_{\beta_i} \|\nabla F_i(\mathbf{w}; \beta_i) - \nabla F(\mathbf{w})\|^2 \leq \sigma_l^2.
	\]
\end{assump}

\begin{assump}[Bounded Gradient via Clipping]
	After clipping, each update satisfies:
	\[
	\|\boldsymbol{x}_i^{t+1}\|_2 \leq C.
	\]
\end{assump}

\begin{assump}[Bounded Gradient Norm]
	\label{assump:bounded-gradient}
	The gradient of the global loss function is bounded, such that for all \( t \):
	\[
	\|\nabla F(\mathbf{w}^t)\|^2 \leq G^2,
	\]
	where \( G \) is a positive constant.
\end{assump}

\begin{assump}[Bounded Compression Error]
	\label{assump:compression-error}
	The top-\(k\) sparsification operator \(\mathcal{C}(\cdot)\) introduces a bounded compression error. Specifically, for all \( i \) and \( t \), it holds that:
	\[
	\mathbb{E} \|\mathcal{C}(\Delta_i^{t+1}) - \Delta_i^{t+1}\|^2 \leq \phi \|\Delta_i^{t+1}\|^2,
	\]
	where \(\phi = 1 - \frac{k}{d}\) is the compression error parameter, with \( k \) being the number of retained components and \( d \) the total dimensionality of the model.
\end{assump}

With these assumptions in place, we now proceed to analyze {\main}'s convergence. Leveraging the \( L \)-smoothness of \( F(\mathbf{w}) \), we start with the standard inequality:

\[
F(\mathbf{w}^{t+1}) \leq F(\mathbf{w}^t) + \langle \nabla F(\mathbf{w}^t), \mathbf{w}^{t+1} - \mathbf{w}^t \rangle + \frac{L}{2} \|\mathbf{w}^{t+1} - \mathbf{w}^t\|^2.
\]

\noindent Taking expectations, we obtain:

\[
\mathbb{E}[F(\mathbf{w}^{t+1})] \leq F(\mathbf{w}^t) + \mathbb{E} \left[ \langle \nabla F(\mathbf{w}^t), u^t \rangle \right] + \frac{L}{2} \mathbb{E} \left[ \|u^t\|^2 \right],
\]

\noindent where \( u^t = \mathbf{w}^{t+1} - \mathbf{w}^t = \frac{1}{qn} \left( \sum_{\mathcal{C}_i \in \mathcal{P}^t} \boldsymbol{x}_i^{t+1} + \eta^t \right) \).
To evaluate the inner product term $\langle \nabla F(\mathbf{w}^t), u^t \rangle$, we define the compression error as \( e_i^{t+1} = \boldsymbol{x}_i^{t+1} - \Delta_i^{t+1} \).
Since the expected local update is \(\mathbb{E} [\Delta_i^{t+1}] = -\eta_l E \nabla F(\mathbf{w}^t)\), we have:

\[
\mathbb{E} [ \langle \nabla F(\mathbf{w}^t), u^t \rangle ] = \mathbb{E} [ \langle \nabla F(\mathbf{w}^t), \frac{1}{qn} \sum_{\mathcal{C}_i \in \mathcal{P}^t} \boldsymbol{x}_i^{t+1} \rangle ],
\]

\noindent since \(\mathbb{E} [\eta^t] = 0\). 
Splitting \(\boldsymbol{x}_i^{t+1} = \Delta_i^{t+1} + e_i^{t+1}\), we get:
$$
\begin{aligned}
	\mathbb{E} [ \langle \nabla F(\mathbf{w}^t)&, \frac{1}{qn} \sum_{\mathcal{C}_i \in \mathcal{P}^t} \boldsymbol{x}_i^{t+1} \rangle ] = -\eta_l E \|\nabla F(\mathbf{w}^t)\|^2 + \\ 
	& \mathbb{E} [ \langle \nabla F(\mathbf{w}^t), \frac{1}{qn} \sum_{\mathcal{C}_i \in \mathcal{P}^t} e_i^{t+1} \rangle ].
\end{aligned}
$$

Applying Young's inequality to bound the error term:
$$
\begin{aligned}
	\mathbb{E} [ \langle \nabla F(\mathbf{w}^t), \frac{1}{qn} \sum_{\mathcal{C}_i \in \mathcal{P}^t} e_i^{t+1} \rangle ] \leq \frac{\eta_l E}{4} \|\nabla F(\mathbf{w}^t)\|^2 + \\ 
	\frac{1}{\eta_l E q^2 n^2} \mathbb{E} [ \| \sum_{\mathcal{C}_i \in \mathcal{P}^t} e_i^{t+1} \|^2 ].
\end{aligned}
$$

\noindent From Assumption \ref{assump:compression-error}, \(\mathbb{E} \|e_i^{t+1}\|^2 \leq \phi \mathbb{E} \|\Delta_i^{t+1}\|^2\). From Assumption \ref{assump:bounded-gradient}, we bound:

\[
\mathbb{E} \|\Delta_i^{t+1}\|^2 \leq \eta_l^2 E^2 G^2 + \eta_l^2 E \sigma_l^2,
\]

\noindent so we have \(\mathbb{E} [ \| \sum_{\mathcal{C}_i \in \mathcal{P}^t} e_i^{t+1} \|^2 ] \leq qn \phi \eta_l^2 (E^2 G^2 + E \sigma_l^2)\). Thus:

\[
\mathbb{E} \left[ \langle \nabla F(\mathbf{w}^t), u^t \rangle \right] \leq -\frac{3 \eta_l E}{4} \|\nabla F(\mathbf{w}^t)\|^2 + \frac{\phi \eta_l (E G^2 + \sigma_l^2)}{qn}.
\]

\noindent Next, we compute the norm term \(\mathbb{E} \|u^t\|^2\). Since the noise and gradient updates are independent, we have:

\[
\mathbb{E} \|u^t\|^2 = \frac{1}{q^2 n^2} \mathbb{E} [ \| \sum_{\mathcal{C}_i \in \mathcal{P}^t} \boldsymbol{x}_i^{t+1} \|^2 ] + \frac{1}{q^2 n^2} \mathbb{E} \|\eta^t\|^2,
\]

\noindent where \(\mathbb{E} \|\eta^t\|^2 = \frac{3}{2} d \sigma^2 C^2\). 
For the gradient term, using the same bound:
\[
\mathbb{E} [ \| \sum_{\mathcal{C}_i \in \mathcal{P}^t} \boldsymbol{x}_i^{t+1} \|^2 ] \leq qn \eta_l^2 E^2 (1 + \phi) G^2 + qn \eta_l^2 E (1 + \phi) \sigma_l^2,
\]

\noindent we have:

\[
\mathbb{E} \|u^t\|^2 \leq \frac{\eta_l^2 E^2 (1 + \phi) G^2 + \eta_l^2 E (1 + \phi) \sigma_l^2}{qn} + \frac{3 d \sigma^2 C^2}{2 q^2 n^2}.
\]

\noindent Combining these, we have:
$$
\begin{aligned}
	\mathbb{E}[F(\mathbf{w}^{t+1})] \leq F(\mathbf{w}^t) - \frac{3 \eta_l E}{4} \|\nabla F(\mathbf{w}^t)\|^2 + \frac{\phi \eta_l (E G^2 + \sigma_l^2)}{qn} + \\
	\frac{L}{2} \left( \frac{\eta_l^2 E^2 (1 + \phi) G^2 + \eta_l^2 E (1 + \phi) \sigma_l^2}{qn} + \frac{3 d \sigma^2 C^2}{2 q^2 n^2} \right).
\end{aligned}
$$
\noindent for the $t$-th ground. 
Summing over \( T \) rounds and rearranging, we derive the final convergence bound:

$$
\begin{aligned}
	& \frac{1}{T} \sum_{t=0}^{T-1}  \mathbb{E} \|\nabla F(\mathbf{w}^t)\|^2 \leq \\
	& \mathcal{O} \left( \frac{1}{\eta_l E T} + \frac{L \eta_l E (1 + \phi) (G^2 + \sigma_l^2)}{qn} \right) +  \mathcal{O} \left( \frac{d \sigma^2 C^2}{q^2 n^2} \right).
\end{aligned}
$$
\end{proof}

\section{Our Protocols with Malicious Security}
\label{sec:complete-protocol}
Our proposed verifiable secret-shared noise sampling protocol is detailed in Algorithm~\ref{alg:malicious-noise-sampling}. 
We integrate this protocol with our proposed maliciously secure mechanisms for decompression and secret-shared shuffle (Section~\ref{sec:malicious:shuffle}), as well as for the aggregation of noise and gradient updates (Section~\ref{sec:malicious:aggregation}), and present the complete protocol of our maliciously secure {\main} in Algorithm~\ref{alg:mali-FL}. 

\begin{algorithm}[t]
\caption{Verifiable Noise Sampling (\textsf{VeriNoiseSamp})}
\label{alg:malicious-noise-sampling}
\begin{algorithmic}[1]
	\Require Noise parameters $\sigma$ and $C$, model dimension $d$, significance level $\alpha$.
	\Ensure $\mathcal{S}_{\{0,1,2\}}$ get encrypted discrete Gaussian noises $\llbracket \eta_i \rrbracket$ for $i\in\{0,1,2\}$, where $\eta_i \sim \mathcal{N}_\mathbb{Z}(0, \frac{1}{2}\sigma^2C^2\mathbf{I}_d)$. 
	\For{$i\in\{0,1,2\}$}
	\State // \emph{\underline{Noise sampling and sharing:}}
	\State $\mathcal{S}_i$ locally samples $\eta_i \sim \mathcal{N}_\mathbb{Z}(0, \frac{1}{2}\sigma^2C^2\mathbf{I}_d)$.
	\State $\mathcal{S}_i$ secret-shares $\eta_i$ across $\mathcal{S}_{\{0,1,2\}}$, yielding $\llbracket \eta_i \rrbracket$. 
	
	\State $\mathcal{S}_{i+1}$ generates masking noise $\xi \sim \mathcal{N}_\mathbb{Z}(0, \frac{1}{2}\sigma^2C^2\mathbf{I}_d)$.
	\State $\mathcal{S}_{i+1}$ secret-shares $\xi$ across $\mathcal{S}_{\{0,1,2\}}$, yielding $\llbracket \xi \rrbracket$.
	\State $\mathcal{S}_{\{0,1,2\}}$ securely compute $\llbracket \kappa \rrbracket \leftarrow \llbracket \eta_i \rrbracket + \llbracket \xi \rrbracket$.
	\State $\mathcal{S}_i$ and $\mathcal{S}_{i+1}$ send $\langle \kappa \rangle_{i+1}$ to $\mathcal{S}_{i-1}$.
	\State \multiline{$\mathcal{S}_{i-1}$ compares $\langle \kappa \rangle_{i+1}$ from $\mathcal{S}_i$ and $\mathcal{S}_{i+1}$. If inconsistent, output \textsf{Reject} and abort.}
	\State $\mathcal{S}_{i-1}$ computes $\kappa \leftarrow \sum_{i=0}^{2}\langle \kappa \rangle_i$ to reconstruct $\kappa$.
	\State // \emph{\underline{Noise verification using KS test on $\mathcal{S}_{i-1}$:}}
	\State Generate a reference sample $\kappa' \sim \mathcal{N}_\mathbb{Z}(0, \sigma^2C^2\mathbf{I}_d)$.
	\State $F_{\kappa}(y) \leftarrow \frac{1}{d} \sum_{i=0}^{d-1} \mathbb{I}(\kappa[i] \leq y)$.
	\State $F_{\kappa'}(y) \leftarrow \frac{1}{d} \sum_{j=0}^{d-1} \mathbb{I}(\kappa'[j] \leq y)$.
	\State $D_{\text{KS}}\leftarrow \sup_{y} |F_{\kappa}(y) - F_{\kappa'}(y)|.$ 
	\If {$D_{\text{KS}} > \sqrt{-\frac{1}{2} \ln\left(\frac{\alpha}{2}\right)} \cdot \sqrt{\frac{2}{d}}$}
	\State Output \textsf{Abort}. 
	\EndIf
	\EndFor
	
\end{algorithmic}
\end{algorithm}

\begin{algorithm*}
\caption{Our Construction for Sparsification-Enabled, Maliciously Secure, and Differentially Private FL} 
\label{alg:mali-FL}
\begin{multicols}{2}  
	\begin{algorithmic}[1]
		\vspace*{-20pt}
		\Require Each client $\mathcal{C}_i$ holds a local dataset $\mathcal{D}_i$ for $i\in[n]$.
		\Ensure The servers $\mathcal{S}_{\{0,1,2\}}$ and clients $\mathcal{C}_i$ ($i \in [n]$) obtain a global model $\mathbf{w}$ that satisfies $(\varepsilon,\delta)$-DP. 
		
		\Procedure{Mal-Train}{$\{\mathcal{D}_i\}_{i\in[n]}$}
		\State Initialize $ {\mathbf{w}}^0$. 
		\For {$t\in [T]$}\
		\State \multiline{Sample a set $\mathcal{P}^t$ of clients.} 
		\For{$\mathcal{C}_i \in \mathcal{P}^t$ \textbf{in parallel}}
		\State // \emph{\underline{Client-side computation:}}
		\State $\Delta^{t+1}_i \leftarrow \mathsf{LocalUpdate}(\mathcal{D}_i, \mathbf{w}^t)$. 
		\State $\mathcal{I}_i \leftarrow \text{argtop}_k \big(\mathsf{abs}(\Delta^{t+1}_i)\big).$ \label{alg:semi-FL:argtop}
		
		\State $\boldsymbol{x}^{t+1}_i \leftarrow \mathbf{0}^d$. 
		\State ${\boldsymbol{x}}^{t+1}_i[j]\leftarrow\Delta^{t+1}_i[j]$ for $j\in\mathcal{I}_i$. \Comment\emph{{Sparsification.}} \label{alg:semi-FL:topk-end}
		\State ${\boldsymbol{x}}^{t+1}_i \leftarrow {\boldsymbol{x}}^{t+1}_i / \max \{1, \|{\boldsymbol{x}}^{t+1}_i\|_2/C\}$. \Comment\emph{{Clipping.}} \label{alg:semi-FL:clipping}
		\State Initialize empty lists $L_i, E_i, P_i$. 
		\For{$j\in[d]$}
		\State $L_i.\mathsf{append}(j)$ if $\boldsymbol{x}^{t+1}_i[j] \neq 0$ else $E_i.\mathsf{append}(j)$.
		\EndFor
		\State $\boldsymbol{x}_i^{\prime} \stackrel{Eq. \eqref{eq:re-order}}{\longleftarrow} \boldsymbol{x}^{t+1}_i,L_i$.
		\State $\pi_i \stackrel{Eq. \eqref{eq:comp_pi}}{\longleftarrow} L_i,E_i$. 
		\State Locally sample seeds ${\mathsf{s}_{i,0},\mathsf{s}_{i,1}}$. 
		\State \multiline{$\pi_{i,0} \leftarrow G(\mathsf{s}_{i,0})$, $\pi_{i,1} \leftarrow G(\mathsf{s}_{i,1})$.}
		\State $\pi_{i,2} \leftarrow \pi_{i,1}^{-1} \circ \pi_{i,0}^{-1} \circ \pi_i$.
		\State $P_{i}.\mathsf{append}\left(\pi_{i,2}(j)\right)$ for $j \in [k]$.
		\State $\boldsymbol{r}^{}_i \leftarrow (\boldsymbol{x}^{\prime}_i[0], \boldsymbol{x}^{\prime}_i[1], \dots, \boldsymbol{x}^{\prime}_i[k-1])$
		\State Sample MAC key seeds $\mathsf{ks}_{i,0}, \mathsf{ks}_{i,1}, \mathsf{ks}_{i,2}$. \label{code:malicious-shuffle-start}
		\State $\boldsymbol{k}^{}_i \gets G(\mathsf{ks}_{i,0}) + G(\mathsf{ks}_{i,1}) + G(\mathsf{ks}_{i,2})$. 
		\State $t_i \gets \sum_{j=0}^{d-1} \boldsymbol{k}^{}_i[j] \cdot \boldsymbol{x}^{\prime}_i[j]$. \Comment \emph{Compute MAC.}
		\State \multiline{Distribute $\mathsf{s}_{i,0}$ and $\mathsf{s}_{i,1}$ to $\mathcal{S}_0$, $\mathsf{s}_{i,1}$ and $P_{i}$ to $\mathcal{S}_1$, $\mathsf{s}_{i,0}$ and $P_{i}$ to $\mathcal{S}_2$. Distribute the shares $(\langle t_i \rangle_j, \langle \boldsymbol{r}_i \rangle_j, \mathsf{ks}_{i,j})$ and $(\langle t_{i} \rangle_{j+1}, \langle \boldsymbol{r}_{i} \rangle_{j+1}, \mathsf{ks}_{{i},j+1})$ to server $\mathcal{S}_j$ ($j\in\{0,1,2\}$).}

		\State // \emph{\underline{Server-side computation:}}
		\State $\mathcal{S}_{\{0,2\}}$ locally compute $\pi^{\prime}_{i,0} \leftarrow G(\mathsf{s}_{i,0})$.
		\State $\mathcal{S}_{\{0,1\}}$ locally compute $\pi^{\prime}_{i,1} \leftarrow G(\mathsf{s}_{i,1})$. 
		\State $\mathcal{S}_{\{1,2\}}$ locally initialize an empty list $R_i$. 
		\For{$j\in[d]$}
		\State $R_i.\mathsf{append}(j)$ if {$j \notin P_i$}. 
		\EndFor
		\State $\mathcal{S}_{\{1,2\}}$ locally compute $\pi_{i,2}^{\prime}\stackrel{Eq. \eqref{eq:comp_pi2_server}}{\longleftarrow}P_i,R_i$.
		\State \multiline{$\mathcal{S}_{\{0,1,2\}}$ locally compute $\llbracket \boldsymbol{x}_i^{\prime} \rrbracket \leftarrow \llbracket \boldsymbol{r}_i \rrbracket \| \llbracket 0,\cdots,0 \rrbracket$. }
		\For{$j \in \{0,1,2\}$}
		\State \multiline{$\mathcal{S}_j$ computes $\langle \boldsymbol{k}_i \rangle_{j} \leftarrow G(\mathsf{ks}_{i,j}), \langle \boldsymbol{k}_i \rangle_{j+1} \leftarrow G(\mathsf{ks}_{i,j+1})$.} 
		\EndFor

		\For {$j \in \{0,1,2\}$} \Comment \emph{Secret-shared shuffle.} 
		\State \multiline{$\mathcal{S}_j,\mathcal{S}_{j+1}$ apply $\pi^{\prime}_{i,j+1}$ over the shares of $\llbracket \boldsymbol{x}^{\prime}_i \rrbracket,\llbracket \boldsymbol{k}_i \rrbracket$ and re-share the results to $\mathcal{S}_{j-1}$. }
		\EndFor 
		
		\State $\llbracket f_i \rrbracket \gets \llbracket 0 \rrbracket$. \Comment \emph{Initialize verification tag.} 
		\State $\llbracket t'_i \rrbracket \gets \sum_{j=0}^{d-1} \llbracket \pi_i(\boldsymbol{x}^{\prime}_i)[j] \rrbracket \cdot \llbracket \pi_i(\boldsymbol{k}_i)[j] \rrbracket$.
		\State $\llbracket f_i \rrbracket \gets \llbracket f_i \rrbracket + (\llbracket t_i \rrbracket - \llbracket t'_i \rrbracket)$. 
		
		\State $\mathcal{S}_{\{0,1,2\}}$ sample a secret-shared random number $\llbracket r \rrbracket$. 
		\State \multiline{$\mathcal{S}_{\{0,1,2\}}$ computes $f\leftarrow\mathsf{Rec}\left(\llbracket r_i \rrbracket \cdot (\sum_{\mathcal{C}_i \in \mathcal{P}^t}\llbracket f_i \rrbracket)\right)$ and output abort if $f\neq 0$.}  \label{code:malicious-shuffle-end}
		\State $\{\llbracket \eta_i \rrbracket\}_{i\in\{0,1,2\}} \leftarrow \mathsf{VeriNoiseSamp}(\sigma, C,\alpha)$. 
		\State $\llbracket {\Delta}^{t+1} \rrbracket \leftarrow \sum_{\mathcal{C}_i\in\mathcal{P}^t}\llbracket\boldsymbol{x}_i\rrbracket + \sum_{i=0}^{2}\llbracket \eta_i \rrbracket$. 
		\State \multiline{$ \Delta^{t+1} \leftarrow \mathsf{Rec}(\llbracket \Delta^{t+1} \rrbracket)$. } 
		\State $\mathbf{w}^{t+1} \leftarrow \mathbf{w}^{t} + \Delta^{t+1}/|\mathcal{P}^t|$. 
		\State \multiline{$\mathcal{S}_{\{0,1,2\}}$ locally hash $\mathbf{w}_{t+1}$, and then exchange the hashes. $\mathcal{S}_{\{0,1,2\}}$ output abort if the hashes do not match, else broadcast $\mathbf{w}_{t+1}$ to the clients.}
		\EndFor
		\EndFor
		\EndProcedure
		
	\end{algorithmic}
\end{multicols}
\end{algorithm*}

\section{Security Analysis}
\label{sec:security-analysis}
In this section, we analyze the security of {\main} using the simulation-based paradigm \cite{Lindell17} under the threat model defined in Section~\ref{sec:threat_model}. 
Since {\main} operates in an iterative training process, it suffices to analyze the security of a single round \cite{ELSA}. 
The security of the entire training process then follows naturally. 
In each FL round, the workflow consists of two phases: \textit{client-side computation} (where clients generate local sparsified gradient updates, MAC keys, and MAC values) and \textit{server-side computation} (where servers securely aggregate and process gradient updates). 
Since clients operate independently without inter-client communication, we focus exclusively on the server-side computation for security analysis.
The server-side computation involves three components: (1) secure permutation decompression and secret-shared shuffle, (2) verifiable noise sampling, and (3) secure aggregation of gradient updates with DP noise. 
We consider a probabilistic polynomial-time (PPT) adversary $\mathcal{A}$ that statically corrupts at most one of the three servers. 
The corrupted server may deviate arbitrarily from the protocol, while the remaining servers remain honest. 
Below we first define the ideal functionality of the multi-party differentially private FL. 

\begin{defn}
\label{def:overall_functionality}
\textit{\textbf{(Multi-party differentially private FL ideal functionality $\mathcal{F}_{\textsf{sc}}$).}} 
\textit{The functionality of $\mathcal{F}_{\textsf{sc}}$ interacts with three servers; at most, one server is controlled by the malicious adversary, and the rest servers are honest. 
	In the $t$-th FL round: \\
	\textit{\textbf{Input.}}
	For a selected subset $\mathcal{P}^t$, $\mathcal{F}_{\textsf{sc}}$ receives from each $\mathcal{C}_i\in \mathcal{P}^t$: (1) The $k$ non-zero values $\boldsymbol{r}_i = (\boldsymbol{x}_i'[0], \dots, \boldsymbol{x}_i'[k-1])$ derived from the clipped, sparsified update $\boldsymbol{x}_i^{t+1}$, (2) Seeds $\mathsf{s}_{i,0}, \mathsf{s}_{i,1}$ representing compressed permutations $\pi_{i,0}, \pi_{i,1}$, (3) The distilled permutation list $P_i$ derived from $\pi_{i,2}$, (4) The MAC value $t_i$ computed on $\boldsymbol{x}'_i$ with key $\boldsymbol{k}_i$, and (5) Key seeds $\mathsf{ks}_{i,0}, \mathsf{ks}_{i,1}, \mathsf{ks}_{i,2}$ used to generate $\boldsymbol{k}_i$. \\
	\textit{\textbf{Computation.}} 
	For each selected client $\mathcal{C}_i \in \mathcal{P}^t$: Reconstruct the intended permutation $\pi_i$ from seeds/distillation. 
	Reconstruct key $\boldsymbol{k}_i$ from key seeds. 
	Pad $\boldsymbol{r}_i$ to get $\boldsymbol{x}'_i$. 
	Verify client MAC: Check if $t_i = \sum_{l=0}^{d-1} \boldsymbol{k}_i[l] \cdot \boldsymbol{x}'_i[l]$. 
	If verification fails for any client, record failure and proceed (or optionally abort). 
	Compute the intended sparse update $\boldsymbol{x}_i = \pi_i(\boldsymbol{x}'_i)$.
	Compute $\mathcal{G} = \sum_{i \in \mathcal{P}^t}$ and add sampled noise ${\eta}$. 
	Compute noisy aggregate $\mathcal{G}' = \mathcal{G} + {\eta}$.
	Interact with Adversary: Send $\mathcal{G}'$ to $\mathcal{A}$. 
	Receive continue/abort from $\mathcal{A}$.
	\\
	\textit{\textbf{Output.}} $\mathcal{F}_{\textsf{sc}}$ outputs the noisy aggregated gradient update $\mathcal{G}^\prime$ if the adversary sends \textsf{continue}.
}
\end{defn}

We now formally prove the security of {\main} against malicious adversary. 
Since the server-side computation involves secret-shared multiplication, we first assume that the underlying multiplication protocol achieves an ideal functionality $\mathcal{F}_{\textsf{mult}}$ \cite{RSS}. 
This functionality accepts secret-shared inputs $\llbracket x \rrbracket$ and $\llbracket y \rrbracket$ from the servers and outputs their product in secret-shared form $\llbracket x \cdot y \rrbracket$. 
The interaction with $\mathcal{F}_{\textsf{mult}}$ can be simulated by a polynomial-time simulator $\mathcal{S}_{\textsf{mult}}$ that generates indistinguishable transcripts for honest parties.

\begin{thm}
\label{thm:mali-topkagg}
\textit{\textbf{(Maliciously secure differentially private FL)}}. 
\textit{
	Assuming that PRG is a random oracle and that the secret-shared multiplication achieves $\mathcal{F}_{\textsf{mult}}$, the protocol of {\main} securely realizes the functionality (Definition. \ref{def:overall_functionality}) in the presence of a static, non-colluding malicious adversary corrupting at most one server, with security with abort. 
}
\end{thm}

\begin{proof}		
The core idea is to construct simulators $\mathcal{S}_0$, $\mathcal{S}_1$, and $\mathcal{S}_2$ that emulate the corrupted server's view while ensuring indistinguishability from the ideal functionality $\mathcal{F}_{\textsf{sc}}$ (Definition~\ref{def:overall_functionality}). 
Due to the symmetry in the protocol design for shuffle, MAC verification, and final reconstruction/hash check, the simulation strategies are similar. We sketch the simulation for the case where $\mathcal{S}_0$ is corrupted by adversary $\mathcal{A}$; simulations for $\mathcal{S}_1, \mathcal{S}_2$ follow analogously, accounting for specific roles in permutation decompression and noise verification phases.

\noindent\textbf{Simulating $\mathcal{S}_0$.}
We decompose the adversary's view into five phases: \textit{Pre-shuffling} (decompressing permutations and padding), \textit{Shuffling} (applying permutations), \textit{Post-shuffling} (blind MAC verification), \textit{Noise-sampling} (sampling and verification of DP noise), and \textit{Aggregation} (secure summation). 
Each phase is simulated as follows:

\textit{Pre-shuffling.} For each client $\mathcal{C}_i$'s sparsified gradient update $\boldsymbol{x}_i^\prime$, $\mathcal{S}_0$ decompresses the permutations $\pi_{i,0}$ and $\pi_{i,1}$ using seeds $\mathsf{s}_{i,0}, \mathsf{s}_{i,1}$ received from $\mathcal{C}_i$. 
The padding of $d-k$ zeros to $\llbracket \boldsymbol{r}_i \rrbracket$ (where $\boldsymbol{r}_i$ contains the top-$k$ values) is done non-interactively \cite{RSS}. 
Therefore, the pre-shuffling process is trivial to simulate since $\mathcal{S}_0$ receives nothing. 
If the adversary provides malformed padding values or invalid permutations during the protocol execution, the probability of such tampered inputs passing the post-shuffle MAC verification is negligible by the DeMillo-Lipton-Schwartz-Zippel lemma \cite{Zippel79,Schwartz80,DemilloL78}.

\textit{Shuffling.} During secret-shared shuffle, $\mathcal{S}_0$ applies its permutation and re-shares the result with $\mathcal{S}_1$ and $\mathcal{S}_2$. 
The re-sharing operation involves generating a random mask sampled from $\mathbb{Z}^d_p$ and sending the masked share to others. 
Also, $\mathcal{S}_0$ receives masked shuffled results from others, which is uniformly random in $\mathcal{S}_0$'s view.

\textit{Post-shuffling.} After shuffling, the simulator verifies the integrity of reconstructed $\boldsymbol{x}_i$ (in secret-shared form) via blind MAC. 
The verification computes $\llbracket f \rrbracket = \sum_{\mathcal{C}_i \in \mathcal{P}^t} (\llbracket t_i \rrbracket - \llbracket t_i^\prime \rrbracket)$, where $t_i$ is the client-generated MAC and $t_i^\prime$ is the server-recomputed MAC. 
The simulator $\mathcal{S}_{\textsf{mult}}$ emulates RSS multiplications \cite{RSS} for computing $t_i^\prime$ during MAC verification. 
If the adversary performs the blind MAC verification honestly, the simulator uses the negation of the adversary's share to enforce a zero flag. 
If the adversary deviates from any part of the protocol, the simulation uses a random element of $\mathbb{Z}_p$. 
In any case where the adversary deviates from the protocol, the simulation sends abort to the ideal functionality.
By the DeMillo-Lipton-Schwartz-Zippel lemma \cite{Zippel79,Schwartz80,DemilloL78}, the probability of a forged MAC passing verification is negligible.

\textit{Noise-sampling.} This involves sampling and verification for ${\eta}_0$, ${\eta}_1$, ${\eta}_2$. We simulate $\mathcal{S}_0$'s view. 
(1) Verifying ${\eta}_0$ ($\mathcal{S}_0$ is source): the simulator simulates $\mathcal{S}_0$ sending shares of ${\eta}_0$. 
Since the mask ${\xi}_1$ is generated and secret-shared from $\mathcal{S}_1$, its shares are uniformly random in $\mathcal{S}_0$'s view. 
The summation of $\llbracket{\xi}_1\rrbracket$ and $\llbracket{\eta}_0\rrbracket$ is conducted locally without interaction, therefore it is trivial to simulate. 
Then the simulator simulates the (robust) reconstruction of ${\kappa}_0 = {\eta}_0 + {\xi}_1$. 
Since $\mathcal{S}_0$ receives nothing in the reconstruction phase, it is trivial to simulate. 
Since the KS test is performed on $\mathcal{S}_2$ locally and $\mathcal{S}_0$ receives nothing, the KS test phase for ${\eta}_0$ is trivial to simulate. 
(2) Verifying ${\eta}_1$ ($\mathcal{S}_0$ is masker): Similarly, this phase is trivial to simulate as $\mathcal{S}_0$ receives nothing. 
(3) Verifying ${\eta}_2$ ($\mathcal{S}_0$ is tester): $\mathcal{S}_0$ receives random shares for ${\eta}_2$ (from $\mathcal{S}_2$) and mask ${\xi}_1$ (from $\mathcal{S}_1$), which is trivial to simulate as the shares are uniformly randomly sampled. 
The reconstruction of ${\kappa}_2 = {\eta}_2 + {\xi}_1$ on $\mathcal{S}_0$ is trivial to simulate as the shares of $\llbracket{\kappa}_2\rrbracket$ are uniformly random in $\mathcal{S}_0$'s view. 
If the shares from $\mathcal{S}_{\{1,2\}}$ do not match, the simulation sends $\textsf{abort}$ to the ideal functionality. 
Regarding the locally performed KS test on $\mathcal{S}_0$, the test should pass w.h.p. the simulator expects $\mathcal{A}$ to report $\textsf{continue}$. 
If $\mathcal{A}$ reports \textsf{Reject}, the simulator sends $\textsf{abort}$ to $\mathcal{F}_{\textsf{sc}}$.

\textit{Aggregation.} 
The simulator knows the final value $\mathcal{G}'$ from $\mathcal{F}_{\textsf{sc}}$. 
It provides shares consistent with $\mathcal{G}'$ to $\mathcal{A}$. 
It receives shares from $\mathcal{S}_1, \mathcal{S}_2$ (also consistent with $\mathcal{G}'$). 
The simulator simulates $\mathcal{S}_0$ performing the reconstruction and comparison (trivial to simulate as the shares required for reconstruction are uniformly random in $\mathcal{S}_0$'s view). 
The simulator then simulates $\mathcal{S}_0$ doing the local gradient update $\mathbf{w}^{t+1} = \mathbf{w}^t + \mathcal{G}'/|\mathcal{P}^t|$ and hash computation. 
Then it simulates receiving the correct hashes from $\mathcal{S}_1, \mathcal{S}_2$ and checks if $\mathcal{S}_0$ reports the same hash. 
If the hashes match, the simulator sends $\textsf{continue}$ to $\mathcal{F}_{\textsf{sc}}$, else it sends $\textsf{abort}$ to $\mathcal{F}_{\textsf{sc}}$.

\end{proof}
\end{document}